%% file: main.tex
\documentclass[journal]{IEEEtran}

\usepackage[colorlinks,urlcolor=blue,linkcolor=blue,citecolor=blue]{hyperref}

\usepackage[dvipsnames]{xcolor}
\usepackage{amsmath, amssymb}
\usepackage{amsfonts}
\usepackage{microtype}
\usepackage{graphicx}
\usepackage{mathtools}
\usepackage{subfigure}
\usepackage{booktabs} 
\usepackage{transparent}
\usepackage{tablefootnote}
\usepackage{optidef} 
\usepackage{comment}
\usepackage[ruled,vlined,linesnumbered]{algorithm2e}
\usepackage{lscape} 
\usepackage{caption} 
\usepackage{tcolorbox}
\usepackage{enumitem}

\usepackage{stackengine}
\usepackage{color,array}
\usepackage{multicol}
\usepackage{graphicx}
\usepackage{./makecell}

\newtcolorbox{algbox}[2][]{%
  colframe=gray!15, 
  colback=gray!05, 
  coltitle=black, 
  fonttitle=\bfseries,
  colupper=black, 
  title=#2, 
  arc=1mm,
  boxrule=0.0mm, 
  left=0pt,
  right=0pt,
  top=0pt,
  bottom=3pt,
  #1 
}

\newtheorem{definition}{Definition}[section] 
\newtheorem{assumption}{Assumption}[section]
\newtheorem{theorem}{Theorem}[section]
\newtheorem{remark}{Remark}[section]

\newtheorem{lemma}{Lemma}[section]

\newtheorem{proof}{Proof}[section]

\usepackage[normalem]{ulem}

\usepackage{pifont}
\newcommand{\cmark}{\color{green!50!black}\ding{51}}%
\newcommand{\xmark}{\color{red!50!black}\ding{53}}

\usepackage{hyperref}

\usepackage[capitalize]{cleveref}

\input{wordings.tex}

\input{math_symbols.tex}

\newif\ifdraft
\drafttrue 

\renewcommand{\arraystretch}{1.3}
\newcommand{\ra}[1]{\renewcommand{\arraystretch}{#1}} 

\newcommand{\appref}[1]{App.~\ref{#1}}
\newcommand{\secref}[1]{Sect.~\ref{#1}}
\newcommand{\figref}[1]{Fig.~\ref{#1}}

\begin{document}

\title{Synthesis of Model Predictive Control and Reinforcement Learning: Survey and Classification}

\author{Rudolf Reiter\textsuperscript{$\star$}, Jasper Hoffmann\textsuperscript{$\star$}, Dirk Reinhardt, Florian Messerer, Katrin Baumgärtner, Shambhuraj Sawant, Joschka Boedecker, Moritz Diehl, Sebastien Gros
\thanks{This research was supported by DFG via Research Unit FOR 2401 and project 424107692, by the EU via ELO-X 953348, and by NFR through the project Safe Reinforcement Learning using Model Predictive Control (SARLEM, grant number 300172).}
\thanks{Rudolf Reiter, Florian Messerer, Katrin Baumgärtner and Moritz Diehl are with the Department of Microsystems Engineering (IMTEK), University of Freiburg, 79110 Freiburg, Germany (e-mail: \{rudolf.reiter, florian.messerer, katrin.baumgaertner, moritz.diehl\}@imtek.uni-freiburg.de).}
\thanks{Jasper Hoffmann and Joschka Boedecker are with the Department of Computer Science, University of Freiburg, 79110 Freiburg, Germany (e-mail: \{hofmaja, jboedeck\}@informatik.uni-freiburg.de).}
\thanks{Dirk Reinhardt, Shambhuraj Sawant and Sebastien Gros are with the Department of Engineering Cybernetics, Norwegian University of Science and Technology (NTNU), 7034 Trondheim, Norway (e-mail: \{dirk.p.reinhardt, shambhuraj.sawant, sebastien.gros\}@ntnu.no).}}


\maketitle
\begingroup\renewcommand\thefootnote{$\star$}
\footnotetext{Equal contribution}
\endgroup

\begin{abstract}
The fields of \ac{MPC} and \ac{RL} consider two successful control techniques for Markov decision processes. Both approaches are derived from similar fundamental principles, and both are widely used in practical applications, including robotics, process control, energy systems, and autonomous driving.

Despite their similarities, \ac{MPC} and \ac{RL} follow distinct paradigms that emerged from diverse communities and different requirements. Various technical discrepancies, particularly the role of an environment model as part of the algorithm, lead to methodologies with nearly complementary advantages. 
Due to their orthogonal benefits, research interest in combination methods has recently increased significantly, leading to a large and growing set of complex ideas leveraging \ac{MPC} and \ac{RL}.

This work illuminates the differences, similarities, and fundamentals that allow for different combination algorithms and categorizes existing work accordingly.
Particularly, we focus on the versatile actor-critic \ac{RL} approach as a basis for our categorization and examine how the online optimization approach of \ac{MPC} can be used to improve the overall closed-loop performance of a policy.
\end{abstract}


\acresetall
\section{Introduction}
\label{sec:introduction}

\IEEEPARstart{S}{olving} \acp{MDP} online is a large and active research domain where different research communities have developed various solution approaches~\cite{sutton_reinforcement_2018,rawlings_model_2017,bertsekas_reinforcement_2019}.
An \ac{MDP} can be stated as the problem of computing the optimal policy of an agent interacting with a stochastic environment that minimizes a cost function, possibly over an infinite horizon.
Two common approaches for obtaining optimal policies are \ac{MPC} and \ac{RL}.

Within \ac{MPC}, an optimization problem that approximates the \ac{MDP} is solved online, involving a simulation of an internal prediction model. The optimized controls are applied to the real-world environment in a closed loop at each time step.
Historically, \ac{MPC} has been developed within the field of optimization-based control engineering, driven by the success of optimization algorithms, e.g., linear programming~\cite{dantzig_simplex_1956}.
\Ac{MPC} design leverages domain knowledge to compose mathematical models, often by first principles. 

In contrast, within the \ac{RL} framework, a policy is expressed as a parameterized explicit function of the state. The policy is iteratively improved by interacting with the environment and adapting its parameters related to the observed cost. 
In general, \ac{RL} algorithms are not required to use models of the environment, but often, models are used during training for offline simulation.

Both \ac{MPC} and \ac{RL} found their way into classical real-world control~\cite{schwenzer_review_2021}. The applications differ depending on the availability of training data. \ac{MPC} is used where measurement data is scarce and expensive, and the environment can be described by optimization-friendly models. On the contrary, \ac{RL} is successfully implemented in settings where lots of training data can be generated~\cite{mnih_human-level_2015, wurman_outracing_2022}.

Besides their sampling efficiency, several further advantages and disadvantages of both methods are nearly orthogonal~\cite{gorges_relations_2017}, i.e., weaknesses of either approach are strengths of the other.
For instance, \ac{RL} struggles with safety issues~\cite{brunke_safe_2022}, whereas \ac{MPC} can guarantee constraint satisfaction related to a particular environment model~\cite{rawlings_model_2017}.
This motivated many authors to combine the advantages and synthesize novel algorithms that build on both approaches.

This paper's contribution is twofold.
We first analyze the properties and orthogonal strengths of \ac{MPC} and \ac{RL}. 
Secondly, we propose a systematic overview of how \ac{MPC} can be combined with \ac{RL} and provide an extensive literature overview guided by the proposed classification. This survey reviews practical and theoretical work and concludes with a section on available open-source software.

\subsection{Related Work}
\label{sec:related}
Several works exist that shed light on various parts of the huge research fields of control systems, machine learning, and optimization.
The author of~\cite{gorges_relations_2017} shows relations between \ac{MPC} and \ac{RL} on a conceptual level and includes a detailed discussion for discrete-time linear time-invariant constrained systems with quadratic costs.
\ac{RL} is contextualized from the perspective of control systems in~\cite{recht_tour_2019} where \ac{MPC} is mentioned briefly as one particular control technique.
The authors of~\cite{brunke_safe_2022} survey methods for safe learning in robotics and also consider the role of MPC.
In particular, the authors show how parameterized \ac{MPC} formulations can be used to guarantee safety during learning and how safety can be integrated into \ac{RL} in general.
The survey of~\cite{bengio_machine_2021} provides an extensive overview of how machine learning is used for combinatorial optimization.
While combinatorial optimization problems have a particular structure related to discrete decision variables, similarities can be observed in how \acp{NN} are integrated into an optimization problem.

The author in~\cite{bertsekas_reinforcement_2019} compares \ac{RL} and \ac{MPC} from the point of approximate dynamic programming and proposes a unified mathematical framework.
In \cite{bertsekas_lessons_2022,bertsekas_newtons_2022}, the RL-based program AlphaZero \cite{silver_mastering_2017,silver_mastering_2017-1} that plays games such as chess, shogi or go, is cast within an MPC framework.
It is argued that the lookahead and rollout-based policy improvement used in online deployment plays a crucial role in enabling the success of the respective algorithms. 

A high-level survey on how general machine learning concepts are used within \ac{MPC} is provided by~\cite{hewing_learning-based_2020}. 
The authors of~\cite{mesbah_fusion_2022} have the same focus on investigating general machine learning used as part of \ac{MPC} and provide technical details for learning models, policy parameters, optimal terminal value functions, and approximations of iterative online optimization.
\ac{RL} is a minor topic of both surveys~\cite{hewing_learning-based_2020} and \cite{mesbah_fusion_2022}. The survey~\cite{mesbah_fusion_2022} is well aligned with our proposed framework. For instance, it considers approximating the terminal value function to be treated in another category as closed-loop learning. In fact, our proposed framework can be seen as a complementary work to~\cite{bertsekas_reinforcement_2019} and~\cite{mesbah_fusion_2022}.

Further works compare MPC and RL for particular applications.
Similarly to~\cite{mesbah_fusion_2022}, the authors in~\cite{norouzi_integrating_2023} compare how \ac{MPC} and machine learning are combined exclusively for automotive applications.
The authors in~\cite{zhang_building_2022} compare \ac{RL} and \ac{MPC} specifically applied to energy management in buildings with a focus on data-driven \ac{MPC} methods.

\subsection{Notation}
\label{sec:notation}
The \ac{MPC} and \ac{RL} literature use different symbols and notations for identical objects, as shown in the comparison~\cite{bertsekas_reinforcement_2019}. The notation used within this paper is given in Tab.~\ref{tab:dictionary}. 
In this survey, mostly \ac{RL} notation is used to unify the language but with few exceptions and the use of control literature synonyms for established language. For instance, we mostly use ``environment'' instead of ``plant'' or ``system'' but may occasionally write ``system'' when it is more common in the context. Note that we exclusively use the term \ac{MDP} that also refers to the conceptually equivalent term discrete-time stochastic optimal control.

Also, within the field of \ac{RL} notation, ambiguities exist. In this survey, we mostly use \ac{RL} in order to refer to model-free \ac{RL} algorithms and include \ac{IL} due to their conceptual overlap. 
\Ac{MPC} that learns a model online is often seen as a variant of model-based \ac{RL}~\cite{wang_benchmarking_2019}. 
Since the proposed synthesis approaches are not limited to model learning, and, moreover, model-based \ac{RL} includes approaches that learn a model without the intention to be used in an optimization problem, we omit a detailed discussion on model-based \ac{RL}

For the concatenation~$z = [x^\top, y^\top]^\top $of column vectors~$x$ and~$z$, we write~$z = (x,y)$.
Given set~$\mathcal{A}$ and set~$\mathcal{B}$, we denote with~$\mathcal{B}^\mathcal{A}$ the set of all functions that map from~$\mathcal{A}$ to~$\mathcal{B}$.
Given a set~$\mathcal{A},$ we denote with~$\mathrm{Dist} (\mathcal{A})$ the space of all probability distributions or probability measures over~$\mathcal{A}$.
Let $p$ be a probablity density of a probability distribution over~$\mathcal{A}$, then the support is defined as $\mathrm{supp}(p) \coloneqq \{a \in \mathcal{A} \mid p(a) > 0\}$.


\subsection{Overview}
\label{sec:overview}
On a high level, this paper is structured into (1) an introductory part, (2) a comparison, and (3) a classification scheme and survey of synthesis approaches. An overview is provided in Fig.~\ref{fig:paper_overview}.

The introductory part introduces in Sect.~\ref{sec:problem_setting} the general problem setting. 
Sect.~\ref{sec:rl} and~\ref{sec:mpc} describe the main concepts behind \ac{RL} and \ac{MPC} and discuss how they aim at solving the general problem of Sect.~\ref{sec:problem_setting}.
Both of these sections are split into a conceptual and an algorithmic part, providing the basis for the remainder of this work. 

The comparison in Sect.~\ref{sec:comparison} highlights practical differences between both approaches and surveys applied comparisons.
Experts in the field of \ac{MPC} and \ac{RL} may skip Sect.~\ref{sec:rl},~\ref{sec:mpc} and potentially the comparison in Sect.~\ref{sec:comparison}.

The remaining paramount sections are devoted to the classification and review of combinations of both approaches. In Sect.~\ref{sec:comb_arch}, essential concepts are introduced to categorize existing combination variants based on the actor-critic framework of \ac{RL}.
Following the proposed categorization, literature that uses \ac{MPC} as an expert for training an \ac{RL} agent is summarized in Sect.~\ref{sec:comb_expert}. The most widespread variants using \ac{MPC} within the policy are outlined in Sect.~\ref{sec:mpc_as_actor}, and variants that use the \ac{MPC} to determine the value function at a particular state are surveyed in Sect.~\ref{sec:comb_critic}. 
An additional Sect.~\ref{sec:theory} focusses on important theoretical results from the field of \ac{MPC} and \ac{RL} and is aligned with the previous categorization. Current open-source software is presented in \secref{sec:software}.

The work is concluded and discussed in~Sect.~\ref{sec:discussion}.

\begin{figure}
	\centering
	\includegraphics[width=\linewidth]{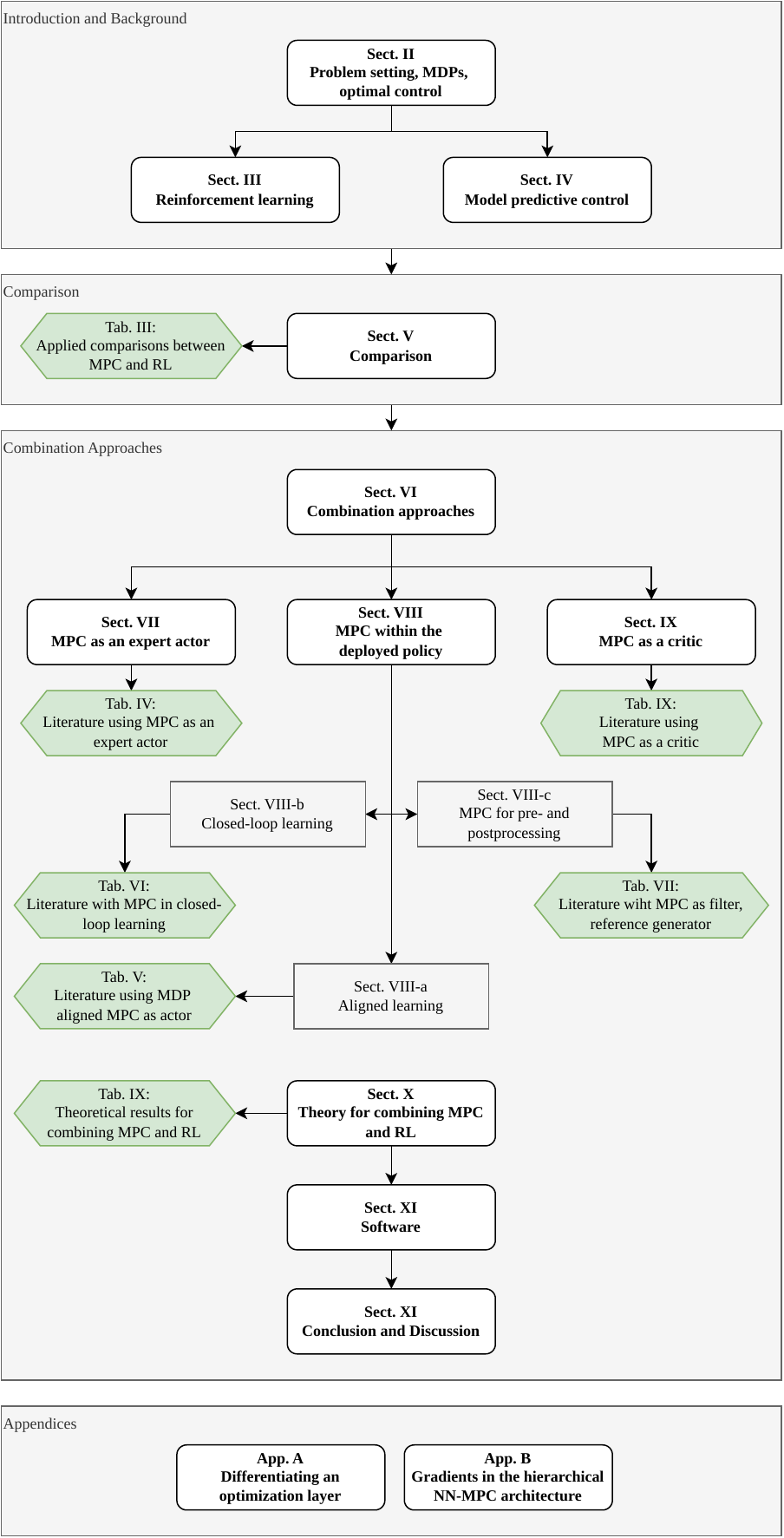}
	\caption{Paper structure. The main sections are highlighted in gray, subchapters in white, tables are green.}
	\label{fig:paper_overview}
\end{figure}


\begin{table}
\caption{Symbols used within this paper.}
	\centering
	\ra{1.3}
	\resizebox{\linewidth}{!}{%
	\begin{tabular}{@{}p{6.7cm}@{}@{}p{0.5cm}@{}@{}l@{}}
		\toprule
		Name && Symbol\\
		\midrule
            state && $s$\\
            sampled state && $S$\\
            planned states within \ac{MPC}&&$x$\\
            action (control) sampled from policy distribution && $a$ \\
            sampled action &&$A$\\
            planned actions within \ac{MPC} &&$u$\\
            stochastic model (part of the environment) &&$P(s^+|s,a)$\\
            deterministic model approximation within MPC && $f^\mathrm{MPC}(x,u)$\\
            stochastic policy (controller)&&$\pi(s)$\\
            deterministic policy &&$\mu(s)$\\
            parameterized policy &&$\pi_\theta (s),\;\mu_\theta (s)$\\
            cost function (part of the environment) && $l(s,a)$\\
            cost function approximation within MPC && $l^\mathrm{MPC}(s,a)$\\
            optimal stochastic policy&&$\pi^\star(s)$\\
            value function under policy $\pi$ && $V^\pi$\\
            value function under policy $\pi$ depending on policy parameters $\theta$ && 
            $J^\pi(\theta)$\\
            optimal value function && $V^\star(s)$\\
            terminal value function approximation under policy~$\pi$ && $\bar{V}^\pi$\\
            value function of MPC && $V^\mathrm{MPC}(x)$\\
            value function of MPC parameterized by~$\theta$&& $V_\theta^\mathrm{MPC}(x)$\\
            terminal value function of MPC && $\bar{V}^\mathrm{MPC}(x)$\\
            terminal value function of MPC parameterized by~$\theta$&& $\bar{V}^\mathrm{MPC}_\theta(x)$\\
            action-value function under policy $\pi$ && $Q^\pi(s,a)$\\
            optimal action-value function && $Q^\star(s,a)$\\
            action-value function of MPC &&$Q^\mathrm{MPC}(s,a)$\\
            action-value function of MPC parameterized by~$\theta$ &&$Q^\mathrm{MPC}_\theta(s,a)$\\
            fixed sampled dataset  &&$\mathcal{D}$\\
            data sampled under policy~$\pi$ &&$\mathcal{D}^{\pi}$\\
            data sampled under $\epsilon$-greedy policy implicitly defined by~$Q$ &&$\mathcal{D}^{\pi^\epsilon_Q}$\\
            replay buffer with transitions stored during training && $\mathcal{D}^\mathrm{buffer}$\\
		\bottomrule
	\end{tabular}}
 	\label{tab:dictionary}
\end{table}

\section{Problem Setting}
\label{sec:problem_setting}
This section describes the \acf{MDP} framework -- a central concept for both RL and MPC. 


\Acp{MDP} provide a framework for modeling a discrete-time stochastic decision process.
An \ac{MDP} is defined over a state space~$\ST$, an action space~$\AC$, and a stochastic transition model 
\begin{equation}
\label{eq:general_model}
    P: \ST \times \AC \rightarrow \mathrm{Dist}(\ST)
\end{equation}
describing a probability distribution over the next states given a current state and action.
A stage cost~$l(s, a)$ with~$l:\ST\times\AC \rightarrow\R\cup\{\infty\}$ defines the cost of each state-action pair, typically discounted by a factor~$\gamma \in (0, 1]$.
The \ac{MDP}~\cite{puterman_markov_2005} is then defined by the 5-tuple 
\begin{equation}
\label{eq:mdp}
    \mathcal{M}:=\big(\ST, \AC, P, l,\gamma\big).
\end{equation}
Solving the \ac{MDP} refers to obtaining actions that minimize the discounted stage cost, which is elaborated in the following.

For solving \acp{MDP} we introduce stochastic policies~$\pi: \ST \rightarrow \mathrm{Dist}(\AC)$ that map a state to a probability distribution over possible actions.
The value function~$V^\pi:\ST\rightarrow\R\cup\{\infty\}$ is the discounted expected future cost starting from state~$s$ and following a policy~$\pi$ defined by 
\begin{gather}
	\label{eq:general_expected_return}
	V^\pi (s) \coloneqq \E \Bigg[ \sum_{k=0}^{\infty} \gamma^k l(S_k, A_k)  \Bigg] \\
         A_k \sim \pi(\cdot \mid S_k),\; S_{k+1} \sim P(\cdot \mid S_k, A_k), \; S_0=s. \nonumber
\end{gather}
The state~$S_k$ and action~$A_k$ describe a sequence of random variables generated by applying the policy~$\pi$ on the \ac{MDP}.

The aim of solving the \ac{MDP} is to obtain an optimal policy~$\pi^\star$ that minimizes the expected cost for all states via
\begin{equation}
\label{eq:general_problem}
    \pi^\star \in \bigcap_{s \in \ST} \arg \min_{\pi} V^\pi (s).
\end{equation}
Solving this equation is often also shortly referred to as solving the~\ac{MDP}.
Note that the intersection in \eqref{eq:general_problem} is never empty, and there also always exists an optimal deterministic policy~$\mu^\star$~\cite{sutton_reinforcement_2018} satisfying~\eqref{eq:general_problem}.
All optimal policies~$\pi^\star$ share the same optimal value function~$V^\star \coloneqq V^{\pi^\star}$. In addition to the value function in \eqref{eq:general_expected_return}, the action-value function~$Q^\pi:\ST\times\AC \rightarrow\R\cup\{\infty\}$ of a policy
\begin{gather}
\label{eq:general_action_expected_return}
    Q^\pi(s, a) := \E \Bigg[ \sum_{k=0}^{\infty} \gamma^k l(S_k, A_k) \Bigg] \\
    A_k \sim \pi(\cdot \mid S_k),\; S_{k+1} \sim P(\cdot \mid S_k, A_k), \; S_0 = s, \; A_0 = a. \nonumber
  \end{gather}
is the expected value of first applying an action~$a$ at the current state~$s$ and following the stochastic policy~$\pi$ afterwards.

We define as~$V^\star \coloneqq V^{\pi^\star}$ and~$Q^\star \coloneqq Q^{\pi^\star}$ the optimal value and action-value functions.
The identities 
\begin{subequations}
\begin{align}\label{eq:derive_optimal_value_function}
 	V^\star(s) &= \min_{a} Q^\star(s, a)\\
   	\mathrm{supp} \big(\pi^\star(\cdot \mid s)\big) &\subseteq \arg \min_{a^\star} Q^\star(s, a^\star).
\end{align}
\end{subequations}
are relating the optimal value function and the optimal policy to the optimal action-value function.
In cases with multiple optimal actions for a state~$s$, an optimal policy~$\pi^\star$ can be stochastic and randomly selects an action in the set of optimal actions with any probability.
\begin{remark}
The optimal policy for a finite horizon problem is generally time-varying unless the terminal cost is $V^\star$.
Time-varying environments can be described in terms of~\eqref{eq:general_model} by augmenting the state with an additional clock state.
\end{remark}

In the following, we introduce \ac{MPC} and \ac{RL}, the two pivotal frameworks of this survey for approximately solving \acp{MDP}~\eqref{eq:general_problem}.

\section{Reinforcement Learning}
\label{sec:rl}
\ac{RL} is a powerful approach for solving \acp{MDP} using concepts from dynamic programming, Monte Carlo simulation, and stochastic approximation.
\ac{RL} typically involves an agent modeled by a policy iteratively collecting data by interacting with an environment, which could be a simulation model or the real world.
The collected data, consisting of state transitions, applied actions, and costs, is then used to iteratively update the policy.
In most \ac{RL} methods, an optimal policy is approximated via estimating the optimal action-value function using \ac{TD} learning or by iteratively updating the policy using \ac{PG} methods \cite{sutton_reinforcement_2018}, whereas in actor-critic methods both forms are combined.

\subsection{Theoretical Background}
In this section, a short theoretical overview of \ac{RL} is provided, including dynamic programming, temporal difference methods, and policy gradient methods.\\
\subsubsection{Dynamic Programming}
Introduced in \cite{bellman_dynamic_1966}, \ac{DP} provides the theoretical foundation for many algorithms in \ac{RL}.
\Ac{DP} solves \acp{MDP} with known transition models by breaking them into subproblems and using stored solutions to avoid recomputation. \Ac{DP} systematically updates value functions given complete knowledge of the environments \ac{MDP}. 
A general \ac{DP} approach to find the action-value function~$Q^\pi$ of a given policy~$\pi$ can be described by the Bellman operator~$T^\pi: \R^\STAC \rightarrow \R^\STAC$, 
\begin{equation}\label{eq:bellman_operator}
     T^\pi Q(s, a) \coloneqq l(s, a) + \gamma \E_{\substack{S^+ \sim P(\cdot \mid s, a), \\ A^+ \sim \pi (\cdot \mid S^+) }} [Q(S^+, A^+)]\,,
\end{equation}
where~$S^+$ is the next sampled state and~$A^+$ a sampled action.
In the space of value functions, the Bellman operator~$T^\pi$ is a contraction mapping for~$\gamma < 1$ with respect to the state supremum norm \cite{bertsekas_neuro-dynamic_1996}, and~$Q^\pi$ is the unique fixed point.
In other words, one can start with an arbitrary action-value function~$Q \in \R^\STAC$ and iteratively apply~$T^\pi$ to converge to the action-value function~$Q^\pi$.
Determining~$V^\pi$ and~$Q^\pi$ for a fixed policy~$\pi$ is referred to as policy evaluation.

Similar to the Bellman operator~$T^\pi$, the Bellman optimality operator~$T: \R^\STAC \rightarrow \R^\STAC$ is defined as
\begin{equation}\label{eq:bellman_optimality_operator}
     T Q(s, a) \coloneqq l(s, a) + \gamma \min_{a^+} \E_{S^+ \sim P(\cdot \mid s,a)} [Q(S^+, a^+)].
\end{equation}
Equivalently to~$T^\pi$, $T$ is a contraction mapping for $\gamma < 1$.
Iteratively applying~$T$ on an arbitrary action-value function~$Q \in \R^\STAC$ converges to the optimal value function~$Q^\star$.
The resulting method is called \ac{VI} and, differently to policy evaluation, tries to solve \acp{MDP} of \eqref{eq:general_problem} implicitly by finding the optimal value function~$Q^\star$.

The main drawback of \ac{DP} methods is the unfavorable scaling to high-dimensional or continuous state and action spaces, which is often referred to as the curse of dimensionality~\cite{bellman_dynamic_1966}.
\Ac{DP} methods require full knowledge of the environment model~$P$, which is a fundamental limitation compared to more generic model-free \ac{RL} algorithms.
\Ac{ADP}~\cite{powell_approximate_2007} addresses the curse of dimensionality by using different approximation strategies to extend classical \ac{DP}, as discussed in the following in the context of \ac{RL}.\\

\subsubsection{\Ac{TD} Methods}
To avoid the scaling issues of \ac{DP}, \acf{TD} methods~\cite{sutton_reinforcement_2018} introduce two extensions:
Firstly, they learn the value functions~$V^\pi$ or $Q^\pi$ for a policy~$\pi$ and their optimal versions~$V^\star$ and~$Q^\star$ with only transition samples utilizing stochastic approximation and without explicitly requiring a transition model~$P$.
Secondly, they use an adaptive exploration-exploitation strategy to decide favorable states for which the value function is updated.

In the following, different \ac{TD} learning algorithms are presented.
The simplest policy evaluation method is called $\mathrm{TD}(0)$, which estimates the value function~$V^\pi$ for a given policy~$\pi$.
Given a current value function~$V$, an update for a given state~$S$ for~$V$ is defined by
\begin{algbox}{TD(0) $(V \approx V^\pi)$}
\begin{subequations}\label{eq:td0}
\begin{align}
    &V(S) \leftarrow V(S) + \alpha \delta, \\
    &\delta \coloneqq l(S, A) + \gamma V(S^+) - V(S), \\
    &\textrm{with }S \sim \mathcal{D}^\pi,\; A \sim \pi( \cdot \mid S),\; S^+ \sim P(\cdot \mid S, A), 
\end{align}
\end{subequations}
\end{algbox}

%

where $\leftarrow$ denotes overwriting the function~$V$ at $S$, $\alpha > 0$ denotes the learning rate, $A \sim \pi(\cdot|S)$ a sampled action from the policy~$\pi$ at state~$S$ and~$S^+ \sim P(\cdot|S, A)$ the next state sampled from the stochastic model~\eqref{eq:general_model}.
Furthermore, $\delta$ is called the temporal difference, measuring the stochastic difference between the sampled target~$l(S, A) + \gamma V(S^+)$ and the current estimate~$V(S)$.
Finally, the distribution~$\mathcal{D}^\pi$ is a sampling or exploration strategy where the fixed policy~$\pi$ sequentially generates new states by interacting with the environment. For a more technical discussion, see Remark~\ref{remark:state_distribution}.
\begin{remark}\label{remark:state_distribution}
    The notations $S \sim \mathcal{D}^\pi$ or $(S, A) \sim \mathcal{D}^\pi$ are mathematically not well defined as distributions.
    Without further elaboration, given an initial state~$s$, an episode is simply run iteratively by applying a policy~$\pi$:
    \begin{equation}
        S_0=s, A_0 \sim \pi(\cdot | S_0), S_1 \sim P(\cdot | S_0, A_0), A_1 \sim \pi(\cdot | S_1), \dotsc \;.
    \end{equation}
    After each time step, an update can then be performed by a \ac{TD} method.
    We still use this notation for instructional purposes, especially to highlight which policy was used to generate states and actions.
\end{remark}

Similar to \ac{DP}, under some technical assumptions from stochastic approximation theory, the update scheme converges to the unique fixpoint~$V^\pi$ \cite{bertsekas_neuro-dynamic_1996}.
For example, one assumption is that the learning rate~$\alpha$ must decrease to zero over time.
The update scheme of \eqref{eq:td0} can be extended to also learn~$Q^\pi$.

Besides estimating the value function~$V^\pi$ or~$Q^\pi$ for given policies~$\pi$, a main target of \ac{RL} algorithms is learning the optimal value functions~$V^\star$ and~$Q^\star$.
For instance, the SARSA algorithm~\cite{sutton_reinforcement_2018} is an \ac{RL} method for approximating the optimal value function~$Q^\star$ and given by the update rule
\begin{algbox}{SARSA $(Q \approx Q^*)$}
\begin{subequations}\label{eq:sarsa}
\begin{align}
    &Q(S, A) \leftarrow Q(S, A) + \alpha \delta, \\
    &\delta \coloneqq l(S, A) + \gamma Q(S^+, A^+) - Q(S, A), \\
    &\textrm{with } (S, A) \sim \mathcal{D}^{\pi^\epsilon_Q},\; S^+ \sim P(\cdot \mid S, A),\; A^+ \sim \pi^\epsilon_Q (\cdot \mid S^+). \label{eq:td0q_sampling}
\end{align}
\end{subequations}
\end{algbox}
The $\epsilon$-greedy policy is implicitly defined by $Q$ via
\begin{equation}
   \pi^\epsilon_Q(a \mid s) =
\begin{cases} 
1 - \epsilon + \frac{\epsilon}{|\mathcal{A}(s)|}, & \text{if } a = \arg\min_{a^\star} Q(s, a^\star), \\
\frac{\epsilon}{|\mathcal{A}(s)|}, & \text{otherwise},
\end{cases} 
\end{equation}
where for notational convenience, it is assumed that there is only one optimal action for a given state $s$.
Note that $\pi^\epsilon_Q$ is used to sample the next action~$A^+ \sim \pi_Q^\epsilon$ and to generate the state and action with $D^{\pi^\epsilon_Q}$, where the update is performed.

Another variant for learning the optimal action-value function~$Q^\star$ is the popular Q-learning method~\cite{watkins_q-learning_1992}, which, instead of sampling an $\epsilon$-greedy action like SARSA, takes the greedy action to build the temporal difference.
More explicitly, the update rule is defined via 
\begin{algbox}{Q-learning $(Q \approx Q^\star)$}
\begin{subequations}\label{eq:q_learning}
\begin{align}
    &Q(S, A) \leftarrow Q(S, A) + \alpha \delta, \\
    &\delta\coloneqq l(S, A) + \gamma \min_{a^\star} Q(S^+, a^\star) - Q(S, A), \\
    &\textrm{with } (S, A) \sim \mathcal{D}^{\pi_Q^\epsilon},\; S^+ \sim P(\cdot \mid S, A). \label{eq:q_learning_sampling}
\end{align}
\end{subequations}
\end{algbox}
Note that Q-learning still uses the $\epsilon$-greedy exploration policy to sample $(S, A) \sim \mathcal{D}^{\pi^\epsilon_Q}$ in \eqref{eq:q_learning_sampling}, whereas the policy that is used for the temporal difference is the greedy policy.
Thus, Q-learning is called an off-policy method, whereas SARSA is an on-policy method as the same $\epsilon$-greedy policy is used for exploration and the temporal difference update.
For a more elaborate discussion, see \cite{sutton_reinforcement_2018}.

As discussed, an essential motivation for \ac{TD} methods is that learning a value function on the whole state space can be intractable, which is one of the major drawbacks of \ac{DP}.
Importantly, typically only a fraction of the state space is encountered under a given policy or an optimal policy.
Thus, \ac{TD} methods often incorporate a trade-off between exploring the state space and exploiting current knowledge using strategies like the $\epsilon$-greedy policy~$\pi^\epsilon_Q$.
Yet, for discrete state spaces, convergence to $Q^\star$ is only guaranteed if each state and action is seen infinitely often \cite{bertsekas_neuro-dynamic_1996}.\\

\subsubsection{Policy Gradient Methods}
\label{sec:policy_gradient_methods}
Different from the previous methods, \acf{PG} methods directly optimize a parameterized policy~$\pi_\theta$ that can be used for discrete and continuous action spaces~$\AC$.
Given a parameterized policy~$\pi_\theta: \ST \rightarrow \mathrm{Dist}(\AC)$ and an initial state distribution~$\rho_0 \in \mathrm{Distr}(\ST)$, the goal of \ac{PG} methods is to find the optimal parameters~$\theta^\star$ that minimize the expected return
\begin{subequations}
\label{eq:policy_obj}
\begin{align}
	 J^{\pi}(\theta) &\coloneqq \E_{S \sim \rho_0} \bigl[ V^{\pi_\theta}(S) \bigr] \\
     &=\frac{1}{1 - \gamma}  \E_{S \sim \rho^{\pi_\theta}(s),\; A \sim \pi_\theta(\cdot \mid S)} \left[l(S, A) \right],
\end{align}
\end{subequations}
where~$\rho^\pi$ is the discounted visitation frequency defined as follows.
Given a policy~$\pi_\theta$, the environment model~$P$ and an initial state~$s$, let~$p(s \rightarrow s^+, k, \pi_\theta)$ be the probability of reaching state~$s^+$ at time step~$k$ by starting from state~$s$ following policy~$\pi_\theta$.
The normalized discounted visitation frequency is defined by~$\rho^{\pi_\theta} (s^+) \coloneqq (1 - \gamma) \E_{S \sim \rho_0} \left[\sum_{k=1}^{\infty} \gamma^{k-1} p(S \rightarrow s^+, k, \pi_\theta) \right]$.
It is important to highlight that finding a policy that minimizes the objective of \eqref{eq:policy_obj} is less restricting as solving the \ac{MDP} over the full state space as defined in \eqref{eq:general_problem}.
The reasoning is that some parts of the state space might not be reached by the policy, which could be omitted if the support of the initial state distribution $\rho_0$ is required to cover the whole state space.

In order to find an update direction in which the expected return~\eqref{eq:policy_obj} improves, its gradient~$\nabla_\theta J^\pi(\theta)$ is required.
Different reformulations of \eqref{eq:policy_obj} exist that allow one to derive sample estimates of the gradient of the \ac{PG} objective, $\nabla_\theta J^\pi(\theta)$.
First, the stochastic policy-gradient theorem reformulates the policy gradient by
\begin{align}\label{eq:policy_gradient}
\nabla_\theta J^\pi (\theta)
    = \frac{1}{1 - \gamma} \E_{\substack{S \sim \rho^{\pi_\theta}(s),\\ A \sim \pi_\theta(\cdot \mid S)}} \left[Q^{\pi_\theta}(S, A) \nabla_\theta \log \pi_\theta (A \mid S) \right].
\end{align}
building the theoretical foundation of the REINFORCE algorithm~\cite{williams_simple_1992} or the first actor-critic methods~\cite{sutton_reinforcement_2018}.
A derivation is provided in \cite{sutton_reinforcement_2018}.

The second reformulation is called the deterministic policy gradient theorem.
Let~$\mu_\theta: \ST \rightarrow \AC$ be a deterministic policy.
By differentiating through the expected state-action value function~$Q^{\mu_\theta}$ the \ac{DPG} is obtained by
\begin{equation}\label{eq:deterministic_policy_gradient}
	\nabla_\theta J^\pi(\theta)  =  \frac{1}{1 - \gamma} \E_{S \sim \rho^{\pi_\theta}(s)} \left[ \left. \nabla_\theta \mu_\theta (S) \nabla_a Q^{\mu_\theta} (S, a) \right|_{a=\mu_\theta(S)} \right].
\end{equation}
A derivation is provided in \cite{silver_deterministic_2014}.
The advantage of the \ac{DPG} is particularly prominent in high-dimensional action spaces since in the stochastic \ac{PG}, actions~$A$ are sampled to estimate the gradient~\cite{silver_deterministic_2014}, leading potentially to a higher gradient variance.
Thus, the \ac{DPG} formulation is used in many of the state-of-the-art methods like \ac{TD3} \cite{fujimoto_addressing_2018} and \ac{SAC} \cite{haarnoja_soft_2018}.
A concrete algorithm for an actor-critic algorithm using the \ac{DPG} is provided in section~\ref{ssec:ddpg}.

An important distinction between the stochastic \ac{PG} and the \ac{DPG} is the ability to handle discrete action spaces.
With the stochastic \ac{PG}, discrete and continuous action spaces can be directly optimized, whereas the \ac{DPG} requires a differentiable parameterized policy with respect to the parameters.
To circumvent this problem, some work extends the \ac{DPG} to stochastic policies using relaxation techniques to differentiate through the sampling process of the discrete actions \cite{tilbury_revisiting_2023}.

\subsection{Deep Reinforcement Learning Methods}
\label{sec:algo_rl}
In the following, an overview of four influential deep \ac{RL} algorithms,  namely \ac{DQN}, \ac{DDPG}, \ac{PPO} and \acf{SAC} is given.\\

\subsubsection{Deep Q-Networks}
Function approximators like \acp{NN} approximate the action-value function to extend Q-learning to continuous state spaces.
One prominent implementation of this is \ac{DQN}~\cite{mnih_human-level_2015}, a combination of deep learning and \ac{RL}.
\ac{DQN} collects transition samples in a buffer~$\mathcal{D}^\mathrm{buffer}$ and minimizes the mean squared error between the current value function~$Q_w$ and the sampled target value by
\begin{multline} \label{eq:qlearning_dqn}
	\mathcal{L}^Q_\mathrm{DQN} (w) \coloneqq  \E_{(S, A, S^+) \sim \mathcal{D}^\mathrm{buffer}} \biggl[ \\
    \left( l(S, A) + \gamma \min_{a^\star} Q_{\bar{w}}(S^+, a^\star) - Q_w(S, A) \right)^2 \biggr].
\end{multline}
For the sample target, a fixed copy~$\bar{w}$ of the parameter~$w$ is used that is only periodically updated.
This stabilizes the training~\cite{mnih_human-level_2015}.
For an overview of different function approximation methods and their potential instabilities, see~\cite{sutton_reinforcement_2018}.

Given a parameterized Q-function~$Q_w$, the resulting update scheme from \ac{DQN} is given by
\begin{subequations}\label{eq:dqn_update_scheme}
\begin{algbox}{DQN $(Q_w \approx Q^\star)$}
\begin{align}
    & w \leftarrow w + \frac{\alpha_w}{B} \sum_{i=1}^B \delta_i \nabla_w Q_w (S_i, A_i), \label{eq:dqn_update} \\
    &\delta_i \coloneqq l (S_i, A_i) + \gamma \min_{a^\star} Q_{\bar{w}} (S^+_i, a^\star) - Q_w (S_i, A_i), \\
    &\textrm{with } (S_i, A_i, S^+_i) \sim \mathcal{D^\mathrm{buffer}}.
\end{align}
\end{algbox}
\end{subequations}
The update rule in \eqref{eq:dqn_update} provides a sample-based estimate of the gradient in \eqref{eq:qlearning_dqn}, where $B$ represents the batch size used for the update and $\alpha_w$ a learning rate.
Averaging over multiple samples is often called mini-batch training and leads to improved performance and accelerated convergence when training \ac{NN} \cite{lecun_efficient_1998}.
Exploration in \ac{DQN} is handled again by the $\epsilon$-greedy policy~$\pi^\epsilon_Q$.
Differently to the update scheme of Q-learning \eqref{eq:q_learning}, the buffer $\mathcal{D}^\mathrm{buffer}$ stores state and actions that were generated from previous policies derived from $Q_w$ earlier in the training.
Similar to Q-learning, \ac{DQN} is also an off-policy method.\\

\subsubsection{Deep Deterministic Policy Gradient}\label{ssec:ddpg}
Extending \ac{DQN}, \acf{DDPG} \cite{lillicrap_continuous_2016} is an off-policy actor-critic method that learns a deterministic policy $\mu_\theta$ and a value function $Q_w$ in parallel.
Both the actor and the critic are $\acp{NN}$.
The update rule of \ac{DDPG} is given by
\begin{algbox}{DDPG $(\mu_\theta \approx \pi^\star)$}
\begin{subequations}\label{eq:ddpg}
\begin{align}
    & w \overset{Q}{\leftarrow} w + \frac{\alpha_w}{B} \sum_{i=1}^B \delta_i \nabla_w Q_w (S_i, A_i), \\
    & \theta \overset{\mu}{\leftarrow} \theta + \frac{\alpha_\theta}{B} \sum_{i=1}^B \nabla_\theta \mu_\theta(S_i) \left.\nabla_a Q_w(S_i, a)\right|_{a=\mu_\theta(S_i)}, \label{eq:ddpg_actor} \\
    &\delta_i \coloneqq l (S_i, A_i) + \gamma Q_{\bar{w}} (S^+_i, A^+_i) -  Q_w (S_i, A_i), \\
    &\textrm{with } (S_i, A_i, S_i^+) \sim \mathcal{D^\mathrm{buffer}},\; A^+_i =  \mu_\theta(S_i^+),
\end{align}
\end{subequations}
\end{algbox}
where $\overset{Q}{\leftarrow}$ denotes the update for $Q$ and $\overset{\mu}{\leftarrow}$ denotes the update for $\mu$.
As can be seen from the update scheme, the critic~$Q_w$ is used to update the actor $\mu_\theta$ in \eqref{eq:ddpg_actor}, in that sense ``criticizing'' the actor.
Differently to the update rule of \ac{DQN} \eqref{eq:dqn_update}, where the greedy action is considered to build the temporal difference, in \ac{DDPG}, a single evaluation update with respect to the current policy~$\mu_\theta$ is performed.
The buffer~$\mathcal{D}^\mathrm{buffer}$ is filled up over time by using the policy~$\mu_\theta + \xi$, where $\xi$ is either an Ornstein-Uhlenbeck process for temporally correlated noise or a Gaussian~\cite{lillicrap_continuous_2016}.
\\

\subsubsection{Proximal-Policy Optimization}
A popular on-policy actor-critic method that can be used for discrete and continuous action spaces is \acf{PPO}~\cite{schulman_proximal_2017}.
Prior to each policy and critic update, data in the form of multiple episodes is collected.
Since \ac{PPO} is an on-policy method, transitions generated earlier in the training by outdated policies are discarded.
In practice, \ac{PPO} is often used with high-speed simulation environments, where generating new samples comes with low computational costs.
A main advantage of \ac{PPO} is its ability to prevent drastic parameter updates that could potentially destabilize the training.
Similar to trust region methods~\cite{schulman_trust_2015}, this is achieved by restricting the policy update.

During training, a stochastic policy~$\pi_\theta$ -- often a parameterized Gaussian -- is used.
Assuming a given initial state~$s_0$, a trajectory is drawn by the forward simulation~$S_{k+1} \sim P(\cdot | S_k, A_k)$ and~$A_k \sim \pi_\theta(\cdot | S_k)$ until a maximum roll-out length~$M$.
Given multiple roll-outs, an estimate~$\hat{A}(S_k, A_k)$ of the advantage function defined by~$A^\pi(S_k, A_k) \coloneqq Q^\pi(S_k, A_k) - V^\pi(S_k)$ can be derived, see \cite{schulman_high-dimensional_2016}.
Additionally, with the probability ratio~$ R_k(\theta) \coloneqq \pi_\theta(A_k|S_k) / \pi_{\bar{\theta}}(A_k|S_k)$, which measures how much the new policy~$\pi_\theta$ changes with respect to the current policy~$\pi_{\bar{\theta}}$, the \ac{PPO} clipping objective is defined by
\begin{align*}
	\begin{split}
		J^\pi_\mathrm{CLIP} (\theta) \coloneqq \E \Biggl[ \sum_{k=0}^{M - 1} &\max \Big\{ R_k (\theta) \hat{A}(S_k, A_k),\\
		&\mathrm{clip}(R_k(\theta), 1- \epsilon, 1 + \epsilon) \; \hat{A}(S_k, A_k) \Big\} \Biggr]\,,
	\end{split}
\end{align*}
where the $\mathrm{clip}$ function projects the ratio~$R_k$ to an interval from $1 - \epsilon$ to $1 + \epsilon$.
Note that the clipping objective requires maximization, whereas the original \ac{PPO} objective involves minimization, as in this work, costs are minimized rather than rewards being maximized.

One of the primary advantages of \ac{PPO} is its simplicity and ease of implementation~\cite{schulman_proximal_2017} compared to previous methods based on trust region optimization, see~\cite{schulman_trust_2015}.
A proof of convergence of \ac{PPO} to the optimal policy for discrete \acp{MDP} is given in \cite{grudzien_mirror_2022}.\\

\subsubsection{Soft Actor-Critic}
The \ac{SAC} algorithm~\cite{haarnoja_soft_2018, haarnoja_soft_2019} is a widely used off-policy actor-critic method that incorporates an entropy bonus for stochastic policies with higher entropy.
Using entropy regularization is considered in the framework of maximum-entropy \ac{RL} \cite{ziebart_maximum_2008}.
Like \ac{DQN} and \ac{DDPG}, \ac{SAC} optimizes the policy in an off-policy manner collecting transitions encountered during training in a replay buffer, leading to an improved sample efficiency when compared to \ac{PPO}.

\Ac{SAC} extends the objective~\eqref{eq:policy_obj} by introducing an entropy regularization term, leading to a ``soft'' policy gradient objective
\begin{multline*}
	J^\pi_\mathrm{soft}(\theta) \coloneqq \\ \frac{1}{1 - \gamma} \E_{S \sim \rho^{\pi_\theta}, \; A \sim \pi_\theta(\cdot | S)} \left[ l(S, A) + \lambda_\mathcal{H} \; \mathcal{H}(\pi_\theta(\cdot \mid S)) \right],
\end{multline*}
where $\mathcal{H}$ denotes the entropy of the paramterized policy~$\pi_\theta (\cdot | S)$ at the sampled state~$S$.
Given a distribution $X$, the entropy is defined by $\mathcal{H}(X) = \E [- \log(X) ]$.
The entropy regularization, scaled by the parameter~$\lambda_\mathcal{H}$, encourages exploration, stabilizes policy training \cite{ball_offcon3_2021} and can lead to more robust policies \cite{eysenbach_maximum_2022}.
Additionally, it has been shown in \cite{muller_essentially_2024} that for discrete \acp{MDP}, the error introduced by the entropy regularization decreases exponentially to the inverse regularization strength,~$1 / \lambda_\mathcal{H}$.

In \ac{TD3}, twin Q-networks addressing the overestimation bias in Q-value estimation were introduced \cite{fujimoto_addressing_2018}.
Whereas in \cite{haarnoja_soft_2018}, the fundamentals of \ac{SAC} are described, \cite{haarnoja_soft_2019} introduced an improved version using twin Q-networks and automatic tuning of the entropy regularization with $\lambda_\mathcal{H}$.
The latter builds the basis for most implementations in current \ac{RL} software frameworks.

\section{Model Predictive Control}
\label{sec:mpc}

This section introduces \ac{MPC}, a commonly used framework to obtain policies for continuous \acp{MDP}.
\Ac{MPC} utilizes a -- typically deterministic -- model that approximates the true stochastic environment~\cite{rawlings_model_2017, grune_nonlinear_2017}.
Starting from the current environment state, \ac{MPC} uses the internal model to predict how different choices of the planned control trajectory would affect the state trajectory and evaluates the cost associated with this prediction.
Usually, evaluating the \ac{MPC} policy involves solving an optimization problem online to obtain the control input.
In this section, we will first give an overview of \ac{MPC} problem formulations and the relevant considerations, followed by a discussion of algorithms used for finding their solution.

\subsection{MPC Problem Formulations}
\label{sec:mpc_prob}
Solving the \ac{MDP} optimization problem~\eqref{eq:general_problem} is, in general, intractable due to several reasons, including the infinite horizon, the optimization over the space of policy functions, and the expectation over nonlinear transformations of stochastic variables.
\ac{MPC} leverages several approximations of~\eqref{eq:general_problem} in order to derive a computationally tractable optimization problem.

As a first step, the optimal policy is computed only for the current state~$s$ and the infinite horizon in~\eqref{eq:general_expected_return} and~\eqref{eq:general_problem} is approximated by a finite horizon, resulting in the optimization problem
\begin{align}\label{eq:soc_finhorz} 
	\min_\pi \;&\E \left[\gamma^N \bar{V}^\pi(S_N)+ \sum_{k=0}^{N-1} \gamma^k l(S_k, A_k)    \right],
	\\
	\mathrm{where}\quad&S_0 = s, \;S_{k+1} \sim P(\cdot \mid S_k, A_k),\; A_{k} \sim \pi(\cdot \mid S_k),  \nonumber
\end{align}
where the terminal cost function~$\bar{V}^\pi$ is an approximation of the exact value (or cost-to-go) function~$V^\pi$.
In the second step, the true stochastic state transition~$S_{k+1} \sim P(S_k, A_k)$ is approximated by a simplified model.
The most commonly used formulation is nominal \ac{MPC}, in which a deterministic model is used, i.e., $x_{k+1} = f^\mpcShrt(x_k, a_k)$. Hence, uncertainty is not explicitly considered.
Here, we introduced~$x_k\in\mathcal{S}$ to denote predictions of the state within the MPC problem.
Since a deterministic model does not capture the possibility of deviations from the prediction, the planned action trajectory is a trajectory of fixed actions~$(u_0, \dots, u_{N-1})$, $u_k \in \mathcal{A}$, as opposed to a policy function~$\pi$.
The resulting deterministic optimal control problem is given by
\begin{align}
	\label{eq:mpc_from_mdp}
\begin{split}
\min_{u_0,\ldots,u_{N-1}}&\;\!\gamma^N \bar{V}^\star(x_{N}) +\sum_{k=0}^{N-1} \gamma^k l(x_k, u_k)  ,\\
\mathrm{where}\quad
&  x_0 = s,\; x_{k+1}= f^\mpcShrt(x_k,u_k),
\end{split}
\end{align}
where $\bar{V}^\star$ is an approximation of the optimal terminal value function.

In the MDP framework~\eqref{eq:mdp}, the stage cost~$l(x_k, u_k)$ may assign some regions of the state space an infinite cost in order to prohibit them.
Similarly, due to actuator limitations, the action space~$\mathcal{A}$ is often a compact subset of~$\R^{n_u}$.
In numerical optimization, this is typically handled by explicitly considering constraints~$h^\mpcShrt(x_k, u_k)$ and the terminal safe set~$h^\mpcShrt_N(x_k, u_k)$ as part of the problem formulation.
This results in a constrained \ac{NLP},
associated with the MPC value function~$V^\mathrm{MPC}(s)$ and the terminal value function~$\bar{V}^\mathrm{MPC}(s)$ within the NLP formulation, and can be stated as
\begin{mini!}|s|
	{z}
	{
		\!\!\bar{V}^\mpcShrt(x_{N})+\sum_{k=0}^{N-1} l^\mpcShrt(x_k, u_k) 
		\label{eq:mpc_cost}
	}
	{\label{eq:mpc-problem}}
	{\!\!V^\mpcShrt(s)\!=\!\!}
	\addConstraint{x_0}{= s \label{eq:mpc_main_c_initial}}
	\addConstraint{x_{k+1}}{= f^\mpcShrt(x_k, u_k),\;}{0\leq k < N \label{eq:mpc_main_c_dyn}}
	\addConstraint{0}{\leq h^\mpcShrt(x_k, u_k),}{1\leq k < N \label{eq:mpc_main_c_const}}
	\addConstraint{0}{\leq h_N^\mpcShrt(x_N), \label{eq:mpc_main_c_const_term}}{}
\end{mini!}
using the vector of decision variables~$z=(x_0,\ldots x_N, u_0, \ldots, u_{N-1})\in\R^{n_z}$.
In the above formulation, we introduced the state trajectory~$(x_0, \dots, x_{N})$ as additional decision variables, which are constrained to start at the given value of the current state~\eqref{eq:mpc_main_c_initial}, and to follow the system dynamics~\eqref{eq:mpc_main_c_dyn}.
This is in contrast to formulation~\eqref{eq:mpc_from_mdp}, where the state trajectory is considered as an explicit function of the action trajectory.
Both approaches are equivalent in terms of the solutions they admit. However, the iterations of numerical optimization algorithms may differ depending on the formulation.
The formulation in~\eqref{eq:mpc_from_mdp} is referred to as a single shooting or sequential formulation, whereas~\eqref{eq:mpc-problem} is a multiple shooting or simultaneous formulation.
For nonlinear unstable systems, the latter is typically preferable.

When deploying \ac{MPC}~\eqref{eq:mpc-problem}, the NLP is solved online at every discrete time instant, based on a specified value of the current state~$s$.
This yields an optimal trajectory of actions, $(u_0^\star, \dots, u_{N-1}^\star)$, of which only the first one, $u_0^\star$, is applied to the environment.
The resulting new state is then used as the initial state for the next optimization problem, and a new input trajectory is computed. Hence \ac{MPC}~\eqref{eq:mpc-problem} defines a policy.

Similarly to the definition of the value function in~\eqref{eq:mpc-problem}, we can define the corresponding Q-function by additionally fixing the initial action vector~$u_0$ to the given value, 

\begin{mini!}|s|
	{z }
	{
		\!\!\sum_{k=0}^{N-1} l^\mpcShrt(x_k, u_k)  \!+\! \bar{V}^\mpcShrt(x_{N})
	}
	{\label{eq:mpc-q-problem}}
	{\!\!Q^\mpcShrt(s, a)\!=\!\!}
	\addConstraint{x_0}{= s}
	\addConstraint{u_0}{= a}
	\addConstraint{x_{k+1}}{= f^\mpcShrt(x_k, u_k),\,\,}{0\leq k < N}
	\addConstraint{0}{\leq h^\mpcShrt(x_k, u_k),}{1\leq k < N }
	\addConstraint{0}{\leq h_N^\mpcShrt(x_N).}{}
\end{mini!}

Based on this Q-function and assuming a unique minimizer, the \ac{MPC} policy is 
\begin{align}
	\mu^\mpcShrt(s)=\arg \min_a ~Q^\mpcShrt(s, a),
\end{align}
where $\mu$ instead of $\pi$ is used to denote that the \ac{MPC} policy is deterministic.

We now take a closer look at the components of the \ac{MPC} problem~\eqref{eq:mpc-problem}, followed by various other relevant considerations. \\

\subsubsection{Dynamics model}
The central component is the dynamics constraint~\eqref{eq:mpc_main_c_dyn}, in which~$f^\mpcShrt(x,u)$ is a model of the environment~\eqref{eq:general_model}.
This model can be derived from first principles or identified from data.
In the case of nominal \ac{MPC}, it is deterministic and does not consider the stochasticity of~\eqref{eq:general_model}.\\

\subsubsection{Objective}
The stage cost~$l^\mpcShrt(x,u)$ in the \ac{MPC} objective~\eqref{eq:mpc_cost} commonly corresponds to the stage cost of the \ac{MDP}~\eqref{eq:mdp}.
Note that the stage cost may be implicitly time-varying by including an augmented clock state. Thus, it may implicitly include a discounting factor~$\gamma$.
Crucially, the \ac{MPC} optimization problem needs to be numerically tractable, and the ultimate goal is to provide a policy that achieves a sufficient closed-loop performance when applied to the real environment. Therefore, the stage cost may also be approximated by a function with favorable numerical properties.
Often, the stage cost imposes a convex penalty on the deviation of the action and state trajectories from a reference point or trajectory. This is referred to as regulatory \ac{MPC} or tracking \ac{MPC}. The focus in regulatory MPC problems is mostly the stabilization of systems. Historically, the ability of MPC to incorporate inequality constraints and to handle multiple inputs simultaneously justified the additional computational complexity. 
In contrast, economic \ac{MPC} uses a cost that directly expresses a quantity of interest to optimize, e.g., time, energy use, financial cost, or yield of a production process. Therefore, economic MPC is closely related to solving \acp{MDP} beyond stabilization.

The terminal cost~$\bar{V}^\mpcShrt(x)$
should ideally capture the cost-to-go at the end of the \ac{MPC} horizon, cf. \eqref{eq:soc_finhorz}.
The choice of horizon length and terminal cost is often crucial for the performance of the resulting \ac{MPC} policy.
Indeed, a close-to-optimal cost-to-go approximation via the terminal cost allows for shorter prediction horizons.
The approximation quality becomes less important as the horizon length grows, and $\bar{V}^\mpcShrt(x_{N}) \equiv 0$ is a widely used choice \cite{grune_economic_2013}. 
However, the NLP~\eqref{eq:mpc-problem} becomes computationally more expensive for longer horizons.
Thus, the horizon length is typically limited by the available computational resources.
In practice, the terminal cost function is often chosen heuristically or based on stability considerations. Still, there is also research on how to explicitly select it as an approximation of the infinite-horizon cost (with respect to the MPC cost), e.g., by simulating forward a pre-selected simple feedback law \cite{ nicolao_stabilizing_1998, nicolao_stabilizing_1996, diehl_online_2003, diehl_efficient_2004}.
When stabilizing a system at a steady state using a locally smooth convex stage cost, a straightforward choice is the infinite horizon \ac{LQR} \cite{anderson_optimal_1990} cost computed for the system resulting from the linearization of the model~\eqref{eq:mpc_main_c_dyn} at that steady state.\\

\subsubsection{Constraints}
The final components are the stage and terminal constraints.
Stage constraints~\eqref{eq:mpc_main_c_const} can be used to avoid prohibited regions of the state space and to take into account actuator constraints.
Terminal constraints~\eqref{eq:mpc_main_c_const_term} can be used to ensure stability and recursive feasibility of the resulting \ac{MPC} policy. 
For the terminal constraints, considerations similar to the terminal cost apply.
In principle, they should capture the system's future behavior over the infinite horizon and ensure that the \ac{MPC} plan will not be too short-sighted regarding the stability and recursive feasibility of the resulting \ac{MPC} policy.

Constraints can lead to situations in which~\eqref{eq:mpc-problem} is infeasible for the current initial state value~$s$, i.e., there exists no value for the decision variables such that all constraints are satisfied.
Consequently, the solver would not be able to return a solution, and the evaluation of the MPC policy would fail.
Thus, for practical \ac{MPC} implementations, it can often be helpful not to enforce the constraints strictly but to penalize their violation.
For exact penalties with sufficiently high penalty weight, constraint satisfaction is guaranteed if the original problem
is feasible \cite{nocedal_numerical_2006, kerrigan_robust_2000}.
Otherwise, a solution that minimizes the constraint violation is returned.\\

\subsubsection{Uncertainty-aware MPC}
In contrast to the deterministic model in~\eqref{eq:mpc_main_c_dyn}, stochastic and robust \ac{MPC} formulations \cite{rawlings_model_2017-1, mesbah_stochastic_2016, kouvaritakis_model_2015} explicitly take into account the uncertainty of the model prediction.
The former models the uncertainty as a probability distribution, whereas the latter predicts bounded sets of all possible realizations (respective tractable outer approximations).
Both can be separated into scenario \cite{shapiro_lectures_2009, calafiore_scenario_2006} and tube \cite{langson_robust_2004, mayne_tube-based_2011} approaches.
Scenario approaches consider discrete distributions, which may be obtained by sampling from a continuous distribution.
This can take the form of sampling several disturbance trajectories separately \cite{nagabandi_neural_2018} and planning the corresponding state trajectories in parallel or of constructing a tree of scenarios that are branched at every time step \cite{kouzoupis_dual_2019}.
Tube approaches predict parameterized approximations of the state distribution respective uncertainty set trajectories.
Typical parametrizations are, e.g., normal distributions, ellipsoids \cite{kurzhanski_ellipsoidal_1997,houska_robust_2011,gillis_positive_2013,villanueva_robust_2017} or polytopes \cite{rakovic_parameterized_2012, rakovic_robust_2019, fleming_robust_2015, rakovic_homothetic_2022, villanueva_configuration-constrained_2024}.
This allows for a finite-dimensional representation of the uncertainty, such that a tractable \ac{NLP} is obtained.

Both scenario and tube approaches can also be classified into open-loop and closed-loop formulations.
Open-loop formulations plan only one fixed trajectory of actions.
This can quickly lead to unrealistically conservative uncertainty predictions because they do not encode that the noise will be counteracted in the real environment by feedback~\cite{ben-tal_adjustable_2004}.
Closed-loop predictions consider future feedback, leading to more realistic predictions.
Using scenario trees, this can be achieved by planning a distinct action for every tree node.
This implicitly corresponds to planning over policies with respect to the discretized disturbance space because the actions depend on past disturbances. Here, nonanticipativity with respect to the causality of the policy should be carefully considered.
Closed-loop tube approaches usually consider explicitly parameterized simple feedback laws, e.g., linear feedback, which reacts to state deviations from the tube center. These feedback laws can be precomputed \cite{mayne_tube-based_2011, parisini_receding-horizon_1995} or optimized \cite{nagy_open-loop_2004, messerer_efficient_2021, villanueva_configuration-constrained_2024}.
While the optimization of state feedback gains is highly nonconvex even for linear systems, the optimization over affine disturbance feedback leads to equivalent convex, but also higher-dimensional,  optimization problems \cite{ben-tal_adjustable_2004,goulart_optimization_2006}. 
In the context of polyhedral sets, it is, under some assumptions, sufficient to consider only their vertices, which results in tree-structured formulations \cite{scokaert_min-max_1998, diehl_formulation_2007}.\\

\subsubsection{Optimization problem classification}
Depending on the mathematical form of the functions in~\eqref{eq:mpc-problem}, the optimization problem can be classified differently.
This is relevant, as it informs both the choice of solution algorithm and the theory regarding the resulting policy.
In \ac{LMPC}, the model is linear, the constraint functions are affine, and the cost functions are convex quadratic, resulting in a \ac{QP}. \Ac{LMPC} problems can be solved reliably and efficiently.
This is often used in contrast to \ac{NMPC}, where typically the model is nonlinear, and the resulting optimization problem is an \ac{NLP}~\eqref{eq:mpc-problem}. 
When the action space is discrete, this corresponds to an additional restriction of the action variables to the space of integers, resulting in a \ac{MINLP} or \ac{MIQP}~\cite{reiter_mixed-integer_2021}.
If the dynamics contain nonsmooth or discontinuous events, such as contact physics, this leads to \acp{MPCC}~\cite{nurkanovic_numerical_2023}.\\

\subsubsection{Suboptimality and control theory}
As discussed, the \ac{MPC} problem~\eqref{eq:mpc-problem} leverages several forms of approximations of the \ac{MDP}~\eqref{eq:general_problem}.
This opens the question of how the resulting \ac{MPC} policy behaves with respect to the MDP or when applied to a real system.
Since it solves the \ac{MDP} only approximately, it can be considered a form of suboptimal control.
An overview of several sources of suboptimality can be found in \cite{bertsekas_dynamic_2005}.
The suboptimality from the finite horizon approximation is analyzed in \cite{grune_infinite_2008, karapetyan_finite-time_2023, li_performance_2023}.
The consequences of approximating the expected value via sampling are addressed in the stochastic programming literature in the context of the sample average approximation \cite{shapiro_lectures_2009}.
In \cite{hadjiyiannis_efficient_2011}, the authors analyze the suboptimality resulting from affine feedback parametrization in a robust problem formulation.
The suboptimality of nominal MPC in a stochastic environment is analyzed in \cite{fleming_stochastic_1971,messerer_fourth-order_2024}, and regret bounds in \cite{lin_bounded-regret_2022}.

The control theory literature often asks a different, though closely related, question, see, e.g., \cite{rawlings_model_2017, grune_nonlinear_2017}.
It investigates under which conditions the \ac{MPC} policy exhibits desirable behavior.
This includes system theoretical properties of the resulting closed-loop system, such as stability \cite{chen_quasi-infinite_1998, mayne_constrained_2000, limon_input--state_2009}, and properties like recursive feasibility \cite{grune_nmpc_2012, lofberg_oops_2012, mayne_apologia_2013}, i.e., the controller should not maneuver itself into a state in which the \ac{MPC} problem~\eqref{eq:mpc-problem} becomes infeasible.
\Ac{MPC} should also be able to work under model-plant mismatch and reject disturbances. This is referred to as inherent robustness, and the theory covers nominal \ac{MPC}, which uses no uncertainty model \cite{scokaert_discrete-time_1997,nicolao_robustness_1996,magni_stability_1997, mcallister_inherent_2022,grimm_examples_2004,yu_inherent_2014}, but can also be extended to stochastic MPC \cite{mcallister_inherent_2022}, even if it is designed with respect to a wrong disturbance model \cite{mcallister_inherent_2024}.
Additionally, fast-paced applications may only allow for a suboptimal solution to~\eqref{eq:mpc-problem} in the assigned computational budget.
Stability \cite{scokaert_suboptimal_1999, diehl_nominal_2005} and inherent robustness results  \cite{pannocchia_conditions_2011,allan_inherent_2017} also exist for this case.
Further, while stability results are usually derived for continuous action spaces, they can also be generalized to discrete actuators \cite{rawlings_model_2017-1}.

\subsection{Numerical Methods for MPC Problems}
\label{sec:mpc_algo}

We distinguish two common and fundamentally different approaches to -- possibly approximately -- solving finite-horizon \ac{OCP} formulations:
Sampling-based methods and methods leveraging derivative-based numerical optimization.
In the following, we briefly introduce and discuss both approaches, focusing on algorithms that address the nominal OCP formulation.\\

\subsubsection{Sampling-Based Methods}

As a first class of methods, we consider sampling-based approaches that aim at finding an approximate solution to the stochastic open-loop or nominal OCP by sampling control trajectories and evaluating the associated cost via forward simulation.

The simplest sampling-based method, also known as random shooting and used, e.g., in \cite{piovesan_randomized_2009, nagabandi_neural_2018}, considers a finite number of independently sampled action sequences and the corresponding state trajectories, which are obtained via forward simulation.
The algorithm then chooses the open-loop action trajectory associated with the lowest cost as an approximate solution.

A second, more sophisticated, sampling-based approach is the \ac{CEM}~\cite{rubinstein_optimization_1997}, where the probability distribution generating the open-loop action trajectory samples is iteratively refined based on previously sampled actions yielding low costs.
In particular, \ac{CEM} samples a finite number of action sequences and evaluates their associated costs via forward simulation.
The action sequences yielding the lowest cost trajectories are used to adapt the probability distribution from which new action samples are generated.
Typically, a Gaussian distribution or a Gaussian mixture model is used, in which the mean and covariance are adapted at each step.
The method has been successfully implemented for \ac{MPC} in \cite{kobilarov_cross-entropy_2012, chua_deep_2018}.


As a third sampling approach, we consider \mppi{}, which provides a framework \cite{kappen_linear_2005, theodorou_relative_2012, theodorou_nonlinear_2015, williams_model_2017} for solving \ac{MPC} problems via trajectory sampling.
In particular, the approach directly addresses the stochastic formulation without resorting to the nominal problem~\eqref{eq:mpc-problem}.
As the method is derived based on a continuous-time model, we refer the interested reader to \cite{williams_model_2017} for an in-depth derivation.
Crucially for this survey, the method poses some strong restrictions on the structure of the model and cost:
First, the stochasticity of the dynamics needs to enter via the actions, i.e., $f_{\textsc{MPPI}}(s_k, a_k+w_k)$, where~$w_k$ is a random variable in the dimension of the actions.
Secondly, the cost~$l^\mpc$ is assumed to be separable in states and actions, as well as quadratic in the actions with a weighting matrix that is inversely proportional to the noise variance.
Furthermore, the main theoretical results and optimality (in the limit of infinite samples) only hold if the dynamics are affine in controls and noise \cite{theodorou_nonlinear_2015}. 
Similar to \ac{CEM}, \ac{MPPI} adaptively updates the probability distribution from which action trajectories are sampled.
\ac{MPPI} uses a weighted average to update the mean of this distribution where the weights are based on the associated costs.
Since the underlying algorithm of \ac{MPPI} requires limited implementation efforts, many papers use custom implementations~\cite{goldfain_autorally_2019}. However, recently an efficient implementation as part of \texttt{TorchRL}~\cite{bou_torchrl_2023} and a CUDA-based parallel computation framework was published~\cite{vlahov_mppi-generic_2024}.

Sampling-based approaches naturally allow for stochastic models \cite{kantas_sequential_2009, chua_deep_2018} and can thus directly tackle the stochastic open-loop \ac{OCP}.
Furthermore, sampling-based methods can be applied to problems with highly nonlinear, or even non-smooth, costs or dynamics as long as a simulator is available.
While sampling-based optimization methods are typically straightforward to implement and benefit from parallelization, they scale poorly with the dimension of both the planning horizon and the dimension of the action space, i.e., they severely suffer from the curse of dimensionality.
Furthermore, state constraints can be addressed only via penalty reformulations or the rejection of infeasible samples.
For highly constrained systems, the rejection method further increases the sampling complexity.\\


\subsubsection{Derivative-Based Numerical Optimization}
\label{sec:mpc_algo_derivative}

As a second class of methods, we discuss derivative-based numerical optimization.
On this account, approaches derived from a continuous-time \ac{OCP} formulation, such as, e.g., indirect methods, collocation, and pseudospectral methods, as well as numerical simulation methods, fall into this class, but are outside of the scope of this survey.
We refer the interested reader to \cite{rao_survey_2009, binder_introduction_2001} for a survey on numerical methods starting from a continuous-time formulation and to \cite[Chap.~8]{rawlings_model_2017} for a textbook overview of direct methods.

Assuming that all problem functions in~\eqref{eq:mpc-problem} are sufficiently smooth, \ac{OCP}~\eqref{eq:mpc-problem}, which is typically referred to as multiple shooting formulation~\cite{bock_multiple_1984}, can directly be addressed with standard numerical methods for constrained nonlinear optimization.
While multiple shooting formulation~\eqref{eq:mpc-problem} keeps both states and actions as optimization variables, one may alternatively eliminate the states from~\eqref{eq:mpc-problem} via the equality constraints yielding the single shooting formulation as introduced at the beginning of this survey and given in~\eqref{eq:mpc_from_mdp}.
The single-shooting formulation typically results in dense subproblems with few optimization variables and can thus be efficiently solved by general-purpose nonlinear solvers.
In the following, we focus on the multiple shooting problem, given in~\eqref{eq:mpc-problem}, and tailored numerical methods addressing the particular problem structure arising from this formulation.

First, we distinguish between numerical methods that use first-order derivative information versus approaches that use second-order derivative information. First-order methods are less computationally complex but may require many more iterations. Second-order methods require less but computationally more complex iterations.
Due to their low complexity, first-order methods have been considered for solving \acp{OCP} in particular in the context of embedded applications \cite{patrinos_accelerated_2013, jerez_embedded_2014, kapernick_gradient_2014, englert_software_2019, kouvaritakis_model_2015}.
Furthermore, a tailored method for the scenario-based \ac{OCP} formulations based on \ac{ADMM} has been developed in \cite{kang_decomposition_2015}.

Second-order methods include nonlinear \ac{IP} methods and \ac{SQP}, two widely used classes of methods for numerical optimization.
We will briefly discuss both of them in the following.
For a more detailed overview of second-order numerical methods for optimal control, we refer to~\cite{magni_efficient_2009}.

Nonlinear interior point methods tackle the non-smooth \ac{KKT} conditions associated with the constrained nonlinear optimization problem by formulating an approximate but smooth root-finding problem parametrized by a homotopy parameter, which is iteratively lowered towards zero -- in the limit recovering the nonsmooth optimality conditions.
The intermediate root-finding problems are solved via Newton-type iterations.
Nonlinear interior point methods tailored to the \ac{OCP} structure are discussed in \cite{rao_application_1998, frey_active-set_2020, frison_hpipm_2020, vanroye_fatrop_2023}.

Within the \ac{SQP} framework, a sequence of quadratic approximations of the nonlinear \ac{OCP} is solved.
The quadratic subproblems arising in \ac{SQP} might, in turn, be solved using an interior point method or an active-set solver.
\ac{SQP} methods can typically be warm-started, rendering them particularly attractive for \ac{MPC} applications where the solutions to subsequent problem instances are expected to be similar.
For further implementation details tailored to \ac{SQP} methods for MPC, such as full and partial condensing, we refer to the survey in \cite{kouzoupis_recent_2018}.
\ac{SQP}-type solvers tailored to \ac{OCP} are implemented in \cite{houska_acado_2011, verschueren_acadosmodular_2022}.
A widely used approximate approach closely tied to \ac{SQP} is the \ac{RTI} \cite{diehl_real-time_2002, bock_constrained_2007, gros_linear_2016}.
Within the \ac{RTI} framework, a single iteration of an \ac{SQP} method is performed, i.e., a single \ac{QP} approximation to the nonlinear \ac{OCP} is solved, in order to obtain an approximate solution drastically reducing the computation time per iteration.

Both \ac{IP} and \ac{SQP} methods require the computation of (an approximation of) the Hessian of the Lagrangian associated with the \ac{OCP}.
While the exact Hessian yields locally quadratic convergence to the solution \cite{nocedal_numerical_2006}, its computation is typically costly.
This motivates the use of Hessian approximations, which are cheaper to compute.
A common choice for OCP is the Gauss-Newton Hessian \cite{gauss_theoria_1809, bock_recent_1983}, which is applicable to nonlinear least squares objectives, and typically very cheap to compute (or even for free, in the case of quadratic objectives).
As an additional advantage, the resulting subproblems are convex by construction, unlike the exact Hessian, which may yield nonconvex subproblems.
However, due to the neglected curvature, in general only a linear convergence rate is achieved.
The core principle of Gauss-Newton can also be extended to problem classes beyond nonlinear least squares, cf. \cite{messerer_survey_2021} for an overview.
Quasi-Newton methods define alternative choices of Hessian approximation. Most prominently, these include the BFGS Hessian, which, with each iteration, converges towards the exact Hessian, yielding a superlinear convergence rate \cite{nocedal_numerical_2006}.


Independently of the particular choice of Hessian approximation, the subproblems encountered in both \ac{IP} and \ac{SQP} methods applied to the multiple shooting formulation~\eqref{eq:mpc-problem} exhibit a particular sparsity pattern, which is typically exploited by tailored solvers.


Another second-order method, which can not directly be interpreted within the Newton-type framework, is \ac{DDP}, originally proposed in~\cite{mayne_second-order_1966}.
\ac{DDP} is also based on solving \ac{QP} subproblems via a Riccati recursion, followed by a nonlinear forward sweep of the system, which employs the linear feedback law returned by the Riccati recursion.
The DDP variant using a Gauss-Newton Hessian approximation, which is more commonly referred to as~\ac{iLQR}, especially within the robotics community \cite{tassa_control-limited_2014}, has been introduced in \cite{li_iterative_2004, todorov_generalized_2005}.
In their standard form, \ac{DDP} and \ac{iLQR} cannot directly handle additional constraints, but extensions to \acp{OCP} with input bounds have been proposed e.g. in \cite{tassa_control-limited_2014, marti-saumell_squash-box_2020}.
\Ac{SLQ}~\cite{sideris_efficient_2005} is often mentioned in the context of \ac{DDP}, although it can more precisely be classified as an \ac{SQP} method with a Gauss-Newton Hessian approximation applied to the single shooting \ac{OCP} \eqref{eq:mpc_from_mdp} for which each of the \ac{QP} subproblems is solved in a sparsity exploiting manner, i.e., by a Riccati recursion, cf. also \cite{baumgartner_unified_2023}.

Crucially, all discussed methods will generally converge to a local optimum and thus require a sufficiently good initialization.
The local rate of convergence is determined by the accuracy of the Jacobian and Hessian approximation \cite{nocedal_numerical_2006}.


Only in the particular case of a convex NLP, such as for \acp{LMPC}, convergence to a global optimum can be guaranteed.
Under these assumptions, the solution map can be shown to be continuous and piecewise affine.
This fact is leveraged in explicit MPC where the optimal feedback law is precomputed offline in order to minimize online computation \cite{bemporad_explicit_2000}.
Explicit MPC is usually limited to small state dimensions, few inequality constraints, or short horizons. Explicit NMPC was investigated e.g. in \cite{grancharova_explicit_2012}

In addition to algorithms and software focusing on nominal \ac{OCP} formulations, numerical methods tailored to tree-structured problems arising in open-loop as well as closed-loop stochastic formulations have been developed 
\cite{steinbach_tree-sparse_2002, klintberg_improved_2016, frison_high-performance_2017, kouzoupis_recent_2018, kouzoupis_structure-exploiting_2019}.
Numerical methods addressing the tube-based open-loop and closed-loop stochastic \ac{OCP} formulation are presented in
\cite{patil_scaling_2015, feng_inexact_2020, zanelli_zero-order_2021, frey_efficient_2024} and 
\cite{goulart_efficient_2008, messerer_efficient_2021, leeman_fast_2024}, respectively.

In contrast to the exponential scaling of sampling-based approaches, e.g., shown for linear unstable systems in~\cite{yoon_sampling_2022}, solving the nonlinear \ac{OCP} via numerical optimization alleviates the curse of dimensionality, as the computational complexity of the nominal problem typically scales linearly with the horizon and polynomially with the state and control dimension.
On the other hand, the sub-problems need to be optimization-friendly and require the availability of derivative information. In comparison, sampling-based methods only require (fast) forward simulation.


\section{Comparison of MPC and RL}
\label{sec:comparison}

Starting from the general problem of solving \acp{MDP}, we have shown how \ac{RL} and \ac{MPC} are different techniques that aim at deriving optimal policies, cf., Sect.~\ref{sec:rl} and Sect.~\ref{sec:mpc}, respectively. 
\Ac{MPC}, for instance, emerged from the problem of solving multivariable constrained control problems, initially with the goal of setpoint stabilization~\cite{morari_model_1999}.
In contrast, \ac{RL} methods aim at maximizing closed-loop performance without necessarily relying on a model of the real environment.
Their discrepancy is unsurprising since both approaches were developed in parallel communities with different focuses. 
This is further substantiated in the nearly orthogonal properties of \ac{MPC} and \ac{RL}, which were reviewed and emphasized in a case study for a specific linear system in~\cite{gorges_relations_2017}.
The more recent adaptions towards economic \ac{MPC}, e.g.~\cite{angeli_average_2012} fit the paradigm of solving \acp{MDP} and, hence, bring the goals of both communities closer together.

Considering the practicalities of \ac{MPC} and \ac{RL} of solving \acp{MDP}, further particularities appear, which we compare in the following. The comparison is concluded in Tab.~\ref{tab:comparison_checkbox} that shows an overview of relevant properties, similar to~\cite{gorges_relations_2017} and Tab.~\ref{tab:literature_comparison}, from work that explicitly compares~\ac{MPC} and \ac{RL} in practical applications.
We compare requirements on the state space representation and the properties of the mathematical model, such as smoothness or continuity required by MPC.
We refer to~\cite{bertsekas_reinforcement_2019, bertsekas_dynamic_2005} for further theoretical comparisons.\\

\begin{table}
\caption{Comparison of practical properties between MPC and RL. The evaluation is simplified and conceptual. Exemptions may exist.}
	\centering
	\ra{1.2}
	\resizebox{\linewidth}{!}{%
	\begin{tabular}{@{}lrr@{}}
		\addlinespace
		\toprule
		Property & MPC& RL\\
		\midrule
		state-space & model specific \xmark& (quite) arbitrary \cmark\\
		model requirements & \makecell[r]{differentiability/\\ online simulation} \xmark& offline simulation \cmark\\
		uncertainty & \makecell[r]{guarantees with\\ known uncertainty} \cmark & \makecell[r]{probabilistic\\ guarantees} \cmark\\
		stability  & strong theory \cmark & minor theory \xmark \\
		constraint handling  & inherent \cmark & hard to achieve \xmark \\
		online computation time  & high \cmark & low \xmark \\
		offline computation time  & low \xmark & high \cmark \\
		adaptability  & inherent \cmark & needs retraining \xmark \\
		generalization  & inherent \cmark & \makecell[r]{poor when \\ out-of-distribution} \xmark \\
		\bottomrule
	\end{tabular}}
	\label{tab:comparison_checkbox}
\end{table}

In the following, the main practical and conceptual differences between MPC and RL are stated. \\

\subsubsection{State Space}
Real environments can only be approximately described by a state, which is usually not always directly measurable, leading to the concept of \acp{POMDP}. So far, we omitted a discussion about \acp{POMDP} and only mentioned some points required for a high-level discussion on the state space.
Considering an environment where the state~$s$ is unknown and only observations~$O_k\in\R^{n_o}$ are made at step~$k$, the most general concept of a state space would involve the collection of all observations~$O_0,O_1,\ldots$ and applied actions~$A_0, A_1, \ldots$. Even this very general concept may not allow to fully describe the real environment due to partial observability. However, often these observations, or even the $M$ most recent observations~$\mathcal{O}_M=\big(O_{k-M},A_{k-M},\ldots,A_{k-1},O_{k}\big)$ at step~$k$ are sufficient to estimate the relevant states of an environment~\cite{bakker_reinforcement_2001}.  

In \ac{MPC} and \ac{RL} algorithms, the state space is interpreted conceptually differently.
In derivative-based \ac{MPC}, the state space is usually constructed by relating it to an optimization-friendly model, usually having a physical interpretation following differential equations.
Yet, alternative approaches also consider a sequence of measurements as state~\cite{coulson_regularized_2019}, cf. Sect.~\ref{sec:mpc_algo}.
States used within derivative-based \ac{MPC} do not necessarily correspond to the measured sensor outputs and, thus, are often estimated by a state observer that converts a series of observations~$\mathcal{O}_M$ into a state estimate.

\ac{RL} methods and sampling-based \ac{MPC} methods may use states based on a corresponding physics-based model and the related estimators \cite{ceusters_model-predictive_2021,mamedov_safe_2024,brandi_comparison_2022} or a state observer. However, the state may also be kept as a recent history of raw sensor data, which could be based on a variety of different input modalities, such as images or text.
Often, end-to-end learning is used, where the raw sensor input, such as images, is fed directly to \acp{NN}. The \ac{NN} outputs subsequently the controls, e.g.,~\cite{kaufmann_champion-level_2023}.
It is common to make the learned dynamics model or policy dependent on a window of the most recent history of observations and actions~$\mathcal{O}_M$, cf., \cite{puterman_markov_2005}, which is particularly useful for \acp{POMDP}.
While an \ac{NN} architecture does not explicitly model states, the stacking of layers, the related transformations, the exploitation of equi/invariances, and the condensing of information can be interpreted as modeling of hidden states, cf.~\cite{lambrechts_recurrent_2022}. Subsequent layers can be interpreted as the policy based on these hidden states. Remarkably, this structure is not enforced explicitly. Also, recent approaches of the more classical observer/controller architecture propose to tune MPC and the state estimator together~\cite{simpson_efficient_2023}.
\\

\subsubsection{Model and Application}
Highly related to the state space are the characteristics of a potentially used model.
As explained in Sect.~\ref{sec:mpc}, \ac{MPC} requires a model to simulate the dynamics. Whereas sampling-based MPC only requires the forward simulation of a model, derivative-based MPC involves the computation of gradients through the model.
Therefore, derivative-based \Ac{MPC} requires an optimization-friendly model with stark limitations to its structure since the model is part of the optimization problem. 
The MPC optimization problem becomes particularly challenging if the model is non-smooth~\cite{nurkanovic_numerical_2023}, stochastic, or contains integer variables~\cite{ceusters_model-predictive_2021}.

A physically motivated prediction model, which is often used in MPC, has the advantage of better understanding and explaining the environment behavior, often referred to as explainability. The possibility of predicting interpretable model states allows for the straightforward definition of constraints. 
For example, a certain velocity must not be exceeded in a vehicle control problem. This constraint can be readily formulated by a model that predicts the system's velocity accurately.
Besides physically motivated models, more general models such as \acp{NN} or nonparametric models such as Gaussian processes~\cite{hewing_cautious_2020} can be used. These models still have the advantage of predicting the environment. However, the explainability of the states can be lost.  

Not requiring the explicit modeling of the environment is a major claim of \ac{RL}.
However, that statement needs some additional framing.
In fact, many successful model-free \ac{RL} applications use models, at least for training the policy~\cite{singh_reinforcement_2022}. 
One fundamental difference between models used for \ac{RL} or \ac{MPC} is that \ac{RL} models are often used for offline simulation but not during deployment of the final policy.
This allows an abundance of complex computations involved in the simulation, which could comprise logical statements or complex high-fidelity models.
Even though most \ac{RL} methods still require simulation models, real-world \ac{RL} is an active research field \cite{dulac-arnold_challenges_2021} making significant progress \cite{smith_demonstrating_2023}.



The low inference time of NNs makes it appealing to increase the sampling frequency of \ac{RL} algorithms for faster feedback loops.
However, a too-small discretization time step in \ac{RL} can result in the function approximation error exceeding the value difference of actions, rendering deep \ac{RL} methods that rely on function approximation useless~\cite{tallec_making_2019}.
Standard offline system identification techniques used to identify the MPC model also involve problems with too high sampling frequencies, e.g., a decreased signal-to-noise ratio and an ill-conditioned model~\cite{ljung_system_1999}.
However, in derivative-based MPC, the choice of the sampling frequency is often limited by the solution time of the underlying optimization problem.\\ 

\subsubsection{Intrinsic stochasticity}

As introduced in \eqref{eq:general_model}, the environment is typically assumed to be intrinsically stochastic.
Thus, even if the environment is perfectly known, it is not possible to precisely predict future state trajectories.
Both robust and stochastic \ac{MPC} explicitly take into account this uncertainty, which is in this context typically referred to as noise or disturbance.
Robust \ac{MPC} is, in a sense, agnostic to the question of whether the uncertainty is intrinsic or due to a systematic modeling error: it considers all trajectories possible under the given assumptions.
For stochastic \ac{MPC}, on the other hand, it is important to be aware of whether the prediction errors correlate over time since this affects the predicted state distribution. 
Intrinsically stochastic noise is often assumed to be independent, i.e., noncorrelating.

As stochasticity is inherently part of the \ac{MDP} framework, \ac{RL} algorithms naturally consider stochastic environments. The policies can be trained directly on the real-world environment to account for the real-world uncertainty, cf., Sect.~\ref{sec:rl}. 
A form of stochasticity particularly challenging for \ac{RL} algorithms are rare events, i.e., large but rare deviations from the average obtained cost. Learning value estimates with rare events is particularly difficult~\cite{frank_reinforcement_2008}.\\

\subsubsection{Model mismatch}
%


In practice, the environment for which a policy is designed or trained will often differ from the one it is deployed on.
Thus, it needs to be ensured that the policy will still perform well during deployment.
Standard \ac{RL} algorithms can be biased towards the specific model used during training.
Different strategies like domain randomization, robust \ac{RL}~\cite{morimoto_robust_2000,moos_robust_2022}, and meta \ac{RL} tackle model uncertainty. 
In all of these approaches, the agent is trained not only on a single model but also on a potentially adaptive distribution of models, which can improve the generalization to new models~\cite{finn_model-agnostic_2017}.
The underlying assumption is that the real-world environment lies in the distribution of the training models~\cite{finn_model-agnostic_2017}.
Optimization-based meta \ac{RL} \cite{finn_model-agnostic_2017} learns weights that can quickly adapt to new models. In contrast, in-context meta \ac{RL} uses history-dependent policies to infer the current dynamics of the environment~\cite{brown_language_2020}.

Like in the case of intrinsic stochasticity, both stochastic and robust \ac{MPC} can explicitly take model mismatch into account. For stochastic \ac{MPC}, it is important to consider that predictive errors due to model uncertainty are typically strongly correlated across time.
Note that the model used in a stochastic \ac{MPC} formulation is typically a simplification or approximation of the environment, leading to a mismatch of the stochastic \ac{MPC} model with respect to the environment.
The field of distributionally robust \ac{MPC} aims to robustify against mismatches in the distribution model \cite{van_parys_distributionally_2016}.
Even without considering the model mismatch explicitly, \ac{MPC} can perform remarkably well. This is the property of inherent robustness of \ac{MPC}, as explained in more detail in Sect.~\ref{sec:mpc}.\\

\subsubsection{Stability}
Stability theory is usually not a major concern of RL algorithms since the main objective is closed-loop performance and practical stability instead. However, as denoted in Sect.~\ref{sec:mpc}, asymptotic stability and constraint satisfaction are, or were at least historically, the main focus of \ac{MPC} algorithms and applications.
A widely used tool for analyzing the stability of control problems involves the construction of Lyapunov functions.
If a Lyapunov function exists for a controlled deterministic MDP, it is said to be asymptotically stable, i.e., trajectories converge. If the MDP is stochastic, the concept of input-to-state-stability can be applied, which requires the trajectories of the controlled system to converge to a region around the origin whose magnitude depends on the maximum norm of the noise~\cite{rawlings_model_2017}.

Input-to-state stability and asymptotic stability may not be applied to general \acp{MDP} since they require certain properties of the cost function, such as a feasible origin. Therefore, the authors in~\cite{gros_economic_2022} propose a stability concept named D-stability that applies to \acp{MDP} and generalizes the dissipativity theory of economic MPC.\\

\subsubsection{Constraints}
The ability to account for constraint satisfaction and stability are some of the main reasons MPC is outstanding compared to other control techniques.
In recent years, the \ac{RL} community also has increasingly focused on providing safety guarantees~\cite{brunke_safe_2022}.
A widely adopted framework in \ac{RL} involves modeling constraints using constrained MDPs~\cite{altman_constrained_1999}.
Following~\cite{brunke_safe_2022}, these include: Hard constraints, the constraints are always fullfilled, chance constraints, the constraints are fullfilled with high probability, and constraint violations transformed as accumulated costs.
The latter can be either formulated such that the accumulated costs can not exceed a safety budget or as a penalty integrated into the task objective~\cite{brunke_safe_2022}.
The choice of formulation dictates the strategies employed to address these constraints.
Common approaches include deriving safe action sets~\cite{kalweit_deep_2020}, optimizing a Lagrange formulation with dual gradient descent~\cite{ray_benchmarking_2019}, using trust-region optimization~\cite{achiam_constrained_2017} or control-barrier functions~\cite{cheng_end--end_2019}.
Additionally, using MPC~\cite{wabersich_predictive_2021} to provide safety guarantees related to a nominal model was proposed.
Further work suggests augmenting the \ac{RL} state with Lagrange multipliers \cite{calvo-fullana_state_2024} or with the currently used safety budget \cite{sootla_saute_2022}.
Still, a common problem is that safe \ac{RL} policies become either too conservative or, otherwise, may violate constraints~\cite{zhang_constrained_2024}.
Thus, combining~\ac{MPC} and~\ac{RL} is highly desirable for constraint satisfaction.\\

\subsubsection{Online Computation Time}
A major concern in embedded applications is the maximum online computation time of an algorithm on embedded hardware. The maximum inference time of many \ac{NN} architectures, which are often explicit functions, can usually be tightly bounded. However, the computation time of optimization solvers for \ac{NMPC} problems is often unbounded. In fact, it cannot be guaranteed in general that a meaningful solution, i.e., at least a feasible solution, is returned by the optimization algorithm outside of particular optimization problem classes such as convex \acp{QP}~\cite{richter_computational_2011, arnstrom_real-time_2023}. For NMPC, the primal and, possibly, dual variable initialization of the optimization algorithm is essential for fast online computations. If the variables are sufficiently close to the optimal solution, the local convergence rate is fast. In fact, it can be quadratic or superlinear, depending on different numerical algorithms, cf.~\secref{sec:mpc}. \\

\subsubsection{Offline Computation, Engineering, and Maintenance}
In contrast to their fast inference time, RL algorithms usually require a tremendous amount of training samples, and therefore, training time, even for small-sized~\acp {MDP}. The training may be performed without human intervention. Yet in practice, fine-tuning on the RL hyperparameters may be required~\cite{zhang_importance_2021}. Besides the initial system identification, standard MPC does not require further training time. 
System identification differs from RL training by the target of predicting relevant outputs and possibly states of a system. In contrast, the RL target and the potential internal model focus on closed-loop performance.
In the RL approach, the possible indirect internal model of the environment is learned only for the particular task described by the MDP.
The classical system identification may also use models based on physical models, which makes it possible to adapt the controller to changing environments or requirements. RL may require a whole new training data set for changes in the cost, model changes, or changes in the distribution of the states-space, e.g., an adjusted operating range in the environment.\\

\subsubsection{Generalization}
To some extent, \ac{RL} policies can generalize out of their training distribution but with hardly any predictable behavior and guarantees on the closed-loop performance \cite{nasvytis_rethinking_2024}.
Thus, practitioners often need to ensure that the support of the training distribution covers the state and transition distribution encountered during deployment~\cite{ross_reduction_2011}.
An alternative is to further train the RL policy during deployment ~\cite{song_rapidly_2020}.
If the model used within \ac{MPC} can approximate the real environment well on the whole state space, MPC generalizes well, and safety and performance guarantees can be found. Arguably, the model may generalize well based on available knowledge, starting from first principles and physical or mathematical insight. The knowledge about the model results in a more interpretable generalization and is among the main motivations for MPC in general, learning-based MPC or model-based RL. It may not be possible to evaluate whether a physical model approximates the real environment well on the full state space of the environment. However, the physical explanation may also provide insights into its limitations. \\

\subsubsection{Performance in Practice}
In practice, the expected performance difference between \ac{MPC} and \ac{RL} depends on the environment's specific characteristics, computational resource availability, and the model or data quality.
Both approaches have their strengths and weaknesses, depending on the particular requirements and constraints of the control problem. Several works compare both approaches on specific applications; see Tab.~\ref{tab:literature_comparison}.
The problems differ vastly, from drone racing~\cite{kaufmann_champion-level_2023} to multi-energy environments that involve integer variables~\cite{ceusters_model-predictive_2021}.
The MPC formulations vary depending on the applications.
For high sampling times, fast solvers such as \acados{}~\cite{verschueren_acadosmodular_2022} are used. Instead, for combinatorial problems, computationally highly demanding mixed integer solvers, such as \gurobi{}~\cite{gurobi_optimization_llc_gurobi_2023}, are required.
Most of the authors use the same state space for MPC and RL, despite the RL's ability to cope with arbitrary information inputs, which is exploited only in~\cite{kaufmann_champion-level_2023,oh_quantitative_2024}.
For a known model, most authors report superiority of \ac{MPC} w.r.t. the closed-loop performance, with less than~$4\%$ cost reduction in~\cite{ernst_reinforcement_2009, brandi_comparison_2022} and~$16\%$ in \cite{wang_comparison_2023} when compared to RL. In~\cite{mamedov_safe_2024}, the authors report that MPC outperforms RL on a flexible robot manipulation task with a rather high-dimensional state space. RL was claimed to be superior to MPC for environments with poor models, e.g., in~\cite{hasankhani_comparison_2021, ceusters_model-predictive_2021}, yet robust control techniques are not implemented. The authors in~\cite{kaufmann_champion-level_2023} claimed superiority of a particular choice of RL algorithms against MPC in \ac{RWE} of racing drones.

\begin{table*}
\caption{Literature that compares MPC with \ac{RL} on practical applications. The table lists whether \acf{RWE} were performed and which \ac{RL} and \ac{MPC} algorithms were used for the experiments.}
	\centering
	\resizebox{\textwidth}{!}{%
	\begin{tabular}{@{}lllllll@{}}
		\toprule
		\multicolumn{7}{c}{Application-oriented comparison between \ac{MPC} and \ac{RL}}\\
		\midrule
		Ref. & Authors &Year & Application & \ac{RWE} & \makecell[l]{\ac{MPC} Formulation\\ (Algorithm, Solver)} & \ac{RL} Algorithm\\
		\midrule
\cite{ernst_reinforcement_2009} & \etal{Ernst} &2009 & power system & No & NMPC (\ac{IP}, \na) & fitted Q iteration~\cite{ernst_iteratively_2003}   \\
\cite{ceusters_model-predictive_2021} & \etal{Ceusters} &2021 & multi-energy system & No & MILP-MPC (\texttt{cplex}~\cite{cplex_v12_2009}) & TD3, \ac{PPO}   \\
\cite{hasankhani_comparison_2021} & \etal{Hasankhani} &2021 & ocean current turbine & Yes & \na{}& \ac{DQN} \\
\cite{lin_comparison_2021} & \etal{Lin} &2021 & cruise control & No & \ac{NMPC} (IP, \ipopt{}~\cite{wachter_implementation_2006}) & \ac{DDPG} \\
\cite{brandi_comparison_2022} & \etal{Brandi} &2022 & thermal energy system & Yes & MILP (\na{}) & \ac{SAC}   \\
\cite{di_natale_lessons_2022} & \etal{Di Natale} &2022 & building temperature control & Yes&  \makecell[l]{GP-MPC~\cite{hewing_cautious_2020},\\Bilevel DeePC~\cite{coulson_data-enabled_2019},\\(custom)} & TD3    \\
\cite{dobriborsci_experimental_2022} & \etal{Dobriborsci} & 2022 & mobile robot & Yes & NMPC  (SLSQP~\cite{noauthor_scipy_nodate}) & \ac{DQN}   \\
\cite{byravan_evaluating_2022} & \etal{Byravan} & 2022 & MuJoCo locomotion & No & \makecell[l]{Sequential Monte \\Carlo~\cite{piche_probabilistic_2018}, CEM} & MPO~\cite{abdolmaleki_maximum_2018}   \\
\cite{wang_comparison_2023} & \etal{Wang} &2023 & building temperature control & No &  \ac{NMPC} (IP, \ipopt{}~\cite{wachter_implementation_2006}) & \ac{SAC}, \ac{DDPG}, DDQN  \\
\cite{song_reaching_2023} & \etal{Song} &2023 & drone racing & Yes &  \ac{NMPC} (SQP, \acados{}\cite{verschueren_acadosmodular_2022}) & \ac{PPO}\\
\cite{shi_model-based_2023} & \etal{Shi} &2023 & vehicle parking & No &  \ac{NMPC} (\na) & \ac{PPO}\\
\cite{imran_comparison_2023} & \etal{Imran} &2023 & traffic control & No & MIQP-MPC (\gurobi{}~\cite{gurobi_optimization_llc_gurobi_2023}) & \ac{DQN} \\%
\cite{reiter_hierarchical_2023} & \etal{Reiter} &2023 & autonomous racing & No & \ac{NMPC} (SQP, \acados{}\cite{verschueren_acadosmodular_2022}) & \ac{SAC}\\ 
\cite{morcego_reinforcement_2023} & \etal{Morcego} &2023 & greenhouse climate control & No & \ac{NMPC} (IP, \ipopt{}~\cite{wachter_implementation_2006})  & \ac{DDPG} \\
\cite{mamedov_safe_2024} & \etal{Mamedov}&2024 & flexible robot manipulation & No & \ac{NMPC} (SQP, \acados{}\cite{verschueren_acadosmodular_2022}) & \ac{PPO}, \ac{SAC}  \\
\cite{hoffmann_comparison_2024} & \etal{Hoffmann} &2024 & vehicle lane change & No &  \na{} & \ac{PPO}, Q-learning  \\
\cite{oh_quantitative_2024} & Oh &2024 & chemical and biological processes & No &  \makecell[l]{LMPC (\na{}),\\ \ac{NMPC} (IP, \ipopt{}~\cite{wachter_implementation_2006})} & \ac{SAC}, \ac{DDPG}, TD3  \\
		\bottomrule
	\end{tabular}}
	\label{tab:literature_comparison}
\end{table*}

\section{Combination Approaches}
\label{sec:comb_arch}

As pointed out in the previous sections, \ac{RL} is a collection of algorithms to learn an optimal policy for \acp{MDP} by interacting with the environments. \Ac{MPC}, in contrast, refers to a mathematical program that implicitly approximates the \ac{MDP}. 
The two approaches exhibit different advantages. Therefore, a combination of both is appealing.

We propose a categorization of combination approaches that distinguishes in which algorithmic part of \ac{RL} the \ac{MPC} framework is used. The related general \ac{RL} building blocks are the actor, the critic, and possibly an expert actor, c.f., Fig.~\ref{fig:structure1}. Accordingly, we propose the categories
\begin{enumerate}
    \item MPC as an expert actor, cf. Sect.~\ref{sec:comb_expert},
    \item MPC within the deployed policy, cf. Sect.~\ref{sec:mpc_as_actor},
    \item MPC as part of the critic, cf. Sect.~\ref{sec:comb_critic}.
\end{enumerate}
In the following, we shortly outline the different categories.\\

\begin{figure}
	\includegraphics[width=.49\textwidth]{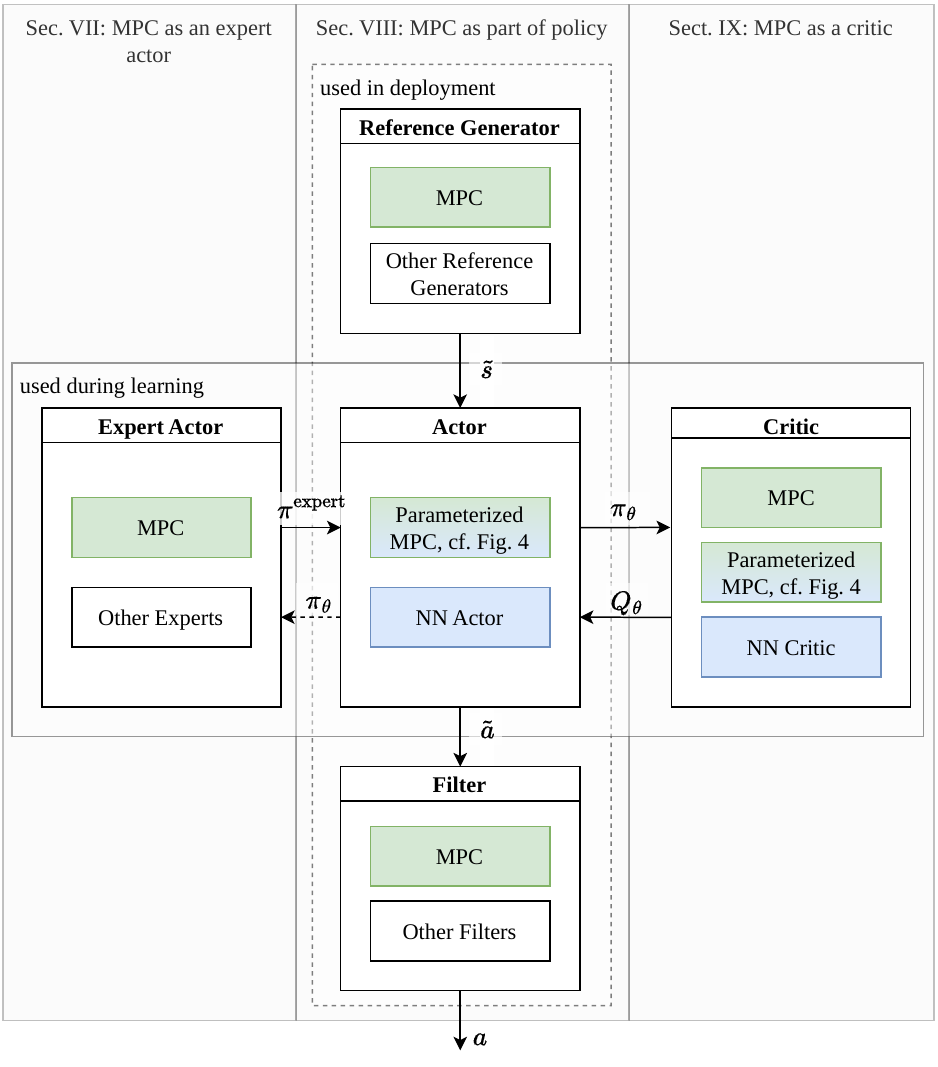}
	\caption{Modular view on combinations of MPC and RL. The combinations are aligned with the sections of this survey.
    The horizontal separation corresponds to using MPC as part of the expert actor, the deployed policy, and the RL critic. This overview highlights the possibility of using several instances of MPC with different roles in various parts of an RL algorithm and a deployed policy.
    \ac{MPC} is used with fixed expert parameters within the expert actor, as a reference generator, as a postprocessing filter, or, possibly, in the critic. Learned \ac{MPC}, i.e., an MPC structure involving learned parameters, can be used within the learned actor or possibly within the learned critic.}
	\label{fig:structure1}
\end{figure}

\textbf{MPC as an expert actor.}
An \ac{MPC} with fixed parameters and the desired behavior can be used as an expert to learn policies. 
The expert behavior can either be used by an \ac{IL} algorithm that tries to purely mimic the behavior or to guide the exploration process during \ac{RL} training~\cite{levine_guided_2013}. MPC used as an expert is elaborated in Sect.~\ref{sec:comb_expert}.\\

\textbf{MPC within the deployed policy.} Another common variant of how \ac{MPC} and \ac{RL} are combined is using MPC within the deployed policy. A parameterized \ac{MPC} can be used as part of the learned actor, as, for instance, proposed by~\cite{gros_economic_2022}. Closely related are concepts that use MPC as a posterior filter or reference provider as an element of the deployed policy but are not used during learning. For instance, the authors in \cite{wabersich_predictive_2021} use MPC as a safety filter and~\cite{jenelten_dtc_2024} use MPC as a reference provider, cf. Fig.~\ref{fig:structure1}. All concepts that use MPC within the deployed control policy are elaborated in Sect.~\ref{sec:mpc_as_actor}.\\

\textbf{MPC as part of the critic.} An \Ac{MPC}, possibly parameterized with learnable parameters, can further be used to evaluate the value function at a given state, i.e., the critic uses an \ac{MPC} variant. For instance, the in~\cite{ghezzi_imitation_2023} uses an MPC with fixed parameters to compute the value at a given state. This structure is reviewed in Sect.~\ref{sec:comb_critic}. \\

Sections \ref{sec:comb_expert} to \ref{sec:comb_critic} describe combination approaches of MPC and RL conceptually, leaving out a detailed theoretical discussion. Nonetheless, the additional theory Sect.~\ref{sec:theory} highlights some important findings and relevant literature when combining the two approaches.\\

To the best of the authors' knowledge, in the following review, the most important works that combine \ac{MPC} and \ac{RL} can be classified according to the proposed categories in Fig.~\ref{fig:structure1}. 
The works are further compared related to their application, whether \acp{RWE} on embedded hardware was performed, which \ac{MPC} type and solver the authors used, and which \ac{RL} algorithm the authors applied.
While the MPC algorithm classification is more distinct, the classification of the RL algorithm can be ambiguous and intertwined with the overall presented algorithm. 
Besides widespread RL algorithms such as SAC and PPO, several forms of \ac{DP}, such as \ac{PI} or \ac{VI}, forms of \ac{IL} such as \ac{BC}, and even supervised learning for learning value functions, such as \ac{MLE}, are denoted as RL algorithm in the tables.
Applications include \acf{UGV}, temperature control, energy systems, chemical processes, traffic management, autonomous racing with vehicles or drones, and \acf{AD}.

\subsection{Architectures of Parameterized MPC}
\label{sec:architectures}

Before diving into the particular approaches of how MPC is used as part of RL or IL algorithms, a short classification of parameterized MPC architectures that use \acp{NN} is given.
This distinction is particularly important for designing combination approaches as it centrally governs the desired properties and choices of algorithms.

We consider an MPC optimization layer that can be parameterized by a parameter~$\phi\in\R^{n_\phi}$ and may depend on the decision variables~$z$ of the MPC optimization problem and the current state~$s$  through a highly nonlinear \ac{FA}~$\phi=\varphi_\theta(s,z)$, e.g., an \acf{NN} with the (very) high dimensional parameter vector~$\theta\in\R^{n_\theta}$. 
Since \acp{NN} are the most common \acp{FA}, we consecutively mainly write \ac{NN} but also comprise other forms for \acp{FA}.
When fixing the parameters~$\phi$, the MPC optimization problem is assumed to contain only optimization-friendly objective, model, and constraint functions.

A general form of the parameterized MPC, based on~\eqref{eq:mpc-problem}, can be written as
\begin{mini!}|s|
	{z }
	{
		L_\phi(z)\label{eq:mpc_main_cost}
	}
	{\label{eq:mpc_hierarchical}}
	{}
	\addConstraint{x_0}{= s \label{eq:mpc_main_c_initial1}}
	\addConstraint{x_{k+1}}{= f^\mpcShrt_\phi\big(x_k, u_k\big),\,\;}{0 \leq k < N\label{eq:mpc_main_c_dyn1}}
	\addConstraint{0}{\leq h^\mpcShrt_\phi\big(x_k, u_k\big),\,\;}{0 \leq k < N \label{eq:mpc_main_c_const1}}
	\addConstraint{0}{\leq \tilde{h}_\phi^\mpcShrt\big(x_N\big),}{}
\end{mini!}
and the parameterized objective is
\begin{equation*}
	L_\phi(z):=\bar{V}_\phi^\mpcShrt(x_{N}) +\sum_{k=0}^{N-1} l_\phi^\mpcShrt\big(x_k, u_k\big)  .
\end{equation*}

For easing the notation, all constraints of~\eqref{eq:mpc_hierarchical} are summarized by~$g_\phi(z;s)\geq0$. Thus, the optimization problem~\eqref{eq:mpc_hierarchical} can be concisely written as
\begin{equation}\label{eq:integrated architecture}
	\min_{z} L_\phi(z)\quad\text{s.t.}\quad g_\phi(z;s)\geq0.
\end{equation}
Note that we use the notation $g_\phi(z;s)$ to indicate that~$z$ are decision variables of the optimization problem and~$s$ are parameters.
The following architectures are proposed within this context, starting with the most general one, i.e., the integrated architecture, cf. Fig~\ref{fig:structure2}.\\

\begin{figure}
	\includegraphics[width=\linewidth,trim={0 0 0 0},clip]{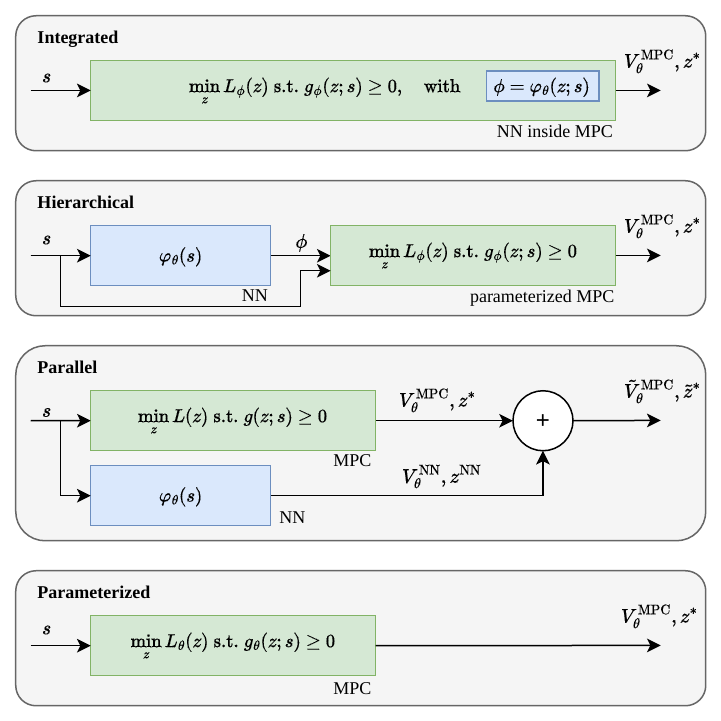}
	\caption{Parameterized MPC architectures: Proposed architectures of actors (or potentially critics) used in \ac{RL} utilizing \acp{MPC} and \acp{NN}/simple parameters. In the \integrated architecture, the \ac{NN} is part of the optimization layer and depends on the decision variables. Suppose an \ac{NN} can be evaluated separately from the optimization routine but provides the parameters to an MPC optimization layer. In that case, the architecture is referred to as \hierarchical.
    If the MPC problem is solved in parallel to the NN, we refer to a parallel architecture.
    If the learned parameters do not depend on the current state~$s$, the architecture is referred to as \paramterized{}}.
	\label{fig:structure2}
\end{figure}

\subsubsection{Integrated Architecture}
In the \integrated{} setting, the parameters depend on both the state~$s$ and the optimization variables~$z$ via the function~$\nnfun_{\theta}(s,z)$, as, for instance, within the algorithm proposed by~\cite{lowrey_plan_2019}.
Therefore, the \ac{NN} is part of the optimization problem and can not be evaluated separately in the inference path, leading to a highly nonlinear and, possibly, nonconvex optimization problem, which may be computationally challenging. 
Particularly, when using derivative-based solvers, this involves differentiating the highly nonlinear parameterized function~$\nnfun_{\theta}(z;s)$ w.r.t. the decision variables~$z$ by $\nabla_z\nnfun_{\theta}(z;s)$ as part of solving the MPC optimization problem.
Note that the optimization layer may provide any of the optimal decision variables~$z^\star$ or the MPC objective function~$V^\mathrm{MPC}_\theta$ of the MPC optimization problem.\\

\subsubsection{Hierarchical Architecture}
In the \hierarchical{} architecture, the neural network provides an input to the MPC. Therefore, the parameter~$\phi$ depends on the current states~$s$ via the highly nonlinear function~$\phi=\nnfun_{\theta}(s)$.
For example, a reference trajectory parameterized by~$\phi$ could be provided in case this architecture is used as a policy, such as in~\cite{brito_where_2021, reiter_hierarchical_2023}, or $\phi$ could parameterize constraints in the \ac{OCP}, similar to~\cite{jacquet_n-mpc_2024}.
Notably, the function~$\nnfun_{\theta}(s)$ can be evaluated before solving the optimization problem. In this case, the potential nonlinearity of~$\nnfun_{\theta}(s)$ does not influence the numerical optimization problem structure, i.e., the optimization problem does not get harder to solve, despite the changing parameters. \\

\subsubsection{Parallel Architecture}
In the parallel architecture, the MPC problem is usually not parameterized. However, an NN is evaluated in parallel to the MPC optimization problem, and its output is used to correct the optimal solution for decision variables or the value function of the MPC. This architecture holds the advantage of not requiring differentiating the optimization problem and the disadvantage of potentially unsafe actions due to the perturbation of the MPC actions. An example of this architecture used to approximate the value function can be found in~\cite{bhardwaj_blending_2021}.\\

\subsubsection{Parameterized Architecture}
In the parameterized MPC architecture, the parameters~$\phi$ are constant and independent of decision variables~$z$ or the state~$s$. 
Conceptually, a parameterized optimization-friendly MPC problem is formulated that does not require the evaluation of highly nonlinear functions to obtain the paramter~$\phi$.
For instance, such a parameterization is proposed in~\cite{gros_data-driven_2020,moradimaryamnegari_model_2022} and \cite{liu_learning_2023}.\\

Above, we introduced a number of architectures for parameterizing MPC with potentially highly nonlinear functions such as \acp{NN}. These architectures may be used in the following MPC and RL combinations, whereby some parameterized MPC architectures are more or less favored in particular MPC and RL combinations.

\section{MPC as an Expert Actor}
\label{sec:comb_expert}
For some problems, it is possible to design an expert \ac{MPC} that achieves good closed-loop performance.
A thorough collection of relevant literature can be found in Tab.~\ref{tab:literature_expert} and a sketch in Fig.~\ref{fig:sketches_expert_actor}.
\begin{figure}
	\includegraphics[width=\linewidth]{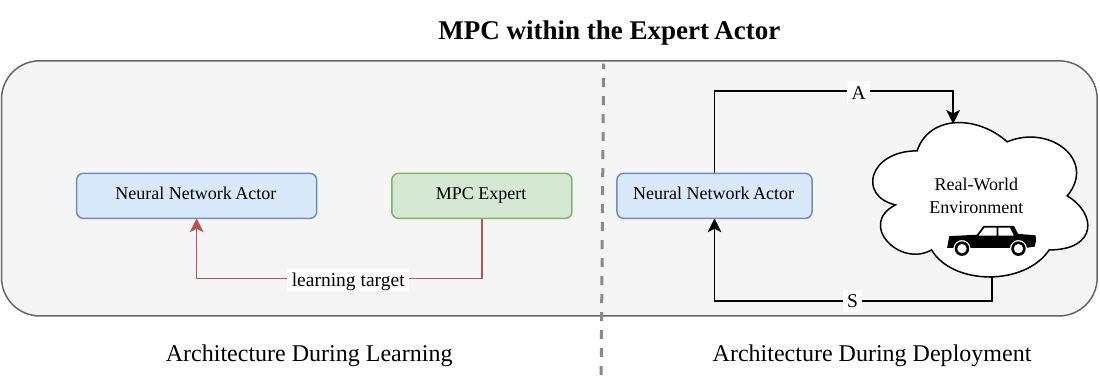}
	\caption{Combinations: MPC as an expert actor. The plot is split into the learning and deployment phases. Blue boxes indicate Neural Networks (NNs), and green boxes are used for MPCs.}
	\label{fig:sketches_expert_actor}
\end{figure}
There are multiple ways to exploit such an expert motivating the structure of the following section:

\begin{enumerate}[label=\Alph*.]
    \item \textbf{Imitation learning from \ac{MPC}}, Sect.~\ref{sec:imitation}. Here, the goal is to replace the expert \ac{MPC} with a \ac{NN} to achieve faster computation time using imitation learning.  
    \item \textbf{Guided Policy Search using \ac{MPC}}, Sect.~\ref{sec:gps}. In this paradigm, the expert \ac{MPC} policy improves the exploration process during the \ac{RL} training guiding the learned policy to low-cost regions.
\end{enumerate}
Besides the explored application, the table specifies whether the algorithm is related to \acf{IL} or \ac{GPS} and which particular MPC, RL, and IL algorithms are used as a basis.

\begin{table*}
	\caption{Literature that uses \ac{MPC} as an expert actor for \ac{RL}. The table lists the applications and whether \acf{RWE} were performed. Additionally, the \ac{IL}/\ac{RL} and \ac{MPC} algorithms are stated. note that \ac{RL} algorithms can be used to perform \ac{IL}, e.g., \cite{kang_rl_2023}. The table further differentiates whether the expert adapts to the approximation error of the learned policy via \acf{GPS}}
	\centering
	\resizebox{\textwidth}{!}{%
	\begin{tabular}{@{}llllllll@{}}
		\toprule
		\multicolumn{7}{c}{\ac{MPC} as an Expert Actor}\\
		\midrule
		Ref. & Authors & Year & Application & \ac{RWE} & \makecell[l]{\ac{MPC} Formulation\\ (Algorithm, Solver)} & \ac{IL}/\ac{RL} Algorithm & \ac{GPS} \\
		\midrule
\cite{levine_guided_2013} & \etal{Levine} & 2013 & MuJoCo locomotion & no & NMPC (custom \ilqr{}) & Off-policy \ac{PG} & yes\\
\cite{levine_variational_2013} & \etal{Levine} & 2013 & MuJoCo locomotion & no & NMPC (custom \ilqr{})& Var. Policy Search \cite{neumann_variational_2011} & yes \\
\cite{mordatch_combining_2014} & \etal{Mordatch} & 2014 & MuJoCo locomotion & no & NMPC (custom \ilqr{}) & \ac{BC}/\ac{ADMM} & yes \\
\cite{levine_end--end_2016} & \etal{Levine} & 2016 & \makecell[l]{PR2 robot arm\\ robot manipulation} & yes & NMPC (custom \ilqr{}) & \ac{BC}/Bregman \ac{ADMM} \cite{wang_bregman_2014} & yes \\
\cite{sun_fast_2018} & \etal{Sun} & 2018 & autonomous driving & no & NMPC (IP, \ipopt{}~\cite{wachter_implementation_2006}) & DAgger~\cite{ross_reduction_2011} & no \\
\cite{nagabandi_neural_2018}          &\etal{Nagabandi} & 2018 & MuJoCo locomotion & no & \makecell[l]{NMPC, random\\ sampling (custom)} & DAgger~\cite{ross_reduction_2011}, TRPO & yes \\
\cite{wang_exploring_2019} & \etal{Wang} & 2019 & MuJoCo locomotion & no & \ac{CEM} & \ac{BC}, \ac{IL}& no \\
\cite{pinneri_extracting_2021} & \etal{Pinneri} & 2021 & MuJoCo locomotion & no & \ac{CEM} & DAgger~\cite{ross_reduction_2011} & no\\
\cite{sacks_learning_2022} &\etal{Sacks}& 2022 & robot arm & no &  \ac{MPPI} (custom) &  DAgger~\cite{ross_reduction_2011} & no\\
\cite{dawood_handling_2023}  & \etal{Dawood} & 2023 & Kuboki Turtlebot 2 & yes &
NMPC (SQP, \acados{}\cite{verschueren_acadosmodular_2022}) & \ac{SAC}, \ac{TD3}& no\\
\cite{kang_rl_2023} & \etal{Kang} & 2023 & \makecell[l]{Unitree Go1 and \\Aliengo (quadruped)} & yes & NMPC (\na) & \ac{PPO}& no\\
\cite{ahn_model_2023} & \etal{Ahn} & 2023 & random linear system & no & LMPC (\na) & DAgger~\cite{ross_reduction_2011}& no \\
\cite{lidec_enforcing_2023} & \etal{Le Lidec} & 2023 & robot arm & no & NMPC (FDDP~\cite{mastalli_crocoddyl_2020}) & \ac{BC}/\ac{ADMM}& yes \\
\cite{mamedov_safe_2024} & \etal{Mamedov} & 2024 & flexible robot arm & no & NMPC (SQP, \acados{}~\cite{verschueren_acadosmodular_2022})& \makecell[l]{BC, DAgger~\cite{ross_reduction_2011},\\AIRL~\cite{fu_learning_2018}, GAIL~\cite{ho_generative_2016}}& no\\
\cite{hoffmann_plannetx_2024} & \etal{Hoffmann} & 2024 & autonomous driving & no & LMPC (SQP, \acados{}~\cite{verschueren_acadosmodular_2022}) & custom \ac{IL} & no \\
\cite{schulz_learning_2024} & \etal{Schulz} & 2024 & cart pole, spacecraft & no & NMPC (SQP, \acados{}~\cite{verschueren_acadosmodular_2022}) & \ac{TD3} & yes \\
\bottomrule
	\end{tabular}}
 \label{tab:literature_expert}.
\end{table*}

\subsection{Imitation Learning from MPC}
\label{sec:imitation}
In a problem setting where MPC achieves a sufficient closed-loop performance burdened only by its online computation time, a trained \ac{NN} may replace the \ac{MPC}. The typically fast inference of the NN may drastically decrease the online computation time by omitting the time-consuming online optimization.
In fact, in many real-world applications~\cite{mamedov_safe_2024,dawood_handling_2023,kang_rl_2023}, the computation time of the \ac{MPC} solver is a significant limitation and potentially even unbounded, particularly for more complex structures, e.g., involving nonlinear systems, stochasticity or discrete decisions, c.f.~Sect.~\ref{sec:mpc_algo}.
Two research directions try to alleviate the problem of large online computation times. 

First, for linear systems, linear constraints, and convex quadratic costs, the optimal feedback control law is piece-wise linear within a polytope of the state space \cite{bemporad_explicit_1999}. 
Within explicit \ac{MPC} \cite{bemporad_model_2002,grancharova_explicit_2012}, these polytopes and their control laws are computed offline and switched during deployment by determining the active region.
For large state spaces or a large number of constraints, explicit \ac{MPC} becomes quickly intractable.
Thus, approximate explicit \ac{MPC} only approximately stores the control law of the \ac{MPC} policy \cite{johansen_approximate_2003}.

A learning-based approach to approximate explicit \ac{MPC} is using \ac{IL} to learn the \ac{MPC} expert.
Methods can learn the whole trajectory predictions~\cite{vaupel_accelerating_2020} or learn the first applied action~\cite{hertneck_learning_2018}.
Early works \cite{akesson_neural_2006, parisini_receding-horizon_1995} and several follow-up works \cite{hertneck_learning_2018, cao_deep_2020, moraes_neural_2020, chen_approximating_2018, karg_efficient_2020, pinneri_extracting_2021} are aiming to imitate a complex \ac{MPC} with a neural network, by training the \ac{NN} with supervised learning methods.
Further aspects of the learned \ac{NN} policy where analyzed regarding safety \cite{cosner_end--end_2022, mamedov_safe_2024}, stability \cite{schwan_stability_2023} and robustness \cite{hertneck_learning_2018}.
Another important aspect is the verification of the learned \ac{NN}, where the object of interest is the worst-case approximation error to the \ac{MPC} expert.
This can be done in a probabilistic fashion using Hoeffdings inequality, such as in \cite{hertneck_learning_2018, drgona_learning_2024}, where confidence bounds are derived for an estimate of how often an approximation error threshold might be exceeded.
Alternatively, verification techniques like in \cite{schwan_stability_2023} solve an \ac{MIQP} to bound the worst-case approximation error.
For a general overview of \ac{NN} verification, see \cite{everett_neural_2021}.

Often, a surrogate loss function as in standard \ac{BC} minimizes the squared distances between the predicted action of the \ac{NN} and the action of the \ac{MPC} expert. However, it disregards the costs and structure of the underlying \ac{OCP} while learning the policy.
Thus, a natural replacement is to use instead the Q-function of the \ac{MPC}, cf. \eqref{eq:mpc-q-problem} as a learning objective \cite{carius_mpc-net_2020, reske_imitation_2021, ghezzi_imitation_2023}.
Since the $\ac{MPC}$ provides a Q-value function in addition to the expert policy, cf. Fig.~\ref{fig:structure1}, we categorize such methods in a class named \ac{MPC} as a critic, as further discussed in Sect.~\ref{sec:comb_critic}.

\subsection{Guided Policy Search using MPC}
\label{sec:gps}
In the previously mentioned \ac{IL} methods, the \ac{MPC} expert does not consider the progress of the learned policy during training.
In the \acf{GPS} regime, the expert gradually guides the learned policy to better trajectories \cite{levine_guided_2013, levine_variational_2013, mordatch_combining_2014, levine_end--end_2016}.
This is ensured by coupling the trajectory optimization and policy learning by requiring that the \ac{MPC} expert does not deviate too far from the current predictions of the learned policy \cite{mordatch_combining_2014}.
The final learned policies can generalize even when the \ac{MPC} expert fails to find a solution~\cite{levine_end--end_2016}.
Further, the authors in~\cite{mordatch_combining_2014} have shown superior performance over policies learned from a fixed dataset of expert trajectories.

Another possibility to guide the exploration process for \ac{RL} was proposed in \cite{schulz_learning_2024} and is related to the options framework \cite{sutton_between_1999}.
A learned high-level exploration policy decides between using an \ac{MPC} expert policy or the currently learned low-level \ac{RL} policy.
The usage of the \ac{MPC} expert is regularized over time to enforce that the learned low-level \ac{RL} policy works as a stand-alone policy during deployment.
A benefit over \ac{GPS} is that the high-level policy can avoid using the \ac{MPC} expert for state regions, where the \ac{MPC} expert guidance is suboptimal.

\section{\Ac{MPC} within the Deployed Policy}
\label{sec:mpc_as_actor}
This section discusses methods that use \ac{MPC} during deployment in the real environment.
A primary distinction is drawn on whether parameterized MPC is trained by RL algorithms or if MPC is used for pre/postprocessing after the RL training. 
When learning a parameterized MPC, we furthermore distinguish between approaches that aim to align the MPC formulation with the MDP structure, e.g., to learn an internal MPC model that aims to approximate the real environment as closely as possible or methods that primarily focus on closed-loop optimality. 
Accordingly, we structure this section as follows
\begin{enumerate}[label=\Alph*.]
\item \textbf{Aligned-learning}, Sect.~\ref{sec:mpc_as_actor_aligned}. In the paradigm of aligned learning, the \ac{MPC} structure~\eqref{eq:mpc_from_mdp} approximates the real-world \ac{MDP} structure~\eqref{eq:mdp}, i.e., it uses a transition model, costs, constraints and a terminal value function that individually approximate the \ac{MDP} and the optimal value function~$V^\star$, respectively.
\item \textbf{Closed-loop learning}, Sect.~\ref{sec:mpc_as_actor_closed_loop_learning}. In the paradigm of closed-loop learning, an MPC is used as an optimization layer within the actor policy and trained within an RL algorithm for closed-loop optimality, which, in general, does not require the individual parts of the MPC to be aligned with the MDP. For instance, the MPC model should not necessarily fit the real system expected or most likely transition to provide closed-loop optimality~\cite{kordabad_equivalence_2024}.
\item \textbf{MPC for pre/postprocessing}, Sect.~\ref{sec:prepostprocessing}. Using MPC for pre/postprocessing does not involve MPC during learning but as part of the deployed policy. This section is further split into reference generation and filtering.
\end{enumerate}

Different structures that include \ac{MPC} in the deployed policy are proposed; see Fig.~\ref{fig:structure1} and Fig.~\ref{fig:sketches_in_policy}.
\begin{figure}
	\includegraphics[width=\linewidth]{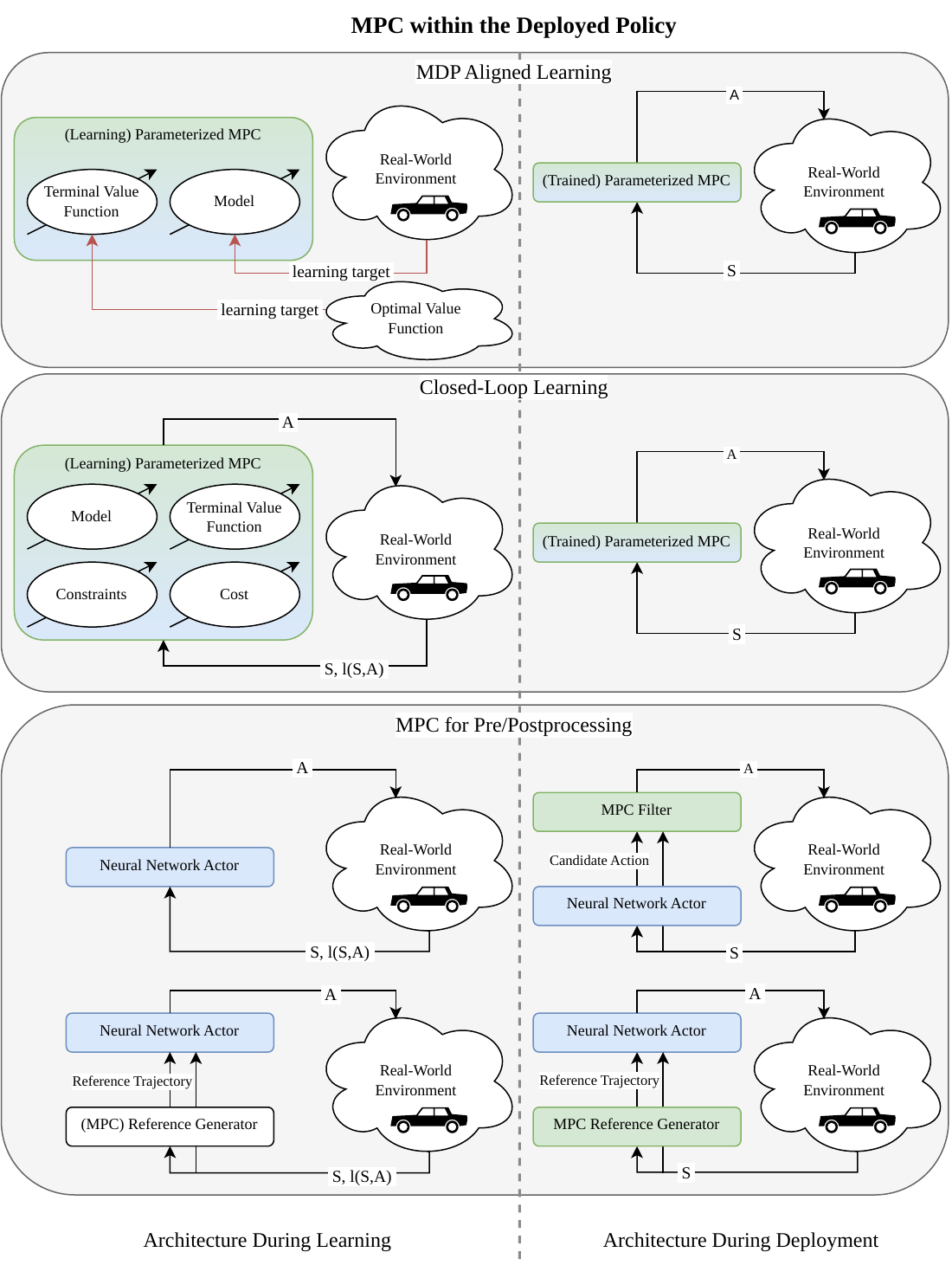}
	\caption{Combinations: MPC within the deployed policy. The plot is split into the learning and deployment phase. Blue boxes indicate Neural Networks (NNs), green boxes are used for MPCs, and blue/green boxes refer to parameterized MPCs that involve parameters/NNs that are learned during the learning phase, see Sect.~\ref{sec:architectures}.}
	\label{fig:sketches_in_policy}
\end{figure}

\subsection{Aligned Learning}
\label{sec:mpc_as_actor_aligned}

The classical design procedure of \ac{MPC} focuses in its first step on finding an usually deterministic mathematical model~$f^{\mathrm{MPC}}_\theta(x,u)$ assumed to be parameterized by~$\theta$ that describes the environment sufficiently well~\cite{ljung_system_1999}. Secondly, a horizon length~$N$ and a terminal value function~$\bar{V}^\mathrm{MPC}_\theta(s)$ as in~\eqref{eq:mpc_from_mdp} with parameters~$\theta$ are obtained to approximate the optimization problem over a finite horizon. 

The advantage of aligning the MPC with the MDP can be seen in the well-interpretable formulation, adaptability to new problems, generalization, potential guarantees for stability and recursive feasibility, and the division of the overall design into subtasks. However, the online computation time may high if complex MPC formulations are used, e.g., stochastic MPC formulations that account for uncertainty.

Obtaining the terminal value function to approximate the infinite horizon value function as in~\eqref{eq:mpc_from_mdp} is computationally challenging and, therefore, usually approximated and often obtained through learning algorithms closely related to the \ac{RL} framework.
 For simple MDPs, e.g., deterministic MDPs with linear models and a quadratic cost, the terminal value function can be computed exactly and does not require learning or approximation.
The costs~$l^\mathrm{MPC}(x,u)$ are usually given by the application, i.e., the MDP, but may be altered to improve the numerical properties related to the optimization problem regarding smoothness and convexity.
We denote parameters of both the model~$f^{\mathrm{MPC}}_\theta(x,u)$ and the terminal value function~$\bar{V}^\mathrm{MPC}_\theta(s)$ by~$\theta$. 

Literature that can be categorized as \MDPAligned{} typically uses the integrated or parameterized architecture according to Fig.~\ref{fig:structure2}. For instance, the paramters~$\theta$ could parameterize a \ac{NN} that is used as the internal MPC model~\cite{nagabandi_neural_2018} or the terminal value function~\cite{lowrey_plan_2019}, corresponding to the integrated architecture.
Tab.~\ref{tab:literature_actor_aligned} summarizes literature that can be related to the \MDPAligned{} learning paradigm. 
In the comparison of Tab.~\ref{tab:literature_actor_aligned}, a distinction is made, whether the model or the terminal value function was learned.
Additionally, the table indicates if the corresponding object was learned during the deployment of the continually improving MPC policy, referred to as on-policy learning (ON-P), or learned by using a separate off-policy sampling strategy (OFF-P). An example of an off-policy sampling strategy would involve system identification to identify a model before deploying the MPC.

In the following, the two objectives of learning the model or the terminal value function are explained in further detail.\\

\subsubsection{Learning the Model}
The first target of aligned learning is to approximate the stochastic model of the real environment~$P$ with a model~$f_{\theta}^\mathrm{MPC}$ that is used within \ac{MPC}.
The model is usually a parameterized mathematical model designed via first principles or a more generic model using a function approximator like a \ac{NN} or \ac{GP}.
Parameters~$\theta$ of candidate models are fit to observed transition data~$\mathcal{D}^\mathrm{id}=\big\{(S_0, A_0,S_1), \ldots, (S_{M-1}, A_{M-1},S_M)\big\}$ of state transitions and related actions. The evaluation of the model fit requires a validation or loss function~$\mathcal{L}^\mathrm{id}(\cdot)$, which could be, for instance, the least-square loss function.
The parameters can be found by minimizing
\begin{equation}
	\label{eq:supervised_loss}
    \min_{\theta} \E_{(S,A,S^+) \sim \mathcal{D}^\mathrm{id}} \biggl[\mathcal{L}^\mathrm{id}\big(S^+-f_{\theta}^\mpcShrt(S,A)\big) \biggr]. 
\end{equation}

\subsubsection{Learning the Terminal Value Function}
In order to approximate the MDP on a finite MPC horizon as in~\eqref{eq:mpc_from_mdp}, a terminal value function~$\bar{V}^\mathrm{MPC}_\theta(s)$ with parameters~$\theta$ needs to be established.
From a system theoretical perspective, the terminal value function is important in terms of stability or recursive feasibility, c.f., Sect.~\ref{sec:mpc}.
However, when focusing on closed-loop cost, the terminal value function needs to tractably approximate the true value function $V^\star$ of the MDP.
In the end, different aspects can be considered when learning an appropriate terminal value function.

For instance, the terminal value function~$\bar{V}_\theta^\mathrm{MPC}$ can be learned via a temporal difference update, cf. \cite{lowrey_plan_2019}, where a learned terminal value function is incorporated into an \ac{MPPI} planner.
In the following, a simplified version is provided to learn a parameterized terminal value function~$\bar{V}_\theta^\mathrm{MPC}$, potentially a \ac{NN}, that can be defined via
\begin{algbox}{Terminal value function learning $(\bar{V}_\theta^\mathrm{MPC} \approx \; ?\footnotemark)$}
\begin{align*}
    &\theta \leftarrow \theta + \frac{\alpha}{B} \sum_{i=1}^B \delta_i \nabla_\theta \bar{V}_\theta^\mathrm{MPC} (S_i),\\
    &\delta_i \coloneqq l(S_i, A_i) + \bar{V}_\theta^\mathrm{MPC} (S^+_i) - \bar{V}_\theta^\mathrm{MPC} (S_i),\\
    &\textrm{with } S_i \sim \mathcal{D}^{\pi^\mathrm{MPC}_\theta},\; A_i = \mu_\theta^\mathrm{MPC} (S_i),\; S^+_i \sim {P(\cdot | S_i, A_i)},
\end{align*}
\end{algbox}
\footnotetext{The question mark indicates that, to the best of the authors' knowledge, the update scheme does not converge to the optimal value function $V^\star$ in general, and its convergence target remains unclear.}
where $\mathcal{D}^{\pi^\mathrm{MPC}_\theta}$ is a distribution of states generated by controlling the system via the parametrized \ac{MPC} exploration policy~$\pi^\mathrm{MPC}_\theta$ and $\alpha$ is the learning rate.
Another approach to learning the terminal value function is shown in \cite{reiter_ac4mpc_2024}, where the authors use either model-free \ac{SAC} or \ac{PPO} to obtain a terminal value function for an \ac{MPC} planner.



\begin{table*}
\caption{Literature that uses MDP-aligned learning. Either the value function, the model, or the cost function of an \ac{MPC} actor are learned. The table lists the applications and whether \acf{RWE} were performed and which \ac{MPC} formulation was used. Moreover, the learning or approximation algorithm used to obtain the model or the terminal value function approximation is stated. The model or the terminal value function of the MPC is either updated by using the current MPC policy to collect samples, i.e., on-policy (ON-P), or another policy, referred to as off-policy (OFF-P) sampling. The terminal value function can also be computed by offline DP or online planning (e.g., RRT) based on a given model. This computation does not utilize sampling strategies, thus is neither related to on-policy or off-policy learning and indicated by N/A.}
	\centering
        \resizebox{\textwidth}{!}{%
	\begin{tabular}{@{}llllllllll@{}}
		\toprule
		\multicolumn{9}{c}{\ac{MPC} as an Actor: MDP Aligned Supervised Learning}\\
		\midrule
		Ref. & Authors & Year& Learning of &\makecell[l]{Architecture/ \\specific FA}& Application & \ac{RWE} & \makecell[l]{\ac{MPC} Formulation\\ (Algorithm, Solver)}& \makecell[l]{Learning/Approx.\\Algorithm} &Sampling\\
		\midrule
	\cite{zhong_value_2013} & \etal{Zhong} & 2013 & \termvalfun{} &\makecell[l]{integrated/\\several FAs} & pendulum and acrobot & no & NMPC (\ilqr{}, custom) & \dpTxt{}&OFF-P \\
        \cite{aswani_provably_2013} & \etal{Aswani} & 2013 & model & parameterized & \makecell[l]{HVAC, quadrotor} & yes & NMPC (SQP, SNOPT~\cite{gill_snopt_2002}) &supervised& ON-P\\
        \cite{nagabandi_neural_2018}          &\etal{Nagabandi} & 2018 & model &integrated/ NN& MuJoCo locomotion & no & random sampling (custom) & DAgger-like & ON-P\\
        \cite{lowrey_plan_2019} &\etal{Lowery}&2019 & \termvalfun{} &integrated/ NN&  \makecell[l]{3D humanoid,\\five-fingered hand} & no & \ac{MPPI} (custom) & supervised & ON-P \\
        \cite{wang_exploring_2019} & \etal{Wang} & 2019 &model&integrated/ NN& MuJoCo locomotion & no & \ac{CEM}~\cite{rubinstein_cross-entropy_2004} (\na)& \ac{IL} & ON-P\\
	\cite{deits_lvis_2019} & \etal{Deits} & 2019 & \termvalfun{}&integrated/ NN&  \makecell[l]{walls pendulum\\ 2D humanoid} & no & \ac{MIQP}-MPC (\gurobi{}~\cite{gurobi_optimization_llc_gurobi_2023}) & supervised& OFF-P\\
        \cite{yang_data_2020} & \etal{Yang} & 2019 & model& integrated/ NN&legged robot~\cite{kenneally_design_2016} & yes & CEM~\cite{rubinstein_cross-entropy_2004} (custom) & supervised& ON-P\\
        \cite{lambert_low-level_2019} & \etal{Lambert} & 2019 & model&integrated/ NN& Crazyflie quadrotor & yes & random sampling (custom)& supervised& OFF-P\\
        \cite{karnchanachari_practical_2020}  &\etal{Karnchanachari} & 2020 &\termvalfun{}&integrated/ NN& \ac{UGV} & yes & NMPC (SQP, \texttt{acado}~\cite{houska_acado_2011}) & \ac{TD}  & ON-P\\
	\cite{beckenbach_q-learning_2020} & \etal{Beckenbach} & 2020 & \makecell[l]{\termvalfun{}, \\stage cost}  &parameterized& chemical reaction&no&\na{}& TD& N/A \\ 
        \cite{hoeller_deep_2020} &\etal{Hoeller} & 2020 &\termvalfun{}&integrated/ NN& ballbot~\cite{hoeller_deep_2020} & yes & NMPC (\ac{SLQ}/\ac{iLQR}, custom)  & \ac{TD} & ON-P \\
        \cite{hatch_value_2021} & \etal{Hatch} & 2021 & \termvalfun{} &\makecell[l]{custom/\\roll-out}&  \makecell[l]{four-wheeled\\skid-steered robot}&no&\ac{MPPI} (custom)& \makecell[l]{DP-related \\(RRT\#)}& N/A   \\
        \cite{morgan_model_2021} & \etal{Morgan} & 2021 & model &integrated/ NN& MuJoCo, robotic hand&yes&\ac{MPPI} (custom) &SAC& ON-P   \\
        \cite{dobriborsci_experimental_2022} & \etal{Dobriborsci} & 2022 & \termvalfun{} &parameterized& mobile robot & yes & NMPC (SLSQP~\cite{noauthor_scipy_nodate}) & \ac{DQN} & OFF-P  \\
	\cite{beckenbach_approximate_2022} & \etal{Beckenbach} & 2022 & \termvalfun{} &parameterized& chemical reaction&no&\na&\ac{ADP} & N/A \\
	\cite{moreno-mora_predictive_2023} & \etal{Moreno-Mora} & 2023 & \termvalfun{} &parameterized& spacecraft & no & NMPC (\texttt{fminunc}~\cite{inc_matlab_2022}) & \ac{VI} & OFF-P \\
        \cite{lin_reinforcement_2024} & \etal{Lin} & 2024 & \termvalfun{} &\makecell[l]{integrated/\\ polynomial}& mobile robot & no & LQR (\na) & \ac{PI} & ON-P\\
        \cite{reiter_ac4mpc_2024} & \etal{Reiter} & 2024 & \termvalfun{} &integrated/ NN& autonomous driving& no & NMPC (SQP, \acados{}~\cite{verschueren_acadosmodular_2022}) & \ac{PPO}, \ac{SAC} & OFF-P\\
        \cite{qu_rl-driven_2024} &\etal{Qu} & 2024 & \termvalfun{}& integrated/ NN& unmanned aerial vehicle & no & \ac{MPPI} (custom)  & \ac{SAC} & OFF-P \\
        \cite{cai_learning-based_2023} & \etal{Cai} & 2023 &
        \makecell[l]{\termvalfun{},\\ stage cost,\\ model,\\ constraints} &parameterized& home energy management & no & NMPC (\na) & \ac{TD}3, AC & ON-P\\
		\bottomrule
	\end{tabular}%
        }
    \label{tab:literature_actor_aligned}
\end{table*}

\subsection{Closed-Loop Learning}
\label{sec:mpc_as_actor_closed_loop_learning}
In the following section, the paradigm of viewing MPC as a parameterized optimization layer as part of an actor policy is explained in more detail.
In the closed-loop learning paradigm, the MPC model is not required to minimize a prediction loss~$\mathcal{L}^\mathrm{id}$ or to provide a terminal value function that approximates the optimal value function as in the \MDPAligned{} setting.
Instead, the model, costs, or constraints may be changed to solely improve the closed-loop performance of the MPC policy~$\mu^\mathrm{MPC}$ according to Definition~\ref{eq:general_problem}.
Fig.~\ref{fig:sketches_in_policy} shows a sketch of the \closedLoopOpt{} learning paradigm.

An advantage of this paradigm is the ability to achieve optimal closed-loop performance in the real environment without the necessity of computationally demanding aligned formulation of Sect.~\ref{sec:mpc_as_actor_aligned}. However, by modifying the various parts of the MPC purely to increase the closed-loop performance, explainability, safety, generalization, and adaptability get lost, particularly if the model or constraints are allowed to change.

Practically, this paradigm differs depending on whether the controls provided by the MPC or the parameterization of the MPC obtained through a learnable NN are assumed as RL actions. This two practical implementations are explained in the following.\\

\subsubsection{Differentiable MPC}
MPC can be used as an optimization layer within the RL policy to provide a good initial performance by leveraging knowledge about the task~\cite{gros_economic_2022}. Ideally, the MPC performance would only be slightly suboptimal and provide safety guarantees for a known model. The hope would be that by only a few RL iterations, the parameters can be modified towards nearly optimal closed-loop performance, particularly related to stochasticity. 

However, the MPC optimization problem as part of the learned actor creates potential challenges for both the forward path, where the policy evaluation includes solving the optimization problem, and the training of parameters, which, in some algorithms, requires appropriately passing gradients through an optimization algorithm.
Solving optimization problems is numerically challenging. Thus, the related iterative algorithms may slow down learning when evaluating a policy. Moreover, obtaining a global optimizer outside the class of convex optimization problems is intractable in general. Although local optima are often sufficient but require warm-starting, the initial states used to warm-start an optimization solver introduce additional states, which are required to be considered in the learning algorithm.

When using the MPC in an RL actor, gradients need to be computed through the optimization solver as part of the backpropagation. These gradients can be computed using the implicit function theorem, cf.~\appref{appendix:differentiating} and related literature~\cite{amos_differentiable_2018,zanon_safe_2021,romero_actor-critic_2024} for further details. However, existing optimization software is required to support such features. Related features were recently developed in various tools, see Sect.~\ref{sec:software}. 

When differentiating through the optimization layer, the integrated, hierarchical, parallel, and parameterized architectures according to Fig.~\ref{fig:structure2} may be used, but to the best of the authors' knowledge, only the hierarchical, e.g.~\cite{zanon_reinforcement_2020,romero_actor-critic_2024}, and the parameterized architectures, e.g.,~\cite{amos_differentiable_2018,gros_data-driven_2020,moradimaryamnegari_model_2022}, were proposed so far.

The following algorithm based on the ideas of~\cite{gros_data-driven_2020} provides a basic version of closed-loop optimal learning, where the parameters are updated by differentiating the \ac{MPC}.
The temporal difference update~\eqref{eq:q_learning} is adapted to the parameterized \ac{MPC} setting to
\begin{algbox}{Q-learning with MPC Q-function~$(Q_\theta^\mathrm{MPC} \approx Q^\star)$}
\begin{align*}
    & \theta \leftarrow \theta + \alpha \delta \nabla_\theta Q_\theta^\mathrm{MPC} (S, A), \\
    &\delta \coloneqq l (S, A) + \gamma V_\theta^\mathrm{MPC} (S^+) - Q_\theta^\mathrm{MPC} (S, A), \\
    &\textrm{with } (S, A) \sim \mathcal{D}^{\pi^\mathrm{MPC}_\theta},\; S^+ \sim P(\cdot \mid S, A).
\end{align*}
\end{algbox}
The state-action distribution~$\mathcal{D}$ is generated by controlling the system via the stochastic parameterized \ac{MPC} exploration policy~$\pi_\theta^\mathrm{MPC}$.
Note that the update scheme only considers a single sample, which is often reasonable in the case that the $Q^\mathrm{MPC}_\theta$ is just a parameterized \ac{MPC} scheme as described in Fig.~\ref{fig:structure2}, without a \ac{NN}.

Another example of closed-loop optimal learning based on a \ac{DDPG} actor-critic formulation similar to \eqref{eq:ddpg} was introduced in \cite{gros_data-driven_2020} and is given by
\begin{algbox}{DDPG with MPC actor $(\mu_\theta^\mathrm{MPC} \approx \pi^\star)$}
\begin{align*}
    & w \overset{Q}{\leftarrow} w + \frac{\alpha_w}{B} \sum_{i=1}^B \delta_i \nabla_w Q_w (S_i, A_i), \\
    & \theta \overset{\mu}{\leftarrow} \theta + \frac{\alpha_\theta}{B} \sum_{i=1}^B \nabla_\theta \mu_\theta^\mathrm{MPC}(S_i) \left. \nabla_a Q_w(S_i, a) \right|_{a = \mu_\theta^\mathrm{MPC}(S_i)}, \\
    &\delta_i \coloneqq l (S_i, A_i) + \gamma Q_{\bar{w}} (S^+_i, A^+_i) - Q_w (S_i, A_i), \\
    &\textrm{with } (S_i, A_i, S^+)\sim \mathcal{D}^\mathrm{buffer},\; A^+_i =  \mu_\theta^\mathrm{MPC}(S_i). 
\end{align*}
\end{algbox}
The parameters are updated for the deterministic policy~$\mu_\theta^\mathrm{MPC}$ and a \ac{NN} critic~$Q_w$ with parameters~$w$.
The stochastic policy~$\pi_\theta^\mathrm{MPC}$ is used for exploration in order to fill the replay buffer~$\mathcal{D}^\mathrm{buffer}$.\\

\subsubsection{MPC as part of the environment}
    Passing gradients through the MPC in the actor in the closed-loop optimal learning paradigm can also be omitted by, conceptually, considering the MPC as part of the environment. 
    In such a setting, a parameterized MPC that controls an environment can be seen as an augmented new environment whose actions are the MPC parameters. Accordingly, the RL critic evaluates the values related to these parameters instead of the actions provided by the MPC, which are, in fact, hidden from the RL framework.
    Considering the parameterization of the MPC as the environment input, instead of the actions obtained from the MPC, avoids computing sensitivities but comes at the cost of potentially vastly increasing the dimensions of the action space and obtaining gradient information through sampling. 
    Such a setting usually involves the hierarchical~\cite{brito_where_2021,reiter_hierarchical_2023} or the parameterized architecture~\cite{zarrouki_adaptive_2024} of Fig.~\ref{fig:structure2}.

Whether MPC is considered as part of the environment in literature, involving its differentiation, is indicated by the differentiating MPC column in Tab.~\ref{tab:literature_actor_layer}. \\

\subsubsection{Exploration}
Tab.~\ref{tab:literature_actor_layer} lists closed-loop optimal algorithms that differ in the architecture and whether they include differentiation through the optimization layer. Moreover, Tab.~\ref{tab:literature_actor_layer} considers how exploration was performed in related literature, cf. Sect.~\ref{sec:algo_rl}. 
Exploration is often required to improve the currently learned policy in \ac{RL}.
Four choices of exploration are used, particularly in the case of hierarchical and parameterized architectures: 
\begin{enumerate}
    \item Adding noise to the action proposed by the optimization layer. The downside of adding noise posterior to the optimizer is the potentially unsafe exploration, depending on the chosen exploration noise.
    \item Modifying the cost in the optimization layer \eqref{eq:mpc_main_cost} with an additive term $d^\top u_0$ where $d$ is a possibly randomly selected vector. Given that only the cost is modified, the actions remain feasible with respect to the model while the additive term introduces a gradient over the initial control input~\cite{kordabad_bias_2023} for exploration.
    \item Add a perturbation to the parameter~$\phi$ for the hierarchical architecture and~$\theta$ for the parameterized architecture as shown in~Fig.~\ref{fig:structure2}, guaranteeing safe actions~\cite{gros_towards_2019}.
    \item Instead of noise on the parameters or controls, optimistic initialization is an \ac{RL} technique where Q-values (or estimates) are initialized with higher-than-expected values, encouraging the agent to explore and gradually reduce these optimistic estimates to reflect the true values. For a discussion in the context of deep \ac{RL} see \cite{rashid_optimistic_2020}.
\end{enumerate}


\begin{table*}
	\caption{Literature that uses \ac{MPC} as part of the actor in closed-loop learning. The table lists the applications, whether \acf{RWE} were performed, and the particular \ac{RL} and \ac{MPC} algorithms. 
    The computation of gradients during the back-propagation pass of the RL algorithm may require the computation of the sensitivities of the MPC solution, indicated in the column ``differentiating MPC''. Alternatively, the sensitivities could be sampled, by, e.g., considering the MPC as part of the environment.
    Similarly, the table compares where the exploration noise of the RL algorithm is added. The noise could be added to the parameters of the MPC optimization problem, or posterior, to the optimal controls provided by the MPC or by using an optimistic initialization of the value functions.}
	\centering
        \resizebox{\textwidth}{!}{%
		\begin{tabular}{@{}llllllllll@{}}
            \toprule
            \multicolumn{10}{c}{\ac{MPC} as an actor: closed-loop optimal}\\
\midrule
Ref.& Authors & Year & Exploration & \makecell[l]{Diff. \\ \ac{MPC}} & Architecture & Application & \ac{RWE} & \makecell[l]{\ac{MPC} Formulation\\ (Algorithm, Solver)} & \ac{RL} Algorithm \\
\midrule
\cite{amos_differentiable_2018}          &\etal{Amos} & 2018 & actions & yes &\paramterized{}&pendulum, cartpole&no&LMPC (\na)&\ac{IL}\\
\cite{greatwood_reinforcement_2019}      &\etal{Greatwood} & 2019 & optimistic & no &\hierarchical{}& drone navigation & yes &  LMPC (custom)&\ac{TD} \\
\cite{tram_learning_2019}                &\etal{Tram} & 2019 & N/A& no &\hierarchical{}& vehicle intersection & no & LMPC (\na)& Q-learning  \\
\cite{gros_data-driven_2020}             &\etal{Gros} & 2020 & actions & yes &\paramterized{}& evaporation process & no & NMPC (SQP, \acados{}~\cite{verschueren_acadosmodular_2022}) & Q-learning  \\
\cite{brito_where_2021}                  &\etal{Brito} & 2021 & parameters & no &\hierarchical{}& multi-agent unicycle & no & NMPC (\na{}, \texttt{ForcesPro}\cite{domahidi_forces_2014}) & \ac{PPO}  \\
\cite{zanon_safe_2021}                   &\etal{Zanon} & 2021 & actions & yes &\hierarchical{}& evaporation process & yes & LMPC (\na) & Q-learning  \\
\cite{moradimaryamnegari_model_2022}     &\etal{Moradimaryamnegari} & 2022 & N/A& yes &\paramterized{}& water tank & yes & NMPC (\ac{IP}, \ipopt{}~\cite{wachter_implementation_2006}) & SARSA \\
\cite{brito_learning_2022}               &\etal{Brito}& 2022 & parameters & no &\hierarchical{}& highway traffic & no &  NMPC (\na{}, \texttt{ForcesPro}\cite{domahidi_forces_2014})& \ac{SAC} \\
\cite{zhang_learning-based_2022}         &\etal{Zhang} & 2022 & parameters & no &\hierarchical{}& quadruped robot & no &   LMPC (\na) \\
\cite{pfrommer_safe_2022}                &\etal{Pfrommer} & 2022 & actions & no &\hierarchical{}& 2D linearized quadrotor & no &   LMPC (\texttt{MOSEK}~\cite{aps_mosek_2024}) & \ac{PG} \\
\cite{reiter_hierarchical_2023}          &\etal{Reiter} & 2023 & parameters & no &\hierarchical{}& autonomous racing & no & NMPC (SQP, \acados{}~\cite{verschueren_acadosmodular_2022}) & \ac{SAC} \\
\cite{liu_learning_2023}                 &\etal{Liu} & 2023 & N/A & yes &\paramterized{}& multi-agent games  & yes &  \makecell[l]{generalized Nash \\equilibrium solver (custom)} & custom\\
\cite{romero_actor-critic_2024}          &\etal{Romero} & 2024 & actions& yes &\hierarchical{}& drone racing  & yes & NMPC (\ac{iLQR}, custom)  & \ac{PPO}  \\
\cite{tao_difftune-mpc_2024}             &\etal{Tao} & 2024 & N/A & yes &\paramterized{}& quadrotor & no & NMPC (SQP, \acados{}~\cite{verschueren_acadosmodular_2022}) & N/A  \\
\cite{zarrouki_safe_2024}                 &\etal{Zarrouki} & 2024 & parameters & no &\hierarchical{}& race car & yes &  NMPC, (SQP, \acados{}~\cite{verschueren_acadosmodular_2022}) &  \ac{PPO} \\  
\cite{zarrouki_adaptive_2024}            &\etal{Zarrouki} & 2024 & parameters & no &\paramterized{}& race car & yes &  NMPC (SQP, \acados{}~\cite{verschueren_acadosmodular_2022}) &  \ac{PPO} \\  
\cite{wen_collision-free_2024}                 &\etal{Wen} & 2024 & parameters & no &\hierarchical{}& mobile robot & no &  LMPC (\na)&  \ac{PPO} \\  
\bottomrule
		\end{tabular}
}
 	\label{tab:literature_actor_layer}
\end{table*}

Closed-loop optimal learning algorithms require the MPC to be part of the RL algorithm. Yet, the following two variants in Sect.~\ref{sec:prepostprocessing} use a hierarchical MPC and RL setting, where MPC is added posterior to the training in the deployed controller.

\subsection{MPC for Pre- and Postprocessing}
\label{sec:prepostprocessing}
In the following Section, two concepts are highlighted that use MPC within the deployed policy, yet, not during the training procedure. 
Since the trained policy is not aware of the filter, the expected performance can not be expected to be closed-loop optimal, as in the previous section.
Depending on the purpose of the MPC, a distinction is made between an MPC used as a reference generator for preprocessing, cf. Fig.~\ref{fig:structure3} and Sect.~\ref{sec:mpc_actor_ref}, and MPC used for postprocessing, cf. Fig.~\ref{fig:structure4} and Sect.~\ref{sec:mpc_actor_filt}. Relevant literature is compared in Tab.~\ref{tab:literature_actor_filter_refgen}. \\

\subsubsection{Preprocessing: MPC as a Reference Generator}
\label{sec:mpc_actor_ref}
MPC may be used as a reference generator for an RL policy such as in~\cite{jenelten_dtc_2024}. 
This combination approach of MPC and RL is particularly useful if the output of MPC is a planned trajectory rather than a single action. Since solving the \ac{MPC} may be slow, the RL policy may be trained with a computationally cheaper reference generator, as proposed in~\cite{jenelten_dtc_2024}. 
Notably, this setting may often be used in publications without an explicit statement that describes the reference provider.\\
\begin{figure}
	\includegraphics[width=.49\textwidth]{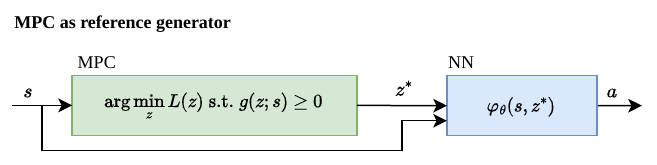}
	\caption{MPC can be used as a reference generator for an RL policy. The distribution of the input of the policy may depend on the distribution of optimizers obtained from MPC while interacting with the environment.}
	\label{fig:structure3}
\end{figure}

\subsubsection{MPC for Postprocessing}
\label{sec:mpc_actor_filt}
Adding MPC to a trained policy may fulfill one of the two purposes: (i) providing safety w.r.t. an assumed model and constraints, known as safety filter~\cite{tearle_predictive_2021} or (ii) providing a coarse solution by a warm-start of an optimization solver.

RL policies may struggle with guaranteeing safety~\cite{brunke_safe_2022} and plan for smooth trajectories, particularly in a high-dimensional state space~\cite{shen_reinforcement_2023}. 

As opposed to the previous Sect.~\ref{sec:mpc_as_actor_closed_loop_learning} of closed-loop optimal learning, the approaches that use MPC as a filter do not consider MPC during learning. Thus, the expected behavior may not be closed-loop optimal since the policy is unaware of MPC during training.

Tab.~\ref{tab:literature_actor_filter_refgen} shows relevant work that uses the MPC for postprocessing and distinguishes the filtering of a single proposed action, e.g., used in the safety-filter framework~\cite{tearle_predictive_2021}, or a whole trajectory of actions, e.g.,~\cite{ceder_birds-eye-view_2024} or filtering of a state trajectory as in~\cite{grandesso_cacto_2023}. The filtering of action or state trajectories can be interpreted as providing an initial guess of decision variables of a nonconvex optimization problem to a solver, which then aims to find a good local minimum. In fact, the MPC needs to be approximately aligned with the MDP, and the role of the policy rather assists the MPC optimization solver by finding global/low-cost optimizers.

One ought to observe, though, that if an MPC formulation is available to ensure the safety of the action taken in the real environment, typically in the form of a robust MPC scheme, then it is debatable whether training a policy (typically based on a NN) to be filtered by the robust MPC scheme or training the robust MPC scheme directly~\cite{zanon_safe_2021} is more effective.

\begin{figure}
	\includegraphics[width=.49\textwidth]{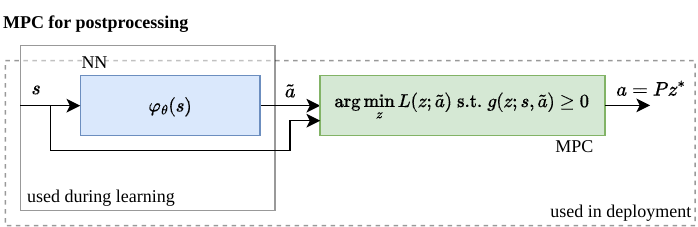}
	\caption{MPC can be used to smooth or filter reference trajectories. For example, MPC may be used to provide safety guarantees that are hard to achieve by RL policies.}
	\label{fig:structure4}
\end{figure}


\begin{table*}
	\caption{Literature that uses MPC during the deployed algorithm together with the RL policy but not within the learning phase. Algorithms either use MPC as a reference provider for an RL policy or for postprocessing. Postprocessing is used to provide safety w.r.t. a known model and constraints or to take an action or state trajectory as an initial guess for the optimization solver.}
	\centering
	\resizebox{\textwidth}{!}{%
	\begin{tabular}{@{}llllllll@{}}
		\toprule
		\multicolumn{8}{c}{\ac{MPC} for Postrocessing}\\
		\midrule
		Ref. & Authors & Year& Filtering of &Application & \ac{RWE} & \makecell[l]{\ac{MPC} Formulation\\ (Algorithm, Solver)}& \ac{RL} Algorithm \\
		\midrule
        \cite{li_robust_2020} & \etal{Li} &2020 &action& point-mass, kin. vehicle & no & NMPC, robust (\na) & \ac{DDPG} \\
		\cite{tearle_predictive_2021} &\etal{Tearle}& 2021 & action & miniature race cars & yes & NMPC (SQP, \acados{}~\cite{verschueren_acadosmodular_2022}) & \ac{IL}, DAgger~\cite{ross_reduction_2011} \\
		\cite{shen_reinforcement_2023} &\etal{Shen}& 2023 & traj. roll-out & multi-vehicle motion planning & no & NMPC (\ac{IP}, \ipopt{}~\cite{wachter_implementation_2006}) & Q-learning \\
        \cite{didier_approximate_2023} &\etal{Didier}& 2023 & action & autonomous driving & no & NMPC (\ac{IP}, \ipopt{}~\cite{wachter_implementation_2006}) & \na \\
		\cite{grandesso_cacto_2023} & \etal{Grandesso} & 2023 &traj. roll-out& 3-DoF planar manipulator & no & NMPC (IP, \ipopt{}~\cite{wachter_implementation_2006}) &   custom AC \\
        \cite{mamedov_safe_2024} & \etal{Mamedov} & 2024 &action& flexible robot arm & no & NMPC (SQP, \acados{}~\cite{verschueren_acadosmodular_2022}) & \makecell[l]{BC, DAgger, AIRL\\ GAIL, PPO, SAC}\\
        \cite{reiter_ac4mpc_2024} & \etal{Reiter} & 2024 &traj. roll-out& vehicle motion planning & no & NMPC (SQP, \acados{}~\cite{verschueren_acadosmodular_2022})& PPO, SAC\\
        \cite{alboni_cacto-sl_2024} & \etal{Alboni} & 2024 & traj. roll-out& point-mass with obstacles& no & NMPC (IP, \ipopt{}~\cite{wachter_implementation_2006}) &  custom AC \\
        \cite{ceder_birds-eye-view_2024} & \etal{Ceder} & 2024 & action roll-out & mobile robot& no & \makecell[l]{NMPC (proximal gradient,\\ \texttt{PANOC}~\cite{stella_simple_2017}) }&  DDPG\\
        \cite{qu_rl-driven_2024} &\etal{Qu} & 2024 & traj. roll-out & unmanned aerial vehicle & no & \ac{MPPI} (custom)  & \ac{SAC} \\
		\midrule
		\midrule
		\multicolumn{8}{c}{\ac{MPC} as Reference Generator}\\
		\midrule
	\cite{jenelten_dtc_2024} &\etal{Jenelten} & 2024 && quadruped robot & yes & NMPC (TAMOLS~\cite{jenelten_tamols_2022})& \ac{PPO}  \\
        \cite{bang_rl-augmented_2024} &\etal{Bang} & 2024 && humanoid robot & no & LMPC (\na) & \ac{PPO}  \\
		\bottomrule
	\end{tabular}}
 	\label{tab:literature_actor_filter_refgen}
\end{table*}



So far, MPC has been considered as part of the actor or as an expert actor. In the following, we also highlight works that consider MPC as part of the critic, i.e., the MPC is used solely during training to provide a value function estimate.

\section{MPC as a Critic}
\label{sec:comb_critic}
As outlined in~\secref{sec:mpc}, parameterized variants of the \acp{OCP}~\eqref{eq:mpc-problem} and~\eqref{eq:mpc-q-problem} allow for a structured function approximation of value function, action-value function, and the policy.
\begin{figure}
	\includegraphics[width=\linewidth]{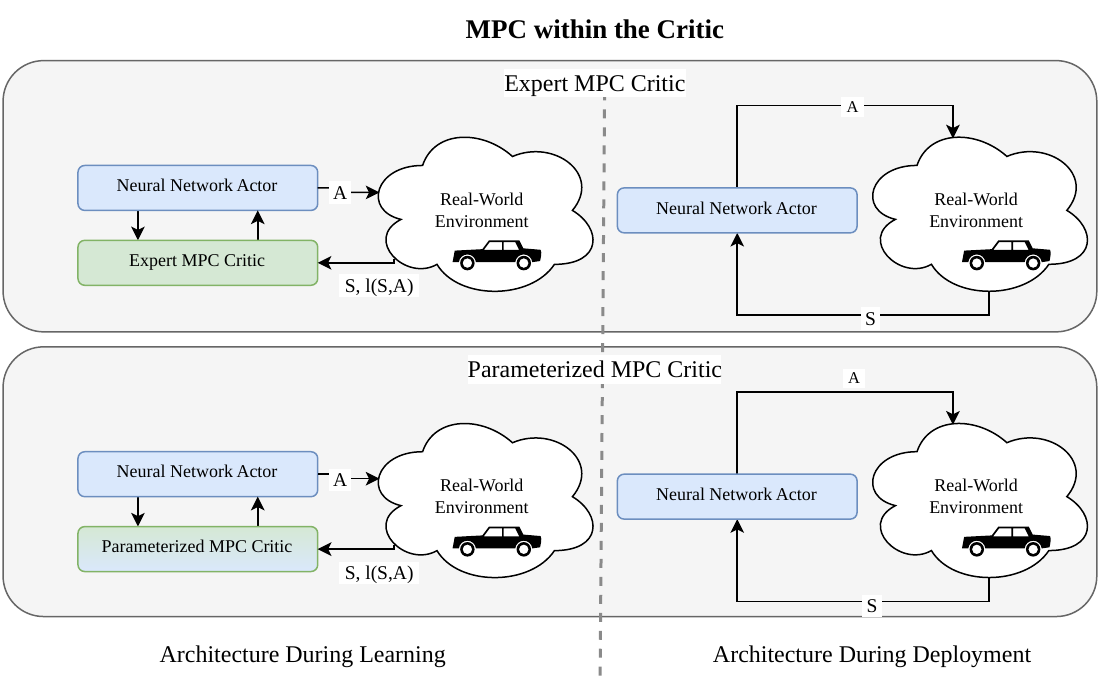}
	\caption{Combinations: MPC as a critic. The plot is split into the learning and deployment phases. Blue boxes indicate Neural Networks (NNs), green boxes are used for MPCs, and blue/green boxes refer to parameterized MPCs that involve parameters/NNs that are learned during the learning phase, see Sect.~\ref{sec:architectures}.}
	\label{fig:sketches_critic}
\end{figure}
While earlier discussions focused on using \ac{MPC} as an actor, we next discuss the role of MPC used as a critic, cf., Fig.~\ref{fig:structure1}. We distinguish between three variants of \ac{MPC} as a critic.
\begin{enumerate}[label=\Alph*.]
    \item \textbf{\ac{MPC} as an expert critic}, \secref{Sec:MPCExpertCritic}. The MPC is parameterized based on expert knowledge of the problem at hand, entering through the cost, model, or constraints, which we refer to as expert critic.
    \item \textbf{\ac{MPC} as a learnable critic}, \secref{sec:mpc_learnable_critic}. The MPC involves learnable parameters to improve the accuracy of the critic.
    \item \textbf{MPC as a learnable actor-critic}, \secref{sec:mpc_learnable_actor_critic}. The critic parameterization (possibly) differs from the one of the actors.
\end{enumerate}

In Tab.~\ref{tab:literature_critic}, different publications are presented, depending on whether \ac{MPC} parameters are tuned and which critic type described in the following three sections was used.

\subsection{MPC as an Expert Critic} \label{Sec:MPCExpertCritic}
An \Ac{MPC} scheme, as previously discussed, can readily deliver an action-value function~$Q^\mathrm{MPC}$, which can be used to derive policy gradient equations to train an actor.
An important assumption here is that the value function~$Q^\mathrm{MPC}$ approximates $Q^\star$ as closely as possible. In fact, the \ac{MPC} may usually just provide a value function of a desirable good suboptimal policy.
Approximating $Q^\star$ by the fixed MPC Q-function was recently proposed in~\cite{ghezzi_imitation_2023}.
Very related is the line of work in \cite{carius_mpc-net_2020, ghezzi_imitation_2023} that, inspired by the Hamiltonian-Jacobi-Bellman equations, uses a first-order approximation of a Q-value function to criticize the actions generated by a learned policy.


Consider the evaluation of the critic, involving solving the MPC optimization problem to obtain the corresponding action-value function~$Q^\mathrm{MPC}$.
Then, for a deterministic, parameterized policy~$\mu_\theta$, the learning objective related to \eqref{eq:policy_obj} is 
\begin{equation}\label{eq:expertcritic}
    J^\mu_\mathrm{MPC}(\theta) = \frac{1}{1 - \gamma} \mathbb{E}_{S\sim \rho^{\mu_\theta}} \biggl[ Q^\mathrm{MPC}(S, \mu_\theta(S)) \biggr]\,,
\end{equation}
with the discounted visitation frequency~$\rho^{\mu_\theta}$  as defined in Sect.~\ref{sec:policy_gradient_methods}.
Note that the policy parameter~$\theta$ is updated to minimize the objective~\eqref{eq:expertcritic} instead of the behavior cloning objective in~\secref{sec:comb_expert} or the MDP objective in \eqref{eq:policy_obj}.

An update scheme using the Q-function from \ac{MPC} to criticize the learned policy~$\mu_\theta$ as in \cite{ghezzi_imitation_2023} is defined by
\begin{algbox}{DDPG with an MPC expert critic~$(\mu_\theta \approx \mu^\mathrm{MPC})$}
\begin{align*}
    & \theta \leftarrow \theta + \frac{\alpha_\theta}{B} \sum_{i=1}^B \nabla_\theta \mu_\theta(S_i) \left. \nabla_a Q^\mathrm{MPC}(S_i, a) \right|_{a = \mu_\theta(S_i)}, \\
    &\textrm{with } S_i \sim \mathcal{D},
\end{align*}
\end{algbox}
which is related to the \ac{DPG} in \eqref{eq:deterministic_policy_gradient} but with a fixed expert critic.
Here, $\mathcal{D}$ could be either a fixed dataset of states or generated iteratively by a mixture of $\mu_\theta$ and $\mu^\mathrm{MPC}$ like in \cite{ross_reduction_2011}.

Similar to \secref{sec:comb_expert}, using an \ac{MPC} as an expert critic leads to imitating the MPC expert policy.
The primary advantage of using the \ac{MPC} as an expert critic when compared to standard \ac{IL} that relies on the mismatch between the expert and learned controls is its ability to guide a function approximator with limited expressive power by emphasizing which actions are more important to fit by using the Q-function.
It further can introduce the constraint satisfactions to the objective using slacked constraints~\cite{ghezzi_imitation_2023}.

It is important to emphasize that the \ac{MPC} expert critic remains fixed and, therefore, does not approximate the action-value function~$Q^{\mu_\theta}$ of the learned policy during training.
Indeed, policy gradient methods formally require the critic to be of the policy, i.e., to deliver the policy action-value function~$Q^{\mu_\theta}$. In contrast, an MPC as an expert critic instead delivers the action-value function~$Q^\mathrm{MPC}$ based on expert knowledge of the environment, with the aim of approximating the optimal action-value function~$Q^\star$ as accurately as possible. 
If the actor has enough flexibility to clone the MPC policy perfectly, then using MPC as an expert critic would ultimately result in the actor matching the MPC performance but not improving it further. 

Consequently, as discussed in the following sections, further improvements can be achieved by allowing the \ac{MPC} critic to adapt to the currently learned policy $\mu_\theta$.

\subsection{MPC as a Learnable Critic}
\label{sec:mpc_learnable_critic}
In the context of policy gradient methods, a parameterized MPC of any architecture shown in Fig.~\ref{fig:structure2} can be used to approximate the policy action-value function~$Q^\pi$, and updated using data to 
capture it as correctly as possible. 
For instance, the parallel architecture of  Fig.~\ref{fig:structure2} is used in~\cite{bhardwaj_blending_2021}.
Using the MPC sensitivities, the MPC parameters can be learned using value-based methods from~\eqref{eq:q_learning}.
As discussed before, a parametric MPC can fully capture~$Q^\pi$ given that it has a rich parameterization \cite{cai_learning-based_2023, seel_convex_2022, esfahani_learning-based_2023, adhau_fast_2023}. 
This approach can be thought of as introducing expert knowledge into the policy gradient pipeline, using a critic combining harmoniously model-based knowledge and learning to capture~$Q^\pi$ as accurately and quickly as possible. The main reasoning behind this combination is that MPC is typically capable of providing a correct structure for the action-value function~$Q^\pi$ prior to any training, giving a good starting point for the policy-gradient method.
Then, the classic learning of $Q^\pi$ allows the MPC to become a better critic of the policy and to remain a good critic as the policy is modified by the policy-gradient method.
It is worth mentioning here that MPC as a learnable critic can be readily combined with a more classic \ac{NN}, e.g., as a summation of their contribution to the action-value approximation \cite{bharadhwaj_model-predictive_2020}.
In that context, the MPC scheme can deliver a broadly correct approximation, while the \ac{NN} can provide fine corrections.
Unlike the approach of Sec. \ref{Sec:MPCExpertCritic}, MPC as a learnable critic aligns well with the actor-critic framework.

\subsection{MPC as a Learnable Actor-Critic}
\label{sec:mpc_learnable_actor_critic}
Using MPC as an actor and as a learnable critic naturally offers the opportunity
to combine them into an actor-critic setup using MPC for both. In this setting,
MPCs can be used both as an actor, delivering a parameterized
policy~$\pi_\theta^\mpcShrt$ to be trained, and as a critic, delivering value
functions that are by construction close to the value function of the \ac{MPC}
policy, i.e., $V_\theta^\mpcShrt\approx V^{\pi^\mpcShrt_\theta},\,
Q_\theta^\mpcShrt\approx Q^{\pi^\mpcShrt_\theta}$, as opposed to random
intializations of \acp{NN}.

During the learning, the MPC scheme operating as an actor typically needs to
differ from the MPC scheme operating as a critic because $\pi$ is
not a minimizer of $Q^{\pi}$ when the policy is not optimal. In practice, two
different MPC schemes can be used and trained in parallel: one as a critic,
maintaining a good approximation of $Q^{\pi^\mpcShrt_\theta}$ associated with
$\pi^\mpcShrt_\theta$ given by \eqref{eq:mpc-q-problem}, and one to support $\pi^\mpcShrt_\theta$ itself, given by \eqref{eq:mpc-problem}. It can be,
however, computationally expensive to use two MPC schemes in parallel, as the
optimal solution of both must be produced at every training step of the policy,
rather than the solution of only the MPC supporting the policy
$\pi^\mpcShrt_\theta$. 

Along that line, \cite{anand_painless_2023} proposed a critic formulation based
on the value function obtained from an MPC-based actor. Their formulation
employed~$V^\mpcShrt_\theta$ along with the compatible function approximation
suggested in \cite{silver_deterministic_2014} to build a local approximation of
the critic by using a first-order Taylor expansion, simplifying the
actor-critic formulation with a single MPC scheme. This approach is effective if the MPC policy is sufficiently close to the optimal policy.

\begin{table*}
	\caption{Literature that uses \ac{MPC} as a critic for \ac{RL}. The \ac{MPC} parameters are either learned within the \ac{RL} algorithm or fixed. Further, a distinction is made, whether the critic is a fixed expert, see Sect.~\ref{Sec:MPCExpertCritic}, a more flexible learned critic as described in Sect.~\ref{sec:mpc_learnable_critic}, or is part of both the actor and the critic, Sect.~\ref{sec:mpc_learnable_actor_critic}.}
	\centering
        \resizebox{\textwidth}{!}{%
	\begin{tabular}{@{}lllllllll@{}}
		\toprule
		\multicolumn{8}{c}{\ac{MPC} as a Critic}\\
		\midrule
		Ref. &Authors& Year & \makecell[l]{Learned \ac{MPC}\\Parameters} & \makecell[l]{Critic Type/\\Par. MPC Arch.} & Application & \ac{RWE} & \makecell[l]{\ac{MPC} Formulation\\ (Algorithm, Solver)} & \ac{RL} Algorithm \\
		\midrule
            \cite{reske_imitation_2021} & \etal{Reske} & 2021 & no & expert / \na{} & quadruped robot & yes & \ac{NMPC} (\ac{DDP}, \na{}) & \ac{IL} \\
            \cite{bhardwaj_blending_2021} &\etal{Bhardwaj} &2021 & no & \makecell[l]{learned /\\parallel} & \makecell[l]{several robotic\\ in-hand manipulation}  &no& MPPI & \ac{TD} \\
		\cite{ghezzi_imitation_2023}  &\etal{Ghezzi}&2023 & no & expert / \na{} & cart pole & no & NMPC (SQP, \acados{}~\cite{verschueren_acadosmodular_2022}) & 	\ac{IL}	 \\
            \cite{anand_painless_2023} & \etal{Anand}&2023 & yes & \makecell[l]{Actor-Critic /\\parameterized} & cart pole & no & NMPC (\ac{IP}, \ipopt{}~\cite{wachter_implementation_2006}) & \ac{DPG}\\
            \cite{grandesso_cacto_2023} & \etal{Grandesso} & 2023 & no & expert / \na{} & 3-DoF planar manipulator & no & NMPC (\ac{IP}, \ipopt{}~\cite{wachter_implementation_2006})&  custom AC \\
            \cite{carius_mpc-net_2020} & \etal{Carius} & 2023 & no & expert / \na{} & quadruped robot & yes & \ac{NMPC} (\ac{DDP}, \na{}) & \ac{IL} \\
            \cite{alboni_cacto-sl_2024} & \etal{Alboni} & 2024 & no & expert / \na{} & Dubin's car parking& no & NMPC (DDP, \ac{IP}, \ipopt{}~\cite{wachter_implementation_2006})&  custom AC \\
		\bottomrule
	\end{tabular}
        }
	\label{tab:literature_critic}
\end{table*}


\section{Theoretical Considerations for Combining MPC and RL}
\label{sec:theory}

\acp{MDP} establish a fundamental bridge between \ac{RL} and \ac{MPC}, as \ac{RL} provides tools to solve \acp{MDP} while \ac{MPC} formulations can provide approximations, as discussed in the context of aligned learning in \secref{sec:mpc_as_actor_aligned} and closed-loop optimal learning in \secref{sec:mpc_as_actor_closed_loop_learning}. This section explores the theoretical foundations of combining \ac{MPC} and \ac{RL} within the framework of \ac{MPC}-based \ac{MDP} approximations. A central theoretical contribution by \cite{gros_data-driven_2020} provides a theoretical link between \ac{ENMPC} and \ac{RL} by connecting the \ac{MDP} value functions and policy discussed in \secref{sec:problem_setting} with those generated by derivative-based \ac{MPC} formulations detailed in \secref{sec:mpc_algo_derivative}. This connection provides theoretical justification for approaches that incorporate \ac{MPC} as either an actor or critic, as described in \secref{sec:mpc_as_actor} and \secref{sec:comb_critic}, respectively. We begin by presenting this key theoretical result before outlining its subsequent developments.

Following the definition in~\secref{sec:policy_gradient_methods}, let $p_{f^\mpcShrt}(s_0 \rightarrow s^+, k, \pi_\theta)$ denote the probability of reaching state~$s^+$ at step $k$ when starting from the initial state $s_0$ and following the policy~$\pi$, under model dynamics $f^\mpcShrt$. A key assumption in \cite{gros_data-driven_2020, kordabad_equivalence_2024} is that the optimal value function remains bounded under the optimal policy and model of the dynamics:

\begin{assumption}[{\cite[Assumption 1]{kordabad_equivalence_2024}}]\label{assume:S} The following set is non-empty for a given $\bar N\in\mathbb{N}$.
    \begin{equation}\label{eq:assum1}
        \resizebox{\columnwidth}{!}{
        $\mathcal{S}=:\left\{s\in\mathcal{X}\,\,\Big|\,\,\left|\mathbb{E}_{S\sim p_{f^\mpcShrt}(s_0 \rightarrow s^+, k, \pi^\star)} \left[V^\star(S)\right]\right|<\infty, \ \forall\, k \leq \bar N\right\}$}
    \end{equation}
\end{assumption}

The authors of~\cite{gros_data-driven_2020} demonstrate that discounted \ac{MPC} can capture the optimal value functions and policy of a discounted \ac{MDP} through modifications of the stage and terminal costs, even when the model $f^\mpcShrt$ differs from the true state-transition function. This result was later extended in \cite{kordabad_equivalence_2024} to undiscounted \ac{MPC} by the following Theorem, establishing a central link to classical \ac{MPC} stability theory:
\begin{theorem}[{\cite[Theorem 1]{kordabad_equivalence_2024}}]\label{theorem1} 
    Suppose that Assumption~\ref{assume:S} holds for $\bar N \geq N$. Then, there exists a terminal cost~$\bar{V}_\theta^\mpcShrt$ and a stage cost~$l_\theta^\mpcShrt$ such that the following identities hold, for all $\gamma$, $N\in\mathbb{N}$ and $s\in\mathcal{S}$:
    \begin{enumerate}
          \item $\pi_\theta^\mpcShrt(s)=\pi^\star(s), $ \label{eq:pipi}
          \item $V_\theta^\mpcShrt(s)=V^\star(s),\,\, $\label{eq:VV}
          \item $Q_\theta^\mpcShrt(s,\,a)=Q^\star(s,\,a),\,\,$ for the inputs $a\in\mathcal{U}$ such that $\lvert\mathbb{E}_{S^{+} = f^\mpcShrt(s,\,A),\,A\sim\pi_\theta^\mpcShrt(\cdot \mid s)}\left[V^\star({S}^{+})| s, A\right]\rvert<\infty$\label{eq:QQ}.
    \end{enumerate}
\end{theorem}
\begin{proof}
     See \cite{kordabad_equivalence_2024}.
\end{proof}

Given the conditions in~\cref{theorem1}, we classify \ac{MPC} formulations that satisfy condition 1) as closed-loop optimal, where the learning process that aims at satisfying condition 1) using \ac{MPC} within the actor as closed-loop optimal learning, see~\ref{sec:mpc_as_actor_closed_loop_learning}. When both conditions 1) and 2) are satisfied, the \ac{MPC} is \MDPComplete{} since it captures both the optimal policy and optimal action-value function of an \ac{MDP}. At the top of this theoretical hierarchy lies the \MDPAligned{} property, which demands that the \ac{MPC} model exactly matches the state-transition kernel of the \ac{MDP}. However, this property is rarely satisfied in practice, as \ac{MPC} formulations typically employ deterministic models, while \ac{MDP} state transitions are inherently stochastic.

\subsection{Extensions of MPC-MDP Equivalence}

\cref{theorem1} was originally developed in the context of \ac{ENMPC}. Subsequent research has extended the result to several \ac{MPC} variants, including real-time iteration \ac{MPC}~\cite{zanon_reinforcement_2020}, mixed-integer \ac{MPC}~\cite{gros_reinforcement_2020}, scenario-based \ac{MPC}~\cite{kordabad_reinforcement_2021}, or any model-based policy~\cite{anand_all_2025}.

Instead of using \ac{OCP} formulations that rely on one-step prediction rollouts of the environment, \cite{sawant_bridging_2022} explores the connection between \acp{MDP} and \ac{QP} formulations that arise for tracking problems with linear dynamics or as subproblems in \ac{OCP} solvers. By parameterizing the \ac{QP} directly rather than the prediction model, this approach offers additional flexibility that can enhance the performance of the \ac{MPC} policy. Building on this idea, \cite{sawant_model-free_2023} incorporates formulations from \ac{SPC}~\cite{favoreel_spc_1999}, showing that the past input/output sequences can serve as a surrogate state from which future predictions are built. While this approach effectively handles systems without state estimators, its application is restricted to problems permitting autoregressive linear maps from inputs and outputs to future trajectories.

For systems requiring state estimation, an estimation layer such as \ac{MHE} can be integrated to provide state feedback to the \ac{MPC}. Research by \cite{esfahani_reinforcement_2021, esfahani_learning-based_2023} demonstrated that jointly optimizing the parameters of both the \ac{MHE} and \ac{MPC} leads to significantly better performance compared to optimizing \ac{MPC} parameters alone. The previously discussed challenges of differentiation through optimization layers and computational complexity in the forward pass extend naturally to the \ac{MHE} layer.

The use of \ac{MPC} and safety filters to ensure the safety of learned policies has emerged as a prominent research direction. One common approach projects control actions from trained policies onto a safe set by minimizing the distance to feasible actions that satisfy model dynamics and safety constraints. The authors of~\cite{wabersich_predictive_2021, wabersich_data-driven_2023} exemplify this strategy by introducing a predictive safety filter that post-processes actions from trained \ac{RL} policies. However, this projection step may lead to sub-optimal actions when the agent is not aware of the filter, as discussed in \cite{gros_safe_2020}. An approach to learning safe and stable policies by construction via robust \ac{MPC} within the policy is discussed in \cite{zanon_safe_2021} and further developed in~\cite{gros_learning_2022, kordabad_safe_2022}.

\subsection{Approximating Discounted MDP with Undiscounted MPC}
Research exploring the relationship between \acp{MPC} and \acp{MDP} has focused on identifying conditions under which undiscounted \acp{MPC} can approximate discounted \acp{MDP}. This connection draws from dissipativity theory~\cite{gros_economic_2022}, a key concept in \ac{ENMPC} stability analysis. Using dissipativity arguments, \cite{kordabad_verification_2021} investigates undiscounted Q-Learing of tracking \ac{MPC} schemes that are locally equivalent to dissipative \ac{ENMPC}, focusing on learning storage functions that maintain \ac{ENMPC} dissipativity. The authors of~\cite{zanon_stability-constrained_2022} prove that under weak stability conditions, the optimal policy of undiscounted and discounted \acp{MDP} coincide - enabling stable undiscounted \acp{MPC} to yield the optimal policy of stable discounted \acp{MDP}. Recent advances by~\cite{kordabad_equivalence_2024} extend these results to unknown dynamics, demonstrating that undiscounted \acp{MPC} can capture optimal policies of discounted \acp{MDP}, as discussed around \cref{theorem1}. Complementary work in~\cite{kordabad_bias_2023} introduces methods to align undiscounted \ac{MPC} cost function minimizers with those of discounted \acp{MDP}. These contributions share a common thread: by enforcing parameter updates that yield dissipative \ac{ENMPC} formulations, stable \ac{ENMPC} policies can be designed constructively rather than by using dissipativity solely as a verification criterion.

\begin{table*}
\caption{Literature that discusses theoretical aspects of the connection between \ac{MPC} and \ac{RL}.}
	\centering
	\label{tab:literature_rlmpc_theory}
	\resizebox{\textwidth}{!}{%
		\begin{tabular}{@{}llllllll@{}}
			\toprule
			Ref. & Authors & Year& Contribution\\ 
			\midrule
			\multicolumn{4}{c}{\ac{MPC} as a Model of the MDP}\\
			\midrule
			\cite{gorges_relations_2017} & Görges &2017&Compares explicit \ac{MPC} for linear input and state-constrained systems to function approximation by \ac{NN} in the context of \ac{RL}\\ 
			\cite{zanon_practical_2019} & \etal{Zanon} &2019& \Acf{ENMPC} as function approximation\\ 
			\cite{gros_data-driven_2020} & \etal{Gros} &2020 & Fundamental principles for \ac{ENMPC} to approximate \ac{MDP} and its connection to \ac{RL}\\ 
			\cite{zanon_reinforcement_2020} & \etal{Zanon} &2020 & Extends \cite{gros_data-driven_2020} to Real-Time Iteration \ac{NMPC}\\ 
			\cite{gros_reinforcement_2020} & \etal{Gros} &2020 &  Extends \cite{gros_data-driven_2020} to Mixed-Integer \ac{MPC}\\ 
			\cite{kordabad_reinforcement_2021} & \etal{Kordabad} &2021 &   Extends \cite{gros_data-driven_2020} to scenario-based \ac{MPC}\\ 
			\cite{esfahani_reinforcement_2021} & \etal{Esfahani} &2021 & Use combined parameterized \ac{MHE} and parameterized \ac{MPC} as a function approximation in RL\\ %
			\cite{sawant_bridging_2022} & \etal{Sawant} &2022 & Extends \cite{gros_data-driven_2020} to generic parameterized \ac{QP} layers instead of MPC\\ 
			\cite{sawant_model-free_2023} & \etal{Sawant} &2023 & Extends to \acf{SPC} using input-output data as surrogate state \\ 
			\cite{seel_convex_2022} & \etal{Seel} &2022 & Extends the cost to convex \acf{NN} formulations to facilitate a more complex function approximation\\ 
			\cite{anand_painless_2023} & \etal{Anand} &2023 & Proposes deterministic policy gradient to train \ac{MPC} via a parameter-perturbed variant as the critic.\\ 
			\cite{esfahani_learning-based_2023} & \etal{Esfahani} &2023 & Use combined parameterized \ac{MHE} and parameterized \ac{MPC} (extension of \cite{esfahani_reinforcement_2021})\\ 
			\midrule
			\multicolumn{4}{c}{Post-processing and Safety}\\
			\midrule
			\cite{gros_safe_2020} & \etal{Gros} &2020 & Discusses the impact of safe policy projections, cf. \cite{wabersich_predictive_2021}, on learning and proposes corrections in the context of Q-learning \\ 
			\cite{zanon_safe_2021} & \etal{Zanon} &2021 & Addresses safety in \ac{RL} via robust \ac{MPC} and safe parameter updates\\ 
			\cite{wabersich_predictive_2021} & \etal{Wabersich} &2021 & Address safety via \ac{MPC}-based predictive safety filters that modify unsafe actions from \ac{RL} policies\\ 
			\cite{wabersich_data-driven_2023} & \etal{Wabersich} &2021 & Address safety via \ac{MPC}-based predictive safety filters that modify unsafe actions from \ac{RL} policies\\ 
			\cite{gros_learning_2022} & \etal{Gros} &2022 & Studies safety and stability of \ac{MPC} schemes subject to parameter updates by an \ac{RL} policy\\ 
			\cite{kordabad_safe_2022} & \etal{Kordabad} &2022 & Use Wasserstein Distributionally Robust \ac{MPC} as a tool to generate safe \ac{RL} policies\\ 
			\midrule
			\multicolumn{4}{c}{Approximation Theory}\\
			\midrule
			\cite{kordabad_verification_2021} & \etal{Kordabad} &2021 & Proposes Q-learning to modify the storage function in \ac{ENMPC} such that it meets disspativity criteria.\\
			\cite{kordabad_bias_2021} & \etal{Kordabad} &2021 & Address the problem that MPC-based policies with hard constraints can lead to biased gradient estimates in actor-critic methods.\\
			\cite{gros_economic_2022} & \etal{Gros} &2022 & Extends dissipativity theory of \ac{ENMPC} to a wider class of probabilistic dynamics to connect to undiscounted \acp{MDP}.\\
			\cite{zanon_stability-constrained_2022} & \etal{Zanon} &2022 & Connects discounted and undiscounted \acp{MDP} with known dynamics to generate stability-constrained optimal policies through \ac{MPC}.\\
			\cite{kordabad_bias_2023} & \etal{Kordabad} &2023 & Proposes a cost modification for discounted \ac{ENMPC} such that stage-cost minimizer coincides with the undiscounted case.   \\
			\cite{seel_combining_2023} & \etal{Seel} & 2023 & Proposes a combination of policy gradient and Q-learning to exploit gradient information and capture cost function properties not affecting the policy.\\
			\cite{kordabad_equivalence_2024} & \etal{Kordabad} &2024 & Extends \cite{zanon_stability-constrained_2022} to mismatching dynamics and extends \cite{gros_data-driven_2020} to capture a discounted \ac{MDP} via undiscounted \ac{MPC}.\\
			\bottomrule
		\end{tabular}
        }
\end{table*}

\subsection{State-transition Model for Closed-Loop Optimality}

Traditional \ac{MPC} applications select prediction models by minimizing the identification loss~\eqref{eq:supervised_loss}, along with considerations regarding considerations such as convexity or smoothness. However, models achieving good mean-error predictions or maximum-likelihood models may not necessarily guarantee closed-loop optimality. Recent research by~\cite{anand_akhil_s_data-driven_2024, reinhardt_economic_2024} establishes formal requirements for prediction models in~\eqref{eq:mpc-q-problem} to achieve closed-loop optimality and capture optimal value functions. They demonstrate that for a model $f^\mpcShrt$ to deliver both the optimal policy $\pi^\star$ and optimal action-value function $Q^\star$, it must satisfy:
\begin{equation}
\label{eq:closed_loop_optimal_model}
            \mathbb E_{S^{+}\sim P(\cdot | s,\,a)}\left[V^\star\left( S^+\right)\,|\, s, a\, \right] - V^\star\left( f^\mathrm{MPC}\left( s, a\right)\right) = V_0,
\end{equation}
where $V_0$ is a constant. While \eqref{eq:closed_loop_optimal_model} does not directly yield a system identification procedure, it provides insights into models minimizing~\eqref{eq:supervised_loss}. Such models prove optimal for deterministic \acp{MDP} and \ac{LQR} problems. Local optimality can be established for tracking problems, set-point stabilization, and dissipative economic problems by bounding the conditional covariance of the state transition dynamics~$P$ given in \eqref{eq:general_model} and the curvature of the (locally smooth) optimal value function~$V^\star$. However, models minimizing \eqref{eq:supervised_loss}  may not yield optimal policies for problems with non-smooth cost functions or dynamics, non-dissipative \acp{MDP}, or strong disturbances.

The authors of~\cite{nikishin_control-oriented_2022} approach the problem of model design by formulating an equivalent expression to \eqref{eq:closed_loop_optimal_model} in terms of the Bellman equation,  introducing the \ac{OMD} method for simultaneous learning of Q-functions and models. This connects to broader developments in \ac{RL}, such as the learned hidden state dynamic models in~\cite{schrittwieser_mastering_2020}, which focus on capturing only the state information essential for optimal actions and rewards. While \cite{anand_akhil_s_data-driven_2024, anand_all_2025, reinhardt_economic_2024} focus on the theoretical properties of models in \ac{MPC} contexts, \cite{nikishin_control-oriented_2022} provides algorithmic tools for their construction.

\section{Software Tools and Implementation Aspects}
\label{sec:software}
This section highlights challenges in practical implementations and points to available software solutions for combining \ac{MPC} and \ac{RL}.

\subsection{Integrating Software from Machine Learning and Numerical Optimization}
The architectural combinations of \acp{NN} and \acp{MPC} described in~\secref{sec:comb_arch} and illustrated in~\figref{fig:structure2} significantly influence the choice of optimization software, particularly for derivative-based \ac{MPC} solvers. When objective or constraint functions are highly nonlinear, the resulting optimization problems become nonconvex and computationally challenging. Furthermore, the selected software must be compatible with machine learning libraries to enable seamless integration.

To address these challenges, recent open-source tools have emerged. The package \texttt{l4casadi}~\cite{salzmann2023learning} enables the integration of \acp{NN}~$\nnfun_{\theta}(s,z)$ within the automatic differentiation tool~\casadi{}~\cite{andersson_sensitivity_2018}. Similarly, \texttt{l4acados}~\cite{lahr_l4acados_2024} bridges learned models from \pytorch{} to \acados{}, focusing on residual models based on \acp{GP} in an integrated approach $\phi = \varphi_\theta(z;\,s)$. While these tools address derivative-based optimization, sampling-based \ac{MPC} solvers offer an alternative approach, as they are less sensitive to nonlinear functions and only require forward simulation.

\subsection{Implementation Aspects}
Besides the difficulty of solving an optimization problem that involves \acp{NN}
with derivative-based solvers, several further issues arise when combining MPC
and RL, and should be considered in the implementation. 

\begin{enumerate}
    \item \emph{Differentiation}: The \ac{MPC} solver differentiation requires the computation of sensitivities - a process that is theoretically straightforward but demands existing solvers to provide an efficient implementation of the parametric sensitivity computations. This capability is currently ongoing work in \texttt{acados}~\cite{verschueren_acadosmodular_2022} and will be detailed in a future publication.
    \item \emph{Solver States}: \ac{MPC} solvers often introduce additional solver states that must be incorporated into an augmented \ac{MDP}. Nonlinear optimization processes start from initial guesses and may converge to different local minima, making these initial guesses part of the state space. This becomes particularly relevant in the \ac{RTI} scheme~\cite{diehl_real-time_2002} (cf.~\secref{sec:mpc}), where the solver tracks optima across iterations rather than achieving full convergence. These additional solver states increase the \ac{MDP}'s complexity.
    \item \emph{Computational Efficiency}: Some \ac{RL} algorithms, such as \ac{SAC}, require random sampling from replay buffers during learning. Computing gradients for each sample necessitates multiple solutions to optimization problems, resulting in significant computational overhead.
    \item \emph{Hardware Utilization}: While machine learning frameworks benefit from parallelizable architectures, \ac{MPC} solvers typically target CPU-based embedded applications. Efforts to leverage GPUs or parallel CPU implementations represent an emerging research direction, with recent developments including GPU-based sampling \ac{MPC}~\cite{vlahov_mppi-generic_2024} and multi-core CPU implementations for derivative-based \ac{MPC} in \acados{} \cite{verschueren_acadosmodular_2022}.
\end{enumerate}

Complementary to this survey, early-stage software development aims at addressing these aspects. The code is publicly available as ~\texttt{\ac{leap-c}}\footnote{\url{https://github.com/leap-c/leap-c}}, a tool for derivative-based closed-loop learning utilizing MPC policies with the \acados{} solver as a numerical optimization backend. Similar software projects with a lower degree of integration with fast \ac{OCP} solvers exist\footnote{\url{https://github.com/FilippoAiraldi/mpc-reinforcement-learning}}.

\section{Conclusion and Discussion}
\label{sec:discussion}
\acresetall

\Ac{MPC} and \ac{RL} each require extensive background knowledge, as detailed in \secref{sec:rl} and \secref{sec:mpc}, and possess complementary features as highlighted in \secref{sec:comparison}.
The latter raises the question of how to effectively combine these approaches to leverage their respective strengths while mitigating their limitations.
In \secref{sec:comb_arch}, various possible combinations within the framework of solving \acp{MDP} through actor-critic \ac{RL} are demonstrated.

Within this framework, \ac{MPC} can serve as an algorithmic part in multiple roles. One role is as a supervisory expert actor in \ac{IL} or \ac{GPS} frameworks as discussed in \ref{sec:comb_expert}.
This allows to include domain knowledge and other features to develop computationally efficient policies for online deployment.

Another role for \ac{MPC} is as a part of the policy - the most prevalent approach in the literature as shown in~\secref{sec:mpc_as_actor}.
The categorization becomes more intricate in this case. 
We distinguish between two main directions: \MDPAligned{} approaches that aim at modifying the \ac{MPC} to reflect the \ac{MDP} structure in terms of cost and state transitions, and \closedLoopOpt{} approaches that focus on the closed-loop optimality of \ac{MPC}. In both cases, the \ac{MPC} enters the learning pipeline through an optimization layer that may be expensive to evaluate in the forward inference path or to differentiate in the backpropagation path in order to provide the respective gradients, cf.~\appref{appendix:differentiating}. A challenge concerning the implementation is that the computation of gradients is often not provided by optimization solvers, despite recent implementations in, e.g., \acados{}~\cite{verschueren_acadosmodular_2022} and other software discussed in \secref{sec:software}. It is possible to avoid the gradient computation of the \ac{MPC} solution by integrating it into the environment (cf. \secref{sec:mpc_as_actor_closed_loop_learning}) such that its parameterization becomes part of the action space. The critic then evaluates the parameterization of the MPC as opposed to the actions provided by the MPC output. Considering \ac{MPC} as a part of the environment and evaluating the critic on the MPC parameterization can be viewed as an example of the second practical architecture that we emphasize, i.e., architectures that use the MPC solver within the forward path but do not require its differentiation. Other architectures that avoid differentiation the category of using an expert MPC (Sect.~\ref{sec:comb_expert}), using a fixed MPC critic ( Sect.~\ref{sec:comb_critic}) and on-policy learning of the model and the terminal value function in the \MDPAligned{} structure (Sect.~\ref{sec:mpc_as_actor_aligned}) can be seen as variants that do not require differentiation but use MPC as part of the forward path. Still, the repeated solutions of optimization problems make learning iterations significantly more expensive, and implementation aspects and tailored software as outlined in \secref{sec:software} should be considered.

The last role of MPC in the RL framework is that of a critic, as discussed in \ref{sec:comb_critic}, which is not as prominent in the literature but may have the advantage of building explainable value functions estimates and significantly improving the sample efficiency and avoidance of local minima of \ac{RL} training. Again, repeated \ac{MPC} evaluations increase the computational cost compared to other \ac{NN} layers.

Beyond these categories, several variants exist in the literature where \ac{MPC} is not used during \ac{RL} training but where \ac{MPC} is solely used in the final deployed policy. These include off-policy learning of the model and some algorithms that learn a terminal value function in the \MDPAligned{} structure ( Sect.~\ref{sec:mpc_as_actor_aligned}) as well as pre- and post-processing components (\secref{sec:prepostprocessing}). In the latter, \ac{MPC} can function as a safety filter to enhance trained policies with its inherent constraint satisfaction capabilities~\cite{hewing_cautious_2020} and as a reference provider (~\ref{sec:prepostprocessing}) - a role that may extend to the learning phase as well.

Yet another important design decision is the combination architecture of \acp{NN} and \acp{MPC} in \secref{sec:comb_arch} for which we identified the integrated, hierarchical, parallel or parameterized approaches outlined in \figref{fig:structure2}. \Acp{NN} are usually highly nonlinear functions that deserve a distinction from moderately nonlinear functions used as part of derivative-based MPC. In Sect.~\ref{sec:comb_arch}, we contrast between using \acp{NN} outside the MPC, i.e., it does not depend on decision variables of the optimization problem, using \acp{NN} inside the MPC optimization problem or avoiding \acp{NN} and using parameters of simpler linear or quadratic functions.

Many variants of classifying \ac{MPC} and \ac{RL} methods exist in the
literature, with notable frameworks presented
in~\cite{bertsekas_reinforcement_2019,gorges_relations_2017} and
\cite{mesbah_fusion_2022}. We used an actor-critic perspective and categories
within that setting, that allows us to give a broad classification of very diverse algorithms.


As indicated in this survey, the expectations for the superiority of MPC and RL
combinations are high, but challenges remain. Theoretical work provided already
a solid foundation, cf. Sect.~\ref{sec:theory}. Nonetheless, practical work that
shows large improvements over various domains is still required. Most
importantly, the practical work requires numerical solvers to align with machine
learning frameworks and tools that allow fast MPC solvers to be embedded in
learning frameworks, see~\cref{sec:software}. With the increased development of
such tools, the outlook on the MPC and RL combination appears prosperous. In future
work, we hope the community will investigate whether the control paradigm shifts
from \MDPAligned{} learning to \closedLoopOpt{} learning and whether the
sampling complexity and other challenges of \ac{RL} can be improved by MPC
optimization modules.

\bibliographystyle{IEEEtran}
\bibliography{lib_zotero_sync}

\appendices

\section{Differentiating Through a Parameterized MPC Layer}
\label{appendix:differentiating}

In the following, we discuss how to compute gradients concerning parameters~$\theta$ within the architectures introduced in Fig.~\ref{fig:structure2}.
The availability of gradient information is a crucial prerequisite for the learning approaches introduced in the following sections.
We restrict the discussion to the parameterized MPC formulation of Fig.~\ref{fig:structure2}, where the corresponding solution map is defined by
\begin{equation}\label{eq:z_star_def}
	z^\star(s, \theta) \coloneqq \argmin_z ~L(z; \theta) \text{ s.t. } g(z; s, \theta) \geq 0.
\end{equation}
Assuming that the objective is twice, the constraints once, continuously differentiable, the \ac{IFT} guarantees the existence of the solution sensitivities~$\frac{\partial z^\star}{\partial \theta}(\bar s, \bar\theta)$ at $(\bar s, \bar \theta)$ if the linear independence constraint qualification, second-order sufficient conditions, and strict complementarity are satisfied at the solution~$z^\star(\bar \theta)$.
Under these conditions, the \ac{IFT} furthermore implies that the solution sensitivity can be computed as
\begin{align}
	\!
	\begin{bmatrix}
		\dfrac{\partial z^\star}{\partial \theta} (\bar s, \bar \theta)\\[6pt]
		\dfrac{\partial \mu_{\mathcal{A}}^\star}{\partial \theta} (\bar s, \bar \theta)
	\end{bmatrix}
	\!= -
	\left(\frac{\partial r}{\partial (z, \mu_{\mathcal{A}})}\right)^{-1} \frac{\partial r}{\partial \theta}
\end{align}
where~$\mu_\mathcal{A}$ are the dual variables associated with the active constraints at the solution~$z^\star(\bar s, \bar \theta)$ and where the partial derivatives of the residual map~$r(z, \mu_\mathcal{A}; s, \theta)$ are evaluated at $z = z^\star(\bar \theta, \bar s), \mu_{\mathcal{A}} = \mu_{\mathcal{A}}^\star(\bar s, \bar \theta), s = \bar s$, and $\theta = \bar \theta$.
The residual map is defined as
\begin{align} \label{eq:residuals}
	r(z, \mu_\mathcal{A}; s, \theta) =
	\begin{bmatrix}
		\grad{z} L(z; \theta) + \grad{z} g_\mathcal{A}(z; s, \theta) \mu_\mathcal{A} \\
		g_\mathcal{A}(z, \mu_\mathcal{A}; s, \theta)	
	\end{bmatrix}
\end{align}
where~$g_\mathcal{A}(z; s, \theta)$ denotes all equalities as well as all inequalities that are active at~$z^\star(\bar s, \bar \theta)$.

At an active set changes, the solution map is not differentiable.
Within the stochastic optimization framework, this nondifferentiability of the solution map for some values of~$\theta$ and~$s$ is, in general, not problematic, as the gradient needs to be well-defined only almost everywhere.

A standard approximation replaces the residual map in~\eqref{eq:residuals} with a smoothed approximation of the KKT conditions associated with~\eqref{eq:z_star_def}, as is done within an interior point solver which leads to a natural smoothing of the solution map alleviating the problem of nondifferentiability at active set changes \cite{gros_data-driven_2020}.

For the interested reader, we provide some further references:
In \cite{andersson_sensitivity_2018}, the authors discuss differentiating through a parameterized \ac{MPC} for active-set methods, whereas in \cite{amos_differentiable_2018}, interior point methods for \ac{MPC} with box constraints on the actions are considered.
The factorization of the \ac{KKT} system required for solving the parameterized \ac{MPC} formulation can be reused to also derive the sensitivities for the parameters if an exact Hessian is used.

\section{Gradients in the Hierarchical Architecture}
\label{appendix:gradients_hierarchical}
The following discussion is restricted to the \hierarchical{} \ac{NN}-\ac{MPC} architecture of Fig~\ref{fig:structure2}.
At each decision step, the \ac{NN} predicts the parameter~$\phi$, which is then processed by the \ac{MPC}-optimization layer to generate a control~$u_0^\star(\phi)$.
We distinguish the learning approaches based on the actor's feedback.
The value function, i.e., the critic, is evaluated for the predicted parameter~$\phi$ or the resulting control~$u_0^\star(\phi)$.
Note that this distinction is first and foremost important for the optimization properties for the \ac{NN} component or when safe exploration is required during training.\\

\subsubsection{Parameter Critic}
Assuming that the critic gives feedback based on the predicted parameter~$\phi$, the optimization layer can be considered part of the environment \cite{reiter_hierarchical_2023}.
Thus, an altered version of the original \ac{MDP} can be defined by introducing a modified MDP transition model 
\begin{equation}\label{eq:transition parameter}
	P^\mathrm{MPC}(s^+ |s, a) \coloneqq P(s^+ | s, u^\star_0 (a)).
\end{equation}
One of the benefits of this approach is that implementations of \ac{RL} methods can be directly used without any further adjustments, as the MPC optimization is only done when generating new samples by forward simulating the system from a given state to the next using~\eqref{eq:transition parameter}.\\

\subsubsection{Control Critic}
Assuming that the critic gives feedback based on the control of the \ac{MPC} optimization layer~$u_0^\star$, there are multiple ways to obtain a gradient.
One can avoid differentiating through the \ac{MPC} optimization layer
\begin{equation*}
	\nabla_\theta J(\theta) = \E_{\psi \sim \pi_\theta(s),\; s \sim \mathcal{D}} \left[ Q(s, u^\star_0 (\psi))\; \nabla_\theta \log \pi_\theta (\psi) \right],
\end{equation*}
by extending the stochastic policy gradient \cite{gros_reinforcement_2021}.
As there is no gradient information, using such a policy gradient can be seen as black box optimization with respect to the \ac{MPC} optimization layer.
This can be especially problematic for high-dimensional parameter spaces, as discussed in \cite{silver_deterministic_2014}.

Alternatively, one can differentiate through the \ac{MPC} \cite{gros_data-driven_2020} extending the deterministic policy gradient~$\eqref{eq:deterministic_policy_gradient}$ to derive a policy gradient with
\begin{equation*}
	\nabla_\theta J^\pi(\theta)  =  \E_{s \sim \mathcal{D}} \left[ \left. \nabla_\theta \mu_\theta (s) \nabla_\psi Q^\mu (s, u^\star_0(\psi)) \right|_{\psi=\mu_\theta(s)} \right].
\end{equation*}
Doing the same for a stochastic policy, using the reparameterization trick, is an extension that has not been considered yet in the literature to the author's best knowledge.

\end{document}

%% file: wordings.tex
\usepackage{acronym}

\newacro{DM}{decision making}
\newacro{AD}{automated driving}
\newacro{SOC}{discrete-time stochastic optimal control}

\newacro{NLP}{nonlinear program}
\newacro{ODE}{ordinary differential equation}
\newacro{MPC}{model predictive control}
\newacro{MHE}{moving horizon estimation}
\newacro{NMPC}{nonlinear model predictive control}
\newacro{ENMPC}{economic nonlinear model predictive control}
\newacro{MPPI}{model predictive path integral control}
\newacro{SPC}{subspace predictive control}
\newacro{QP}{quadratic program}
\newacro{MIQP}{mixed-integer {\ac{QP}}}
\newacro{MILP}{mixed-integer \ac{LP}}
\newacro{MINLP}{mixed-integer \ac{NLP}}
\newacro{MI}{mixed-integer}
\newacro{MIP}{mixed-integer programming}
\newacro{MPCC}{mathematical program with complementarity constraints}
\acrodefplural{MPCC}{mathematical programs with complementarity constraints}
\newacro{BB}{branch-and-bound}
\newacro{SQP}{sequential quadratic programming}
\newacro{RNN}{recurrent neural network}
\newacro{OCP}{optimal control problem}
\newacro{LQR}{linear quadratic regulator}
\newacro{iLQR}{iterative linear quadratic regulator}
\newacro{SLQ}{sequential linear quadratic programming}
\newacro{DDP}{differential dynamic programming}
\newacro{MCP}{mixed complementarity problem}
\newacro{IP}{interior point}
\newacro{ADMM}{alternating direction method of multipliers}
\newacro{RTI}{real time iteration}
\newacro{POMDP}{partially observable \ac{MDP}}
\newacro{CEM}{cross-entropy method}
\newacro{GGN}{generalized Gauss-Newton}
\newacro{KKT}{Karush-Kuhn-Tucker}
\newacro{IFT}{implicit function theorem}
\newacro{RL}{reinforcement learning}
\newacro{DP}{dynamic programming}
\newacro{LSTM}{long short-term memory}
\newacro{NN}{artificial neural network}
\newacro{PPO}{proximal policy optimization}
\newacro{PG}{policy gradient}
\newacro{DPG}{deterministic \ac{PG}}
\newacro{SPG}{stochastic \ac{PG}}
\newacro{TD}{temporal difference}
\newacro{SAC}{soft actor-critic}
\newacro{PI}{policy iteration}
\newacro{SARSA}{state action reward state action}
\newacro{IL}{imitation learning}
\newacro{MLE}{maximum likelihood estimation}
\newacro{MDP}{Markov decision process}
\newacroplural{MDP}[MDPs]{Markov decision processes}
\newacro{DDPG}{deep deterministic \ac{PG}}
\newacro{DQN}{deep Q-networks}
\newacro{TD3}{twin-delayed actor-critic}
\newacro{VI}{value iteration}
\newacro{TRPO}{trust region policy optimization}
\newacro{BC}{behavior cloning}
\newacro{GAN}{generative adversarial network}
\newacro{GPS}{guided policy search}
\newacro{GP}{Gaussian process}
\newacro{leap-c}{Learning for Predictive Control}
\newacro{FA}{function approximator}
\newacro{ADP}{approximate \ac{DP}}
\newacro{LMPC}{linear \ac{MPC}}

\newacro{FF}{feed forward network}
\newacro{DNN}{deep neural network}
\newacro{OMD}{optimal model design}

\newacro{UGV}{unmanned ground vehicle}

\newacro{RWE}{real-world experiments}

\newcommand{\dpTxt}{\ac{DP}}
\newcommand{\mppi}{\ac{MPPI}}
\newcommand{\ilqr}{\ac{iLQR}}

\newcommand{\MDPAligned}{\ac{MDP} aligned}
\newcommand{\closedLoopOpt}{closed-loop optimal}
\newcommand{\MDPComplete}{\ac{MDP} complete}

\newcommand{\hierarchical}{hierarchical}
\newcommand{\integrated}{integrated}
\newcommand{\paramterized}{parameterized}
\newcommand{\termvalfun}{term. val. fun.}

\newcommand{\etal}[1]{{#1} et al.}

\newcommand{\ipopt}{\texttt{IPOPT}}
\newcommand{\acados}{\texttt{acados}}
\newcommand{\casadi}{\texttt{CasADi}}

\newcommand{\gurobi}{\texttt{Gurobi}}

\newcommand{\pytorch}{\texttt{PyTorch}}

\newcommand{\na}{N/A} 

%% file: math_symbols.tex
\newcommand{\mpc}{\mathrm{MPC}}

\newcommand{\R}{\mathbb{R}}

\newcommand{\grad}[1]{\nabla_{\!#1}}

\newcommand{\argmin}{\mathop{\mathrm{argmin}}\limits}
\newcommand{\STAC}{{\mathcal{S}\times \mathcal{A}}}
\newcommand{\ST}{\mathcal{S}}
\newcommand{\AC}{\mathcal{A}}
\newcommand{\mpcShrt}{\textsc{mpc}}
\DeclareMathOperator*{\E}{\mathbb{E}}

\newcommand{\nnfun}{\varphi}